\documentclass[a4paper,USenglish,cleveref,autoref,thm-restate,numberwithinsect]{lipics-v2021}

\usepackage{amsmath,amssymb,amsthm}
\usepackage{thmtools}
\usepackage{url}
\usepackage{subcaption}

\usepackage{todonotes}

\usepackage[draft]{fixme}
\fxsetup{mode=multiuser,theme=color, layout=inline}
\FXRegisterAuthor{tk}{atk}{\color{blue} {\bf T}}
\FXRegisterAuthor{ai}{aai}{\color{red} {\bf GPT}}
\FXRegisterAuthor{rv}{arv}{\color{purple} {\bf Reviewer}}

\usepackage{verbatim}
\usepackage{tikz}
\usetikzlibrary{positioning,arrows,arrows.meta,fit,decorations.pathreplacing}
\tikzset{nosep/.style={inner sep=0, outer sep=0}}

\newtheorem{fact}[theorem]{Fact}
\def\restatespacefix{
\vspace{-1.9ex}
}
\newenvironment{proofsketch}{\begin{proof}[Proof sketch]}{\end{proof}}

\newcommand{\floor}[1]{\lfloor#1\rfloor}
\newcommand{\angles}[1]{\langle#1\rangle}

\newcommand{\ceil}[1]{\lceil#1\rceil}

\newcommand{\absolute}[1]{\lvert#1\rvert}

\newcommand{\aabsolute}[1]{\|#1\|}
\newcommand{\absoluteenc}[1]{\absolute{\senc{#1}}}

\newcommand{\opstyle}[1]{\textnormal{\textsf{{#1}}}}

\def\Oh{{\mathcal O}}
\def\tOh{{\tilde{\Oh}}}

\def\eps{{\varepsilon}}
\newcommand{\Z}{\mathbb{Z}}
\newcommand{\Zz}{\Z_{\ge 0}}
\newcommand{\Zp}{\Z_{+}}

\def\poly{\textnormal{poly}}
\def\polylog{\textnormal{polylog}}
\def\emptystring{{\epsilon}}

\def\B{\textnormal{\textbf{B}}}

\def\tB{\textnormal{\textbf{C}}}

\def\len{{\opstyle{len}}}

\def\int{{\opstyle{int}}}
\def\dd{\mathop{.\,.}}
\newcommand{\senc}[1]{\mathsf{sp\smash{\underline{\phantom{o}}}enc}(#1)}

\def\rank{\textup{\textsf{rank}}}
\def\select{\textup{\textsf{select}}}
\def\pred{\textup{\textsf{pred}}}
\def\zip{\textnormal{\textsf{zip}}}

\newcommand{\Sync}{\mathsf{Sync}}

\let\eqmatch\cong
\let\neqmatch\ncong

\title{Time-Optimal Construction\texorpdfstring{\\}{\ }of String Synchronizing Sets}
\titlerunning{Time-Optimal Construction of String Synchronizing Sets}

\author{Jonas Ellert}{DIENS, \'{E}cole normale sup\'{e}rieure de Paris, PSL Research University, France}{ellert.jonas@gmail.com}{https://orcid.org/0000-0003-3305-6185}{Partially funded by grant ANR20-CE48-0001 from the French National Research Agency.}

\author{Tomasz Kociumaka}{Max Planck Institute for Informatics, Saarland Informatics Campus, Saarbrücken, Germany}{tomasz.kociumaka@mpi-inf.mpg.de}{https://orcid.org/0000-0002-2477-1702}{}

\authorrunning{J. Ellert and T. Kociumaka} 

\Copyright{Jonas Ellert and Tomasz Kociumaka} 



\category{} 

\hideLIPIcs
\nolinenumbers

\EventEditors{Meena Mahajan, Florin Manea, Annabelle McIver, and Nguy\~{\^{e}}n Kim Th\'{\u{a}}ng}
\EventNoEds{4}
\EventLongTitle{43rd International Symposium on Theoretical Aspects of Computer Science (STACS 2026)}
\EventShortTitle{STACS 2026}
\EventAcronym{STACS}
\EventYear{2026}
\EventDate{March 9--September 13, 2026}
\EventLocation{Grenoble, France}
\EventLogo{}
\SeriesVolume{364}
\ArticleNo{53}

\begin{document}

\maketitle

\begin{abstract}
    A powerful design principle behind many modern string algorithms is \emph{local consistency}: breaking the symmetry between string positions based on their short contexts so that matching fragments are handled consistently. 
Among the most influential instantiations of this principle are \emph{string synchronizing sets} [Kempa \& Kociumaka; STOC 2019]. 
A $\tau$-synchronizing set of a string of length~$n$ is a set of $O(n/\tau)$ string positions, chosen using their length-$2\tau$ contexts, such that (outside of highly periodic regions) every block of $\tau$ consecutive positions contains at least one element of the set.
Synchronizing sets have found dozens of applications in diverse settings, from quantum and dynamic algorithms to fully compressed computation.
In the classic word RAM model, particularly for strings over small alphabets, they enabled faster solutions to core problems in data compression, text indexing, and string similarity.

In this work, we show that any string $T \in [0 \dd \sigma)^n$ can be preprocessed in $O(n \log \sigma / \log n)$ time so that, for any given integer $\tau\in [1\dd n]$, a $\tau$-synchronizing set of $T$ can be constructed in $O((n \log \tau)/(\tau \log n))$ time.
Both bounds are optimal in the word RAM model with machine word size $w=\Theta(\log n)$, matching the information-theoretic minimum for the input and output sizes, respectively.
Previously, constructing a $\tau$-synchronizing set required $O(n/\tau)$ time after an $O(n)$-time preprocessing [Kociumaka, Radoszewski, Rytter, and Waleń; SICOMP 2024], or, in the restricted regime of $\tau < 0.2 \log_\sigma n$, without any preprocessing needed [Kempa \& Kociumaka; STOC 2019].

A simple instantiation of our method outputs the synchronizing set as a sorted list in $O(n/\tau)$ time, or as a bitmask in $O(n/\log n)$ time.
Our optimal construction produces a compact fully indexable dictionary, supporting select queries in $O(1)$ time and rank queries in $O\!\left(\log\!\left(\tfrac{\log \tau}{\log \log n}\right)\right)$ time.
The latter complexity matches known unconditional cell-probe lower bounds for $\tau \le n^{1-\Omega(1)}$.

To achieve this, we introduce a general framework for efficiently processing sparse integer sequences via a custom variable-length encoding. We also augment the optimal variant of van Emde Boas trees [Pătraşcu \& Thorup; STOC 2006] with a deterministic linear-time construction. When the set is represented as a bitmask under our sparse encoding, the same guarantees for select and rank queries hold after preprocessing in time proportional to the size of our encoding (in words).
\end{abstract}

\section{Introduction}

In many string processing tasks, one can afford a relatively costly preprocessing for a small fraction of positions in the input string.
The choice of these positions often governs how useful the preprocessing is.
The most natural selection mechanisms rely only on the length $n$ of the input string and the ``selection rate''~$\frac1\tau$; this includes selecting multiples of $\tau$, sampling positions uniformly at random, as well as difference covers~\cite{Mae85}.
Although these simple mechanisms are sufficient in some scenarios, the principle of a \emph{locally consistent} selection has enabled a myriad of new applications.
Here, the idea is to select a position $i$ based on the symbols at the nearby positions, without looking at the integer $i$ itself. 
This ensures that any two positions are handled in the same way whenever they appear in the same context.

The concept of local consistency can be traced back to Sahinalp and Vishkin's \emph{locally consistent parsing}~\cite{SV94a,SV94b,SV96}.
Further hierarchical parsing mechanisms appeared in~\cite{CM05,CM07,DBLP:journals/tcs/Jez15,DBLP:journals/jacm/Jez16,MSU97}.
Their drawback is that the context size is typically bounded in terms of the number of fixed-level phrases of the parsing, which can vary between regions of the string.

Local consistency is also frequently used in applied research, especially bioinformatics, where minimizers~\cite{10.1093/bioinformatics/bth408,DBLP:conf/sigmod/SchleimerWA03} are the most popular mechanism. 
Here, the context size is easily controlled, but the number of selected positions is known to be $\Oh(n/\tau)$ only for~random~strings.

String synchronizing sets~\cite{KK19} are a relatively recent local consistency mechanism that addresses the two issues with previous alternatives, paving the way for many more applications (see a dedicated paragraph near the end of this section).
A \mbox{$\tau$-synchronizing} set for a length-$n$ string $T$ is a set $\Sync\subseteq [0\dd n - 2\tau]$ satisfying two conditions:
The \emph{consistency} condition states that, if two positions $i,j$ share their size-$\Oh(\tau)$ contexts (formalized as $T[i\dd i+2\tau)=T[j\dd j+2\tau)$),
then either both positions are synchronizing ($i,j\in \Sync$) or neither of them is ($i,j\notin \Sync$).
The \emph{density} condition requires that, among every $\tau$ consecutive positions in $T$, at least one belongs to $\Sync$, except in highly periodic regions of $T$ (see \cref{def:sync} for a formal definition).
The central contribution of Kempa and Kociumaka~\cite{KK19} is that every length-$n$ string $T$ admits, for every positive integer $\tau<\frac12 n$, a $\tau$-synchronizing set of size $\Oh(n/\tau)$ that can be constructed deterministically in $\Oh(n)$ time.
This construction uses sliding window minima, similarly to how minimizers are defined.

The original motivation for string synchronizing sets was to efficiently process strings over small alphabets of size $\sigma = n^{o(1)}$.
In this regime, one can store $\Theta(\log_{\sigma} n)= \omega(1)$ symbols in a single machine word (integer variable) of $\Theta(\log n)$ bits,
which allows solving problems in $o(n)$ time, e.g., in $\Oh(n / \log_\sigma n)$ time, proportional to the size of the input in machine words.
The main algorithms of~\cite{KK19} utilize $\tau$-synchronizing sets for a small $\tau = \Theta(\log_{\sigma} n)$, which are built in $\Oh(n/\log_\sigma n)$ time.
More generally, \cite{KK19} provides an $\Oh(n/\tau)$-time $\tau$-synchronizing set construction for every $\tau < \frac15 \log_{\sigma} n$ by adapting the $\Oh(n)$-time construction for arbitrary $\tau$.

Several further works~\cite{ACIKRRSWZ19,DBLP:conf/esa/Charalampopoulos21,DBLP:journals/siamcomp/KociumakaRRW24} rely on a hierarchy of $\tau$-synchronizing sets constructed for different values of $\tau$.
Kociumaka, Radoszewski, Rytter, and Waleń~\cite{DBLP:journals/siamcomp/KociumakaRRW24} build such a hierarchy using a novel construction algorithm based on \emph{restricted recompression}, a locally consistent parsing scheme that modifies recompression by Jeż~\cite{DBLP:journals/tcs/Jez15,DBLP:journals/jacm/Jez16}. 
After an $\Oh(n)$-time preprocessing, their algorithm constructs a $\tau$-synchronizing set (given any $\tau$) in $\Oh(n/\tau)$ time.

While the two existing constructions have already been very impactful, several natural questions remain open, limiting further use cases.
The following is perhaps the simplest one:
\begin{quote}\itshape
    Can one construct a small $\tau$-synchronizing set in $o(n)$ time if $\tau \ge \log_\sigma n$ and $\sigma=n^{o(1)}$?
\end{quote}
In general, we can ask for a single construction algorithm subsuming the two current methods:
\begin{quote}\itshape
    Is there an algorithm that, after  $\Oh(n/\log_\sigma n)$-time preprocessing of a given string, allows constructing a $\tau$-synchronizing set in $\Oh(n/\tau)$ time for any given $\tau < \frac12n$?
\end{quote}
The first main result of this work, presented in \cref{sec:recompression,sec:runs,sec:sync}, answers both questions positively:

\begin{theorem}[Simplified version of \cref{thm:ss-explicit}]\label{thm:ss-explicit-simple}
    A string $T\in [0\dd \sigma)^n$ can be preprocessed in $\Oh(n/\log_\sigma n)$ time so that, given $\tau \le \frac12n$, a $\tau$-synchronizing set $\Sync$ of $T$ of size ${|\Sync| < \frac{70n}{\tau}}$ can be constructed in $\Oh(\frac{n}{\tau})$ time.
\end{theorem}

Our algorithm returns $\Sync$ as a sorted list, but the same techniques allow returning $\Sync$ as a bitmask in $\Oh(n/\log_\sigma n)$ time, which is faster when $\tau = o(\log_\sigma n)$; see \cref{thm:ss-bitmask}.

While these two representations are \emph{compact} (asymptotically optimal up to a constant factor) for $\Oh(n/\tau)$-size subsets when $\tau = n^{\Omega(1)}$ and $\tau = \Oh(1)$, respectively, we can hope for a smaller representation and a faster construction algorithm for intermediate values of~$\tau$.
We address this in \cref{sec:faster}, where we provide an $\Oh(\frac{n \log \tau}{\tau \log n})$-time algorithm that outputs a representation supporting efficient select and rank queries (asking for the $r$-th smallest synchronizing position and the number of synchronizing positions smaller than~$i$,~respectively).

\begin{theorem}[Simplified version of \cref{thm:sss_sparse_with_support}]\label{thm:sss_sparse_with_support_simple}
  A string $T\in [0\dd \sigma)^n$ can be preprocessed in $\Oh(n/\log_\sigma n)$ time so that, given $\tau \le \frac12n$, a $\tau$-synchronizing set $\Sync$ of $T$ of size $|\Sync| < \frac{70n}{\tau}$ can be constructed in $\Oh(\frac{n \log \tau}{\tau \log n})$ time.
  The set is reported in an $\Oh(\frac{n \log \tau}{\tau})$-bit representation that supports select queries in $\Oh(1)$ time and rank queries in $\Oh(\log\frac{\log \tau}{\log \log n})$~time.
\end{theorem}

As discussed in \cref{lem:opt}, our representation size is asymptotically optimal. 
Hence, the preprocessing time and the construction time are also optimal for the machine word size of $\Theta(\log n)$ bits.
The time complexity of rank queries reduces to constant time for $\tau = \log^{\Oh(1)} n$ and matches the unconditional lower bound of Pătraşcu and Thorup~\cite{DBLP:conf/stoc/PatrascuT06} for $\tau \le n^{1-\Omega(1)}$.

\subparagraph*{Our Techniques.}
The algorithm behind \cref{thm:ss-explicit-simple} builds upon the construction of \cite{DBLP:journals/siamcomp/KociumakaRRW24}.
For this, in \cref{sec:recompression}, we show how to implement restricted recompression in $\Oh(n/\log_{\sigma} n)$ time. 
During the first $K=\Theta(\log \log_{\sigma} n)$ rounds, we simulate restricted recompression on every possible length-$\Oh(\log_{\sigma} n)$ context and keep track of how many times each context occurs in the input string.
In the remaining rounds, we process each phrase in constant time.
What enables such an approach is that the context size in the $k$-th round of restricted recompression is bounded by $\Oh(\lambda_k)$ and the number of phrases is $\Oh(n/\lambda_k)$, where $\lambda_k=2^{\Theta(k)}$.

Before we can derive \cref{thm:ss-explicit-simple} in \cref{sec:sync}, we also need a data structure for reporting runs (maximal periodic fragments) in the input string, which we present in \cref{sec:runs}.

Upon a transition from \cref{thm:ss-explicit-simple} to \cref{thm:sss_sparse_with_support_simple}, several seemingly simple steps become difficult to implement.
The prevailing challenge is to process sparse integer sequences in $o(1)$ time per non-zero entry.
Our strategy is to introduce a variable-length encoding of such sequences, implement operations on them using \emph{transducers} (finite automata with output tapes), and provide a general method for speeding up transducer execution using precomputed tables.
We believe that our approach, outlined in \cref{sec:faster}, will be useful in many other contexts to efficiently manipulate compact data representations.

In particular, in order to provide rank and select support in \cref{thm:sss_sparse_with_support_simple}, we provide a new compact variant of van Emde Boas trees~\cite{DBLP:conf/focs/Boas75} matching the optimal query time bounds of Pătraşcu and Thorup~\cite{DBLP:conf/stoc/PatrascuT06} and, unlike existing variants, constructible deterministically in linear time.

\begin{restatable}{theorem}{restatevEB}\label{lem:vEB_final_space_and_time}
    Let $S \subseteq [0\dd 2^\ell)$ of size $\absolute{S} = n$ with $\ell \geq 2$ and $n, 2^\ell \in 2^{\Oh(w)}$ be given as an array of $\ell$-bit integers in increasing order. For $m \geq n$, let $a = \lg(m / n) + \lg w$.
    A~deterministic data structure that answers rank and predecessor queries in $\Oh(\lg \frac{\ell - \lg m}{a})$ time can be built in $\Oh(m)$ time and words of space.
\end{restatable}

If the input set is provided using our sparse representation, we even achieve construction time proportional to the number of machine words in the compact encoding of the set (and, in many cases, sublinear in the number of set elements); see \cref{sec:rankselect}.

\subparagraph*{Applications of Synchronizing Sets.}
Since their introduction in 2019, string synchronizing sets have found numerous applications across a variety of settings.
For strings over small alphabets, they enabled $o(n)$-time algorithms for fundamental decades-old tasks such as Burrows--Wheeler transform~\cite{KK19}, Lempel--Ziv factorization~\cite{DBLP:conf/spire/Ellert23,DBLP:conf/focs/KempaK24}, and the longest common factor problem~\cite{DBLP:conf/esa/Charalampopoulos21}.
They are also behind the only $o(n)$-time constructions of compact suffix array and suffix tree representations~\cite{DBLP:conf/soda/KempaK23}
and data structures for longest common extension (LCE)~\cite{KK19,DFHKK20} and internal pattern matching (IPM)~\cite{DBLP:journals/siamcomp/KociumakaRRW24} queries.
Further applications include detecting regularities in strings, such as palindromes~\cite{DBLP:conf/cpm/Charalampopoulos22}, squares~\cite{DBLP:conf/mfcs/Charalampopoulos25}, and covers~\cite{DBLP:conf/spire/RadoszewskiZ24}.

In some applications, including IPM queries~\cite{DBLP:journals/siamcomp/KociumakaRRW24} and the longest common factor with mismatches problem~\cite{DBLP:conf/esa/Charalampopoulos21}, string synchronizing sets also enabled speed-ups for large alphabets.
In certain cases, such as for the longest common circular factor problem~\cite{ACIKRRSWZ19}, the obtained speed-up is as large as from $\Oh(n^2)$ to $\Oh(n \log^{\Oh(1)} n)$. 
Polynomial-factor speed-ups also arise beyond the classic setting, thanks to adaptations of synchronizing sets (with custom constructions) to the dynamic~\cite{DBLP:conf/stoc/KempaK22}, quantum~\cite{DBLP:journals/algorithmica/AkmalJ23,DBLP:journals/talg/JinN24}, and fully-compressed~\cite{DBLP:journals/cacm/KempaK22,DBLP:conf/focs/KempaK24} settings.

\subparagraph*{Other Local Consistency Mechanisms.}
There are dozens of local consistency mechanisms with different features and use cases.
Beyond those mentioned above, theoretical methods include sample assignments~\cite{KRRW15} and the partitioning sets~\cite{BGP20}, which predate synchronizing sets.
In the algorithm engineering community, beyond minimizers (see \cite[Chapter 6]{GK25} for a very recent literature overview), notable alternatives include locally consistent anchors~\cite{Loukides2021,Ayad2025} and the prefix-free parsing~\cite{DBLP:journals/almob/BoucherGKLMM19}, which is particularly useful in the context of suffix sorting.

\section{Preliminaries}

For $i,j \in \Z$, we write $[i\dd j] = [i\dd j + 1) = (i-1\dd j] = (i-1\dd j + 1)$ to denote $\{ h \in \Z \mid i \leq h \leq j \}$. 
We also write $\lg i$ to denote $\log_2 \max(1,i)$.
A string $T$ of length $\absolute{T} = n$ over alphabet $[0 \dd \sigma)$ is a sequence of $n$ symbols from $[0\dd \sigma)$.
For $i, j \in [0 \dd n)$, the $i$-th symbol is $T[i]$, and the sequence $T[i]T[i + 1] \cdots T[j]$ is denoted by $T[i \dd j] = T[i \dd j + 1)$ (which is the empty string if $i > j$).
For $i', j' \in [0 \dd n)$, consider $T[i'\dd j']$. We may interpret $T[i \dd j]$ and $T[i'\dd j']$ as \emph{substrings}. 
We then write $T[i \dd j] \eqmatch T[i'\dd j']$ and say that the substrings match if and only if $j - i = j' - i'$ and $\forall_{h \in [0\dd j - i]}\; T[i + h] = T[i' + h]$. 
We can choose to interpret them as \emph{fragments} instead, in which case we write $T[i \dd j] = T[i'\dd j']$ and say that the fragments are equal if and only if $i = i'$ and $j = j'$ (or if both fragments are empty).
For fragments $T[i \dd j]$ and $T[i'\dd j']$, we denote their intersection as $T[i \dd j] \cap T[i'\dd j'] = T[\max(i, i')\dd \min(j, j')]$.

The concatenation of two strings $T[0\dd n)$ and $S[0\dd m)$ is defined as $T \cdot S := T[0] \cdots T[n - 1]\allowbreak S[0]\cdots S[m - 1]$, and the $k$-times concatenation of $T$ with itself is written as $T^k$, where $k \in \Zz$. Note that $T^0$ is the empty string.
The unique primitive root $R$ of a non-empty string $T$ is the shortest prefix of $T$ such that there is $k \in \Zp$ with $T = R^k$.

\subparagraph{Model of Computation.} We assume the word RAM model (see, e.g., \cite{DBLP:conf/stacs/Hagerup98}) with words of $w = \Theta(\lg n)$ bits when processing a text $T \in [0\dd \sigma)^n$ with $\sigma \in n^{\Oh(1)}$. The text and all other strings encountered throughout the paper are assumed to be in \emph{$\ceil{\lg \sigma}$-bit representation}, i.e., each symbol uses exactly $\lceil \lg \sigma\rceil$ bits, the entire text occupies $\Theta(n \lg \sigma)$ consecutive bits of memory, and each memory word fits $\Theta(\log_\sigma n)$ symbols. Using arithmetic and bitwise operations, we can extract any substring of length up to $\floor{\log_\sigma n}$ in a single word~in~$\Oh(1)$~time.

\subparagraph{Adding sentinel symbols to the text.}
In the analysis, $\sigma$ will only appear when the claimed (pre-)processing time is $\Oh(n / \log_\sigma n)$. For obtaining these results, we assume that $\lg \sigma \in \Z$, that the symbol $\texttt\textdollar := \sigma - 1$ does not occur in $T[0\dd n)$, and that we have access to the ${\lg \sigma}$-bit representation of $T[-n \dd 2n) := \texttt\textdollar^{n} \cdot T[0\dd n) \cdot \texttt\textdollar^n$. 
Now we show that this assumption is without loss of generality.

Let $\hat\sigma := 2^{\ceil{\lg (\sigma + 1)}}$ and note that $\hat\sigma \in (\sigma\dd 2\sigma]$.
We show how to compute the $\lg \hat\sigma$-bit representation of the string $\hat T[-n \dd 2n) := \hat{\texttt\textdollar}{}^{n} \cdot T[0\dd n) \cdot \hat{\texttt\textdollar}{}^n$ with $\hat{\texttt\textdollar} := \hat\sigma - 1 \geq \sigma$. 
This is done with a simple lookup table for translating short strings from $\ceil{\lg \sigma}$- to $\lg\hat\sigma$-bit representation. (This step is not needed if $\ceil{\lg \sigma} = \lg\hat\sigma$, i.e., if $\sigma$ is not a power of two.) The table stores, for each $S \in [0\dd \sigma)^{\floor{\frac12\log_\sigma n}}$, the string $S$ in $\lg \hat\sigma$-bit representation. It has $\Oh(\sqrt{n})$ entries, and each entry can be computed naively in $\Oh(\log_\sigma n)$ time. 
Hence, the table can be computed in $\tilde\Oh(\sqrt{n}) \subset \Oh(n / \lg n)$ time.
Finally, we use the table to convert $T[0\dd n)$ into $\lg\hat\sigma$-bit representation in a word-wise manner in $\Oh(n / \log_\sigma n)$ time.
Then, it is easy to obtain $\hat T[-n \dd 2n)$ by prepending and appending $\hat{\texttt\textdollar}{}^{n}$, again in a word-wise manner.

Recall that we only use this reduction whenever our aim is to achieve $\Oh(n / \log_\sigma n)$ \mbox{(pre-)processing} time. Due to $\Oh(n / \log_{\hat\sigma} n) = \Oh(n / \log_{2\sigma} n) = \Oh(n / \log_\sigma n)$, the reduction does not asymptotically increase the time complexity. Hence, we can indeed assume without loss of generality that $\sigma$ is a power of two, that $\texttt\textdollar = \sigma - 1$ does not occur in $T[0\dd n)$, and that we have access to the $\ceil{\lg \sigma}$-bit representation of $T[-n \dd 2n) := \texttt\textdollar^{n} \cdot T[0\dd n) \cdot \texttt\textdollar^n$.

\subparagraph{Accessing lookup tables with substrings.}
We will use substrings of length at most $(\log_\sigma n) / 4$ to access lookup tables.
For a string $S \in [0\dd \sigma)^*$ of length $0 \leq \absolute{S} \leq (\log_\sigma n) / 4$, we define its integer representation $\int(S)$ as follows:
\smallskip%
\begin{itemize}
    \item $\int(S)$ consists of $2 \cdot \floor{(\lg n) / 4}$ bits, i.e., $\int(S) \in [0 \dd 2^{2 \cdot \floor{(\lg n) / 4}}) \subseteq [0\dd \floor{\sqrt{n}})$.
    \item The upper half of $\int(S)$ contains the binary representation of $S$ (consisting of $\absolute{S} \cdot \lg \sigma \leq \floor{(\lg n) / 4}$ bits), padded with 0-bits.
    \item The lower half of $\int(S)$ stores the length $\absolute{S}$. This value can indeed be stored in $\ceil{\lg (1 + (\log_\sigma n) / 4)} < \floor{(\lg n) / 4}$ bits. (The inequality holds if $n$ exceeds some constant.)
\end{itemize}

\smallskip%
As mentioned earlier, extracting a substring $S = T[i\dd j)$ of length $0 \leq j - i \leq (\log_\sigma n) / 4$ takes constant time. Mapping $S$ to $\int(S)$ and vice versa also takes constant time.

\subsection{Revisiting Restricted Recompression}
\label{sec:linear_algorithm}

Following \cite{DBLP:journals/siamcomp/KociumakaRRW24}, we define integer sequences $(\lambda_k)_{k=0}^\infty$ so that  $\lambda_k = {\left(\frac87\right)}^{\floor{k/2}}$
and $(\alpha_k)_{k=0}^\infty$ so that $\alpha_0 = 1$ and recursively $\alpha_k = \alpha_{k - 1} + \floor{\lambda_{k - 1}}$.
These sequences control, respectively, the phrase lengths and context sizes at each recompression round.
As observed in \cite{DBLP:journals/siamcomp/KociumakaRRW24}, $\alpha_{k + 1} \leq 16 \lambda_k$ holds for every $k\in \Zz$.
Our goal is to compute the following sets:

\begin{proposition}[{\cite[Propositions 3.4 and 4.7]{DBLP:journals/siamcomp/KociumakaRRW24}}]
    \label{prop:bk_properties}
    For every length-$n$ text $T$, there exists a descending chain $[1\dd n)=\B_0 \supseteq \B_1 \supseteq \cdots \supseteq \B_q = \emptyset$ with $q = \Oh(\lg n)$ such that, for each $k\in \Zz$, the set $\B_{k}$ satisfies the following:
    \smallskip%
    \begin{enumerate}[(a)]
        \item\label{prop:bk_properties:set_size} $\absolute{\B_k} \leq \frac{4n}{\lambda_k}$
        \item\label{prop:bk_properties:synchronize} For $i,j \in [\alpha_k\dd n - \alpha_k]$, if $i \in \B_{k}$ and ${T[i - \alpha_k\dd i+\alpha_k)} \eqmatch {T[j - \alpha_k\dd j+\alpha_k)}$, then $j \in \B_k$.
        \item\label{prop:bk_properties:phrase_len} If $i,j$ are consecutive positions in $\B_k \cup \{0,n\}$, then $T[i\dd j)$ has length at most $\frac{7}{4}\lambda_k$, or its primitive root has length at most $\lambda_k$.
    \end{enumerate}%
    \smallskip%
    \noindent If $T\in [0\dd n^{\Oh(1)})^n$, then one can construct $\B_0 \supseteq \cdots \supseteq \B_{q-1} \supsetneq \B_q = \emptyset$ in $\Oh(n)$ time.
\end{proposition}

\begin{remark}
    While $q=\Oh(\lg n)$ is not stated in \cite[Propositions 3.4 and 4.7]{DBLP:journals/siamcomp/KociumakaRRW24}, it readily follows from \cref{prop:bk_properties}\eqref{prop:bk_properties:set_size}, in which we have $\B_k=\emptyset$ if $\lambda_k > 4n$ (see also \cite[p.~1542]{DBLP:journals/siamcomp/KociumakaRRW24}).
\end{remark}

Our sublinear-time solution for the initial rounds of restricted recompression closely follows the structure of the linear-time algorithm in the proof of \cref{prop:bk_properties} (see \cite[Proposition~4.7]{DBLP:journals/siamcomp/KociumakaRRW24}).
We now give a conceptual description of this algorithm (ignoring implementation details).
If we can show that the new solution is functionally identical to this algorithm, then it correctly computes the sets satisfying \cref{prop:bk_properties}.

We start with the set $\B_0 = [1\dd n)$.
Assume that, for some $k \in [0\dd q)$, we have computed the set $\B_k = \{f_1, \dots, f_m\}$ with $f_1 < f_2 < \dots < f_m$. 
Let $f_0 = 0$ and $f_{m+1} = n$, and define $F_i = T[f_i \dd f_{i+1})$ for $i \in [0\dd m]$.
Now we compute the set $\B_{k + 1}$.

\subparagraph{Case 1: \boldmath$k$ \unboldmath is even.} For each $i \in [1\dd m]$, we add $f_i$ to $\B_{k + 1}$ if and only if $\max(\absolute{F_{i - 1}}, \absolute{F_i}) > \lambda_k$ or $F_{i - 1} \neqmatch F_i$ (or both).
This can be viewed as merging each run of identical phrases (of length at most $\lambda_k$) into a single new phrase (see \cref{fig:evenround}).

\colorlet{palegray}{black!25!white}
\begin{figure}
\subcaptionbox{Construction of $\B_{k+1}$ for even $k$. Equal labels indicate matching substrings. In this example, $\textsf F$ is the only phrase of length more than $\lambda_k$ before constructing $\B_{k+1}$.\label{fig:evenround}}[\textwidth]{\centering\small
\definecolor{colA}{HTML}{8dd3c7}
\definecolor{colB}{HTML}{ffffb3}
\definecolor{colC}{HTML}{bebada}
\definecolor{colD}{HTML}{fb8072}
\definecolor{colE}{HTML}{80b1d3}
\definecolor{colF}{HTML}{fdb462}
\definecolor{colG}{HTML}{b3de69}
\definecolor{colH}{HTML}{fccde5}
\begin{tikzpicture}[x=.5em, y=1.2em,every node/.style={nosep}]

    \def\xstart{0}
    \foreach[
        evaluate=\x as \xend using int(\xstart + \x),
        remember=\xend as \xstart,
        count=\i from 0,
    ] \x in {3,2,3,3,3,3,2,4,3,5,5,5,3,2,2,3} {
        \node (s) at (\xstart, 0) {};
        \node (e) at (\xend, 1) {};
        \node[fit=(e.center)(s.center)] (F\i) {};
    }

    \foreach[count=\i from 0] \fname in {A,B,C,C,C,D,B,E,C,F,F,F,C,B,B,A} {
        \node at (F\i) {\textsf{\fname}};
        \node[draw, fill=col\fname, fit=(F\i)] {};
        \node at (F\i) {\textsf{\fname}};
    }

    \node[left=0 of F0] (tlab) {$T\ =\enskip$};
    
    \def\xstart{0}
    \foreach[
        evaluate=\x as \xend using int(\xstart + \x),
        remember=\xend as \xstart,
        count=\i from 0,
    ] \x in {3,2,9,3,2,4,3,5,5,5,3,4,3} {
        \node (s) at (\xstart, -1) {};
        \node (e) at (\xend, -0) {};
        \node[fit=(e.center)(s.center)] (F\i) {};
    }

    \foreach[count=\i from 0] \fname in {A,B,G,D,B,E,C,F,F,F,C,H,A} {
        \node at (F\i) {\textsf{\fname}};
        \node[draw, fill=col\fname, fit=(F\i)] {};
        \node at (F\i) {\textsf{\fname}};
        \node (tmp) at (F\i.east) {};
    }

    \node[left=0 of F0] (llab) {$T\ =\enskip$};
    
    \node[right=-\textwidth of tmp.east] (llab) {{factorization for $\B_{k + 1}$:}};
    \node[right=-\textwidth of tmp.east |- tlab] {{factorization for $\B_{k}$:}};
\end{tikzpicture}}

\medskip\smallskip

{\color{black!20!white}\hrule}

\medskip\smallskip

\subcaptionbox{Construction of $\B_{k+1}$ for odd $k$. Equal labels indicate matching substrings. In this example, $\textsf F$ is the only phrase of length more than $\lambda_k$ before constructing $\B_{k + 1}$. White, gray, and green blocks respectively correspond to phrases in $L$, phrases in $R$, and merged phrase pairs.\label{fig:oddround}}[\textwidth]{\centering\small
\colorlet{palegreen}{green!25!white}
\colorlet{palered}{red!25!white}
\colorlet{colA}{white}
\colorlet{colB}{palegray}
\colorlet{colC}{palegray}
\colorlet{colD}{white}
\colorlet{colE}{white}
\colorlet{colF}{palered}
\colorlet{colAB}{palegreen}
\colorlet{colAC}{palegreen}
\colorlet{colDC}{palegreen}
\colorlet{colEB}{palegreen}
\begin{tikzpicture}[x=.5em, y=1.2em,every node/.style={nosep}]

    \def\xstart{0}
    \foreach[
        evaluate=\x as \xend using int(\xstart + \x),
        remember=\xend as \xstart,
        count=\i from 0,
    ] \x in {3,2,3,3,2,3,4,3,2,6,6,3,4,3,2} {
        \node (s) at (\xstart, 0) {};
        \node (e) at (\xend, 1) {};
        \node[fit=(e.center)(s.center)] (F\i) {};
    }

    \foreach[count=\i from 0] \fname in {A,B,A,C,D,C,E,A,B,F,F,C,E,B,C} {
        \node at (F\i) {\textsf{\fname}};
        \node[draw, fill=col\fname, fit=(F\i)] {};
        \node at (F\i) {\textsf{\fname}};
    }

    \node[left=0 of F0] (tlab) {$T\ =\enskip$};
    
    \def\xstart{0}
    \foreach[
        evaluate=\x as \xend using int(\xstart + \x),
        remember=\xend as \xstart,
        count=\i from 0,
    ] \x in {5,6,5,4,5,6,6,3,7,2} {
        \node (s) at (\xstart, -1) {};
        \node (e) at (\xend, -0) {};
        \node[fit=(e.center)(s.center)] (F\i) {};
    }

    \foreach[count=\i from 0] \fname in {AB,AC,DC,E,AB,F,F,C,EB,C} {
        \node at (F\i) {\textsf{\fname}};
        \node[draw, fill=col\fname, fit=(F\i)] {};
        \node at (F\i) {\textsf{\fname}};
        \node (tmp) at (F\i.east) {};
    }

    \node[left=0 of F0] (llab) {$T\ =\enskip$};
    
    \node[right=-\textwidth of tmp.east] (llab) {{factorization for $\B_{k + 1}$:}};
    \node[right=-\textwidth of tmp.east |- tlab] {{factorization for $\B_{k}$:}};
\end{tikzpicture}}

\medskip\smallskip

{\color{black!20!white}\hrule}

\medskip\smallskip

\subcaptionbox{Multigraph (left) and weighted graph (right) for the factorization from \cref{fig:oddround}, where $\absolute{E} = 11$ and the sum of weights from $L$ to $R$ is $5$.\label{fig:graph}}[\textwidth]{\centering\small
\begin{tikzpicture}[y=2em, every node/.style={inner sep=1pt}]
    \def\theshift{3.5cm}
    \node[circle, draw, xshift=-\theshift] (A) at (1.0, 0.0) {\textsf A};
    \node[circle, draw, xshift=-\theshift, fill=palegray] (B) at (0.309, 0.951) {\textsf B};
    \node[circle, draw, xshift=-\theshift, fill=palegray] (C) at (-0.809, 0.588) {\textsf C};
    \node[circle, draw, xshift=-\theshift] (D) at (-0.809, -0.588) {\textsf D};
    \node[circle, draw, xshift=-\theshift] (E) at (0.309, -0.951) {\textsf E};

    \draw[-Latex] (A) to (B);
    \draw[-Latex] (A) to[bend right] (B);
    \draw[-Latex] (A.200) to[out=200, in=0] (C);
    \draw[-Latex] (B) to[bend right] (A);
    \draw[-Latex] (B) to[bend right] (C);
    \draw[-Latex] (C) to (D);
    \draw[-Latex] (C) to[bend right] (D);
    \draw[-Latex] (C) to[out=-40, in=150] (E);
    \draw[-Latex] (D) to[bend right] (C);
    \draw[-Latex] (E) to[bend right] (A);
    \draw[-Latex] (E) to[bend left] (B);

    \node[circle, draw] (A) at (1.0, 0.0) {\textsf A};
    \node[circle, draw, fill=palegray] (B) at (0.309, 0.951) {\textsf B};
    \node[circle, draw, fill=palegray] (C) at (-0.809, 0.588) {\textsf C};
    \node[circle, draw] (D) at (-0.809, -0.588) {\textsf D};
    \node[circle, draw] (E) at (0.309, -0.951) {\textsf E};

    \draw[-{Latex[length=5pt]}, very thick, densely dotted] (A) to[bend right] (B);
    
    \draw[-Latex] (A.200) to[out=200, in=0] (C);
    \draw[-Latex] (B) to[bend right] (A);
    \draw[-Latex] (B) to[bend right] (C);
    
    \draw[-{Latex[length=5pt]}, very thick, densely dotted] (C) to[bend right] (D);
    
    \draw[-Latex] (C) to[out=-40, in=150] (E);
    \draw[-Latex] (D) to[bend right] (C);
    \draw[-Latex] (E) to[bend right] (A);
    \draw[-Latex] (E) to[bend left] (B);

    \node[above right=.25em and 3em of A] (u) {};
    \node[below right=.25em and 3em of A] (d) {};

    \draw[-Latex] (u) to node[pos=1, right=1em] {edge of weight $1$} ++(1, 0) {};
    \draw[-{Latex[length=5pt]}, very thick, densely dotted] (d) to node[pos=1, right=1em] {edge of weight $2$} ++(1, 0) {};
    
\end{tikzpicture}}
\caption{Even and odd rounds of restricted recompression.}
\end{figure}

\subparagraph{Case 2: \boldmath$k$ \unboldmath is odd.}
Let $\mathcal F = \{ F_i \mid i \in [0\dd m]\textnormal{\ and\ } \absolute{F_i} \leq \lambda_k \}$ be the set of distinct phrases of length at most $\lambda_k$, viewed as strings rather than fragments (so matching phrases represent the same element of $\mathcal F$). We partition $\mathcal F$ into sets $L$ and $R$.
For each $i \in [1\dd m]$, we add $f_i$ to $\B_{k + 1}$ unless both $F_{i - 1} \in L$ and $F_{i} \in R$. This amounts to merging adjacent pairs of phrases whenever the left phrase is in $L$ and the right phrase is in~$R$ (see \cref{fig:oddround}). 

The sets $L, R$ are computed by modeling the factorization as a directed multigraph whose nodes are the elements of $\mathcal F$. For every $i \in [1\dd m]$, we add an edge from $F_{i - 1}$ to $F_i$ if 
$\max(\absolute{F_{i - 1}}, \absolute{F_i}) \leq \lambda_k$.
(Due to the preceding even round, it is clear that $F_{i - 1} \neqmatch F_i$.)
Let $E$ be the multiset of all edges.
To obtain the properties in \cref{prop:bk_properties}, it suffices to partition $\mathcal F$ into $L$ and $R$ so that at least $\frac14\cdot\absolute{E}$ edges go from $L$ to $R$.
This can be achieved by approximating a maximum directed cut in the multigraph (e.g., using \cite[Lemma A.1]{DBLP:journals/siamcomp/KociumakaRRW24}).

Equivalently, we can use a weighted directed graph instead of a multigraph.
The set of nodes is still $\mathcal F$. For each pair of nodes $F', F''$, the weight of the edge from $F'$ to $F''$ is the number of edges from $F'$ to $F''$ in the multigraph. Then, we partition $\mathcal F$ into $L$ and $R$ so that the sum of the weights of edges from $L$ to $R$ is at least $\frac14 \cdot \absolute{E}$ (see \cref{fig:graph}).

\section{Restricted Recompression in Sublinear Time}\label{sec:recompression}

In this section, we show how to implement restricted recompression in $\Oh(n / \log_\sigma n)$ time. If $\log_\sigma n$ is constant, then we can afford linear time and use the solution from \cite{DBLP:journals/siamcomp/KociumakaRRW24}. Hence, assume $\log_\sigma n \geq 256$.
We fix $K := 2 \cdot \floor{\log_{8/7}\left(256^{-1} \cdot {\log_{\sigma} n}\right)}$. Recalling that $\lambda_K = (\frac87)^{\floor{K/2}}$, we observe that ${\frac78 \cdot 256^{-1} \cdot \log_\sigma n} \leq \lambda_K \leq 256^{-1} \cdot \log_\sigma n$.
During rounds $0, \dots, K$ of restricted recompression, by \cref{prop:bk_properties}\eqref{prop:bk_properties:synchronize}, the phrase boundaries are chosen based on short contexts of length at most $2\alpha_K \leq 32\lambda_K \leq 8^{-1} \cdot \log_\sigma n$, and hence this process can be accelerated using word-packing techniques. After round $K$, the number of boundary positions is at most $\frac{4n}{\lambda_K} = \Oh(n / \log_\sigma n)$ by \cref{prop:bk_properties}\eqref{prop:bk_properties:set_size}, and we can implement the remaining rounds in $\Oh(n / \log_\sigma n)$ time using the linear-time solution presented in \cite{DBLP:journals/siamcomp/KociumakaRRW24}. We only need the following straightforward corollaries.

\begin{definition}
Let $\B_k = \{ f_1, \dots, f_m \} \subseteq [1\dd n)$ with $f_1 < f_2 < \dots < f_m$ be one of the sets computed in \cref{prop:bk_properties}, and let $f_0 = 0$ and $f_{m + 1} = n$.
We define
\begin{itemize}
    \item the (not necessarily unique) string of phrase names $I_k[0\dd m]$ over alphabet $[0\dd m]$, where, for every $i, j \in [0\dd m]$, it holds $I_k[i] = I_k[j]$ if and only if $T[f_i\dd f_{i + 1}) \eqmatch T[f_{j}\dd f_{j + 1})$,
    \item the array $\len_k[0\dd m]$ of phrase lengths with $\len_k[i] = f_{i + 1} - f_{i}$ for all $i \in [0\dd m]$.
\end{itemize}

\end{definition}

\begin{corollary}[{of \cite[Proposition 4.7]{DBLP:journals/siamcomp/KociumakaRRW24}}]
    \label{cor:large_rounds:from_Ik}
    Let $k \in [0\dd q)$. Given $\B_k$, $I_k$, and $\len_k$, one can compute 
    $\B_{k + 1} \supseteq \B_{k + 2} \supseteq \dots \supseteq \B_{q - 1} \supsetneq\B_q = \emptyset$
    in $\Oh(n / \lambda_k)$ time.
\end{corollary}

\begin{proof}
    The algorithm from the proof of \cite[Proposition 4.7]{DBLP:journals/siamcomp/KociumakaRRW24} uses $\B_k, I_k$, and $\len_k$ to compute $\B_{k+1}, I_{k+1}$, and $\len_{k+1}$ in $\Oh(\absolute{\B_k})$ time. Hence, the remaining sets $\B_{k + 1} \supseteq \B_{k + 2} \supseteq \dots \supseteq \B_q = \emptyset$ can be computed in $\Oh(\sum_{k' = k}^{q-1} \absolute{\B_{k'}})$ time.
    By \cref{prop:bk_properties}\eqref{prop:bk_properties:set_size} and the definition of $\lambda_{k'}$, the sum is bounded by $\Oh(n / \lambda_k)$.
\end{proof}

\begin{corollary}\label{cor:large_rounds:from_Bk}
    Given the elements of $\B_K$ in increasing order, one can compute the sets
    $\B_{K + 1} \supseteq \B_{K + 2} \supseteq \dots \supseteq \B_{q-1} \supsetneq \B_q = \emptyset$
    in $\Oh(n / \log_\sigma n)$ time.
\end{corollary}

\begin{proof}
    Let $\B_K = \{f_1, \dots, f_m\}$ with $f_i < f_{i + 1}$ for all $i \in [1\dd m)$, and define $f_0 = 0$ and $f_{m + 1} = n$. Note that $m = \Oh(n / \lambda_K) = \Oh(n / \log_\sigma n)$ by \cref{prop:bk_properties}\eqref{prop:bk_properties:set_size} and the definitions of $K$ and $\lambda_K$. By \cref{cor:large_rounds:from_Ik}, it suffices to compute $I_K$ and $\len_K$ in $\Oh(n / \lambda_K)$ time, which is trivial for $\len_K$. To compute $I_K$, we first produce a string $I'_K$ over a slightly larger alphabet. For each $i \in [0\dd m]$, we encode phrase $T[f_{i}\dd f_{i + 1})$ as $I'_K[i] = (\int(S_i), \len_K[i])$, where $S_i$ is the truncated phrase ${T[f_{i}\dd \min(f_{i + 1}, f_{i} +2\lambda_K))}$.
    Trivially, two phrases are identical if and only if they have the same length and primitive root. By \cref{prop:bk_properties}\eqref{prop:bk_properties:phrase_len}, the primitive root of any phrase is of length at most $\frac{7}{4}\lambda_K$. Since $S_i$ is either the entire phrase or a length-$2\lambda_K$ prefix of the phrase, it is also a (possibly fractional) power of the primitive root of the phrase. Hence, it is easy to see that we indeed encode two phrases identically if and only if they are identical.
    Clearly, $I'_K$ can be computed in $\Oh(m)$ time.
    We obtain $I_K$ by reducing the alphabet to $[0\dd m]$ using radix sort in $\Oh(n / \log_\sigma n)$ time.
\end{proof}

\subsection{Performing the Initial \texorpdfstring{\boldmath$K$\unboldmath}{K} Rounds}

In the initial rounds of recompression, the weights in the graph of adjacent phrases depend on the abundance of short substrings. Hence, we use the following simple index for counting short substrings. (This kind of result is well-known; we provide a proof merely for completeness.)

\begin{restatable}{lemma}{restateindexshortsubstrings} \label{lem:short_substrings}
    For any string $T \in [0\dd \sigma)^n$, in $\Oh(n/ \log_\sigma n)$ time, we can construct a data structure that answers the following type of query in constant time. Given a string $S \in [0\dd \sigma)^*$ of length at most $(\log_\sigma n) / 8$, return the 
        number of occurrences of $S$ in $T$.
\end{restatable}

\begin{proof}
	We first construct an auxiliary lookup table $L[0\dd \floor{\sqrt{n}})$ that is initially all-zero.
	If $\log_\sigma n<8$, then no queries are possible; hence, assume $\log_\sigma n \geq 8$. Let $b = \floor{(\log_\sigma n) / 8}$.
    For every $i \in [0\dd \ceil{n / b})$, we extract $S_i := T[ib\dd \min(n, ib + 2b))$ and increment $L[\int(S_i)]$ in constant time.
	Afterwards, for every $S \in [0\dd \sigma)^{2b}$, we may interpret the entry $L[\int(S)]$ as follows. There are $L[\int(S)]$ occurrences of~$S$ in~$T$ that start at positions that are multiples of~$b$.
	The time for computing $L$ is $\Oh(n / \log_\sigma n)$.
	
	A second table $L'[0\dd \floor{\sqrt{n}})$, also initialized with zeros, will serve as the actual index. For every string $S' \in [0\dd \sigma)^*$ of length at most $b$, entry $L'[\int(S')]$ will store the number of occurrences of $S'$ in $T$.
	The table is constructed as follows.
	We consider every $S \in [0\dd \sigma)^*$ of length at most $2b$ and obtain the value $s := L[\int(S)]$. (If the length of $S$ is less than $2b$, then necessarily $s = 0$, unless $S$ is one of the final two blocks.)
	For every $x, \ell \in [0 \dd b)$ that satisfy $
    x + \ell \leq \absolute{S}$, we increase $L'[\int(S[x \dd x + \ell])]$ by $s$, which takes constant time.
	The time for computing $L'$ is $\Oh(\sqrt{n} \cdot b^2) \subset \Oh(n / \lg n)$.
		
	It is easy to see that the procedure works as intended. Particularly, every fragment of length at most $b$ will be considered by exactly one of the aligned fragments of length~$2b$. This is because, for every length-$2b$ block starting at a position that is a multiple of $b$, we only consider fragments starting within the first $b$ positions of the block.	
\end{proof}

\subparagraph{Defining phrase boundaries via substrings.} 
Instead of directly computing $\B_k$, we compute the intermediate representation%
\footnote{Recall that $T$ is padded so that $T[-n \dd 2n) = \texttt\textdollar^{n} \cdot T[0\dd n) \cdot \texttt\textdollar^n$. Hence, $T[i-\alpha_k\dd i+\alpha_k)$ with $i \in [0\dd n]$ is always defined. Also, if ${i \in [0\dd n] \setminus [\alpha_k\dd n-\alpha_k]}$, then $T[i-\alpha_k\dd i + \alpha_k)$ is unique in $T[-\alpha_k\dd n+\alpha_k)$.}
%
$%
\tB_k = \{T[i - \alpha_k\dd i + \alpha_k) \mid  i \in \B_k \cup \{0, n \}\},
$ %
%
i.e., rather than explicitly listing the boundary positions, we instead list the set of distinct contexts that cause a boundary.
From now on, for $k \in [0\dd q]$ and $i \in \B_k$, we say that a string $T[i-\alpha_k\dd i+\alpha_k)\in \tB_k$ is a \emph{boundary context of $\B_k$}.
For every $i \in [0\dd n]$, it follows from \cref{prop:bk_properties}\eqref{prop:bk_properties:synchronize} that $i \in \B_k \cup \{0, n\}$ if and only if $T[i - \alpha_k\dd i+\alpha_k) \in \tB_k$.

\begin{lemma}\label{lem:ck}
    The sets $\tB_1, \dots, \tB_{K-1},\tB_{K}$ can be computed in $\Oh(n / \log_\sigma n)$ time,
    with each set $\tB_k$ encoded as a bitmask of length $\floor{\sqrt{n}}$ whose set bits are $\{\int(S) : S\in \tB_k\}$.
\end{lemma}

\begin{proof} 
The sets $\tB_1, \dots, \tB_{K-1},\tB_{K}$ consist of strings of length up to $2\alpha_K \leq 8^{-1} \cdot \log_\sigma n$. 
We can encode any $S\in \tB_k$ as $\int(S) \in [0\dd \floor{\sqrt{n}})$. 
We construct the data structure from \cref{lem:short_substrings} for the padded string $\texttt\textdollar^{2\alpha_K} \cdot T \cdot \texttt\textdollar^{2\alpha_K}$, which allows querying for $S$.
We can also query for strings of the form $\texttt\textdollar^{j_1} \cdot T[0\dd j_2)$ and $T[n - j_1\dd n) \cdot \texttt\textdollar^{j_2}$ with $j_1 + j_2 \leq 2\alpha_K$.

\subparagraph{Initialization.} We initialize each of $\tB_1, \dots, \tB_K$ as an all-zero bitmask of length $\floor{\sqrt{n}}$. Recall that $\alpha_0 = 1$. For computing $\tB_0$, we enumerate all possible length-two strings over $[0\dd \sigma)$, except for $\texttt\textdollar^2$. For each such string~$S$, we query the data structure from \cref{lem:short_substrings}. If $S$ occurs in $T$, we set the $\int(S)$-th bit of $\tB_0$ to one. 
Recall that $2 \leq 2\alpha_K \leq 8^{-1} \cdot \log_\sigma n$; hence, we can indeed afford to enumerate all strings of length up to $2\alpha_K$ (e.g., the $\sigma^2$ strings of length two during initialization) in $\Oh(2^{2\alpha_K}) \subset \Oh(\sqrt{n})$ time.

\subparagraph{Preparing round \boldmath$k + 1$\unboldmath.} Assume that we have already computed $\tB_k$ for some $k \in [0\dd K)$, and our goal is to compute $\tB_{k + 1}$.
The main computational challenge is the following. For each length-$2\alpha_{k + 1}$ string $S$ over alphabet $[0\dd \sigma)$, except for $\texttt\textdollar^{2\alpha_{k + 1}}$, we have to decide if $S\in \tB_{k+1}$.
We enumerate all the $\Oh(\sqrt{n})$ possible strings of length $2\alpha_{k + 1}$. For each such string~$S$, we first check if it is indeed a substring of the padded string $T$ using the data structure from \cref{lem:short_substrings}.
If $S$ is not a substring, then it is not a boundary context of $\B_{k+1}$ and can be skipped.
Otherwise, we consider its central part $S' = S[\alpha_{k + 1} - \alpha_k\dd \alpha_{k + 1} + \alpha_k)$.
If $S'$ is not a boundary context of $\B_k$, i.e., if $S' \notin \tB_k$ (a check performed in constant time), then $S$ cannot be a boundary context of $\B_{k + 1}$ due to $\B_{k + 1} \subseteq \B_k$. Hence, if $S' \notin \tB_k$, we do not have to process $S$ any further.

If, however, it holds $S' \in \tB_k$, then we have to decide whether we will make $S$ a boundary context of $\B_{k + 1}$ by adding $S$ to $\tB_{k + 1}$. 
For the sake of explanation, consider any position $i \in [1\dd n)$ such that $T[i-\alpha_{k + 1}\dd i+\alpha_{k + 1}) \eqmatch S$.
We now explain how to compute the minimal values $\ell,  r \in [1\dd \floor{\lambda_k}]$ such that 
$i - \ell \in \B_k\cup\{0,n\}$ and $i + r \in \B_k\cup\{0,n\}$.
If both values exist, then
the factorization induced by $\B_k$ contains fragments $T[i-\ell\dd i)$ and $T[i\dd i+r)$ as phrases.
If, however, $\ell$ or $r$ does not exist, then we know that the phrase ending at position $i - 1$ or the phrase starting at position $i$ is of length more than $\floor{\lambda_k}$ (possibly both). When creating $\B_{k + 1}$, the algorithm from \cite[Proposition~4.7(a)]{DBLP:journals/siamcomp/KociumakaRRW24} (see also \cref{sec:linear_algorithm}) merges two adjacent phrases only if both of them are of length at most $\lambda_k$. 
Hence, if $\ell$ or $r$ does not exist, then the phrases around boundary $i$ cannot be merged, and $i$ is a boundary position in $\B_{k + 1}$. 

We now explain how to compute $\ell$ or show that it does not exist; the computation for $r$ is symmetric.
Consider any $\ell \in [1\dd \floor{\lambda_k}]$. Due to $\alpha_{k + 1} = \alpha_k + \floor{\lambda_k}$ and $S \eqmatch T[i- \alpha_{k + 1}\dd i+\alpha_{k + 1})$, it holds
\[
T[i-\alpha_k - \ell\dd i+\alpha_k - \ell) \eqmatch 
S[\floor{\lambda_k} - \ell\dd \floor{\lambda_k} + 2\alpha_k - \ell).
\]
Hence, we can check if $i - \ell$ is in $\B_k\cup\{0,n\}$ by probing $\tB_k$ with $S[\floor{\lambda_k} - \ell\dd \floor{\lambda_k} + 2\alpha_k - \ell)$ in constant time.
By trying all possible values, finding the minimal suitable $\ell$ takes $\Oh(\lambda_k)$ time.
If both $\ell$ and $r$ exist, then we add a tuple $\angles{S, \ell, r}$ to a list $\mathcal L$. Otherwise, as explained above, we have to make $S$ a boundary context of $\B_{k+1}$ by adding $S$ to $\tB_{k + 1}$.

After processing all possible strings $S$, we have the following situation.
If a boundary position in $\B_k\cup\{0,n\}$ is adjacent to a phrase of length over $\floor{\lambda_k}$, then the corresponding boundary context has been added to $\tB_{k + 1}$. (This is always the case for the contexts $\texttt\textdollar^{\alpha_{k + 1}} \cdot T[0\dd \alpha_{k + 1})$ and $T[n - \alpha_{k + 1}\dd n) \cdot \texttt\textdollar^{\alpha_{k + 1}}$ of positions $0$ and $n$, respectively).
If a boundary position in~$\B_k$ has context~$S$ and is adjacent to phrases of respective lengths $\ell, r \leq \floor{\lambda_k}$, then $\angles{S, \ell, r}$ has been added to~$\mathcal L$. Since each context has been added to $\mathcal L$ at most once, the elements of $\mathcal L$ are distinct.

\subparagraph{Performing an even round.}
Consider a boundary position in $\B_k$ for which the two adjacent phrases are identical and of length at most $\floor{\lambda_k}$. The run of identical phrases will be merged in $\B_{k + 1}$, and the boundary position will no longer exist. Hence, we proceed as follows. We consider each element $\angles{S, \ell, r}$ of the list $\mathcal L$. If $S[\alpha_{k + 1} - \ell\dd \alpha_{k + 1}) \neqmatch S[\alpha_{k + 1}\dd \alpha_{k + 1} + r)$, then we add $S$ to $\tB_{k + 1}$; otherwise, we do nothing (skip the element of $\mathcal L$).

\subparagraph{Performing an odd round.}
We have to produce the weighted directed graph described in \cref{sec:linear_algorithm}.
The set of nodes is %
$%
\mathcal F = \bigcup_{\angles{S, \ell, r} \in \mathcal L} \{ S[\alpha_{k + 1} - \ell\dd \alpha_{k + 1}), S[\alpha_{k + 1}\dd \alpha_{k + 1} + r )\}.%
$%

For each $\angles{S, \ell, r} \in \mathcal L$, we obtain the number $s$ of occurrences of $S$ in the padded~$T$ using the data structure from \cref{lem:short_substrings}. We increase the weight of the edge from $S[\alpha_{k + 1} - \ell\dd \alpha_{k + 1})$ to $S[\alpha_{k + 1}\dd \alpha_{k + 1} + r)$ by $s$. 
(This results in no self-loops, as phrases in~$\mathcal F$ are of length $\leq \floor{\lambda_k}$, and runs of phrases of length $\leq \floor{\lambda_k}$ have been eliminated in the preceding even round.)
Observe that $\absolute{\mathcal L} = \Oh(\sqrt{n})$, and each element of $\mathcal L$ contributes weight to one edge. Therefore, the number of edges is $\Oh(\sqrt{n})$, and we can approximate the maximum cut in $\Oh(\sqrt{n})$ time (see, e.g., \cite[Lemma A.1]{DBLP:journals/siamcomp/KociumakaRRW24}, which immediately works for weighted graphs). This reveals the two parts $L$ and $R$ used for computing $\B_{k + 1}$.
We once more consider each $\angles{S, \ell, r} \in \mathcal L$ and check if $S[\alpha_{k + 1} - \ell\dd \alpha_{k + 1}) \in L$ and $S[\alpha_{k + 1}\dd \alpha_{k + 1} + r) \in R$.
Whenever this is the case, we do nothing (skip the element of $\mathcal L$). Otherwise, we add $S$ to $\tB_{k + 1}$.

\subparagraph{Time complexity and correctness.} 
In each round, we have to consider all the $\Oh(\sqrt{n})$ possible strings of length $2\alpha_{k + 1}$. We process each string in $\Oh(\lambda_k) \subseteq \Oh(\lg n)$ time, dominated by the time needed to compute $\ell$ and $r$. We spend additional $\Oh(\sqrt{n})$ time to approximate the maximum cut.
Recalling that $K = \Oh(\log \log n)$, the overall time is $\Oh(K \cdot \sqrt{n} \cdot \lg n) \subset o(n / \lg n)$, plus $\Oh(n / \log_\sigma n)$ time for the preprocessing of \cref{lem:short_substrings}. The correctness follows from the fact that the algorithm directly implements the steps described in \cref{sec:linear_algorithm}. 
\end{proof}

\subsection{Reporting the Phrase Boundaries}

The number of boundaries in the initial $K$ rounds may significantly exceed $\Oh(n / \log_\sigma n)$, and thus we cannot afford to report them explicitly. Instead, we report a bitmask of length $n$ that marks the boundaries.

\begingroup\def\strset{\mathcal C}
\begin{lemma}\label{lem:bitmask}
    For an integer $1\le \ell \le (\log_\sigma n)/8$, consider a set $\strset\subseteq [0\dd \sigma)^\ell$.
    Given a text $T\in [0\dd\sigma)^n$ and an $\Oh(1)$-time membership oracle to $\strset$, a bitmask representing $\{i\in [0\dd n-\ell] : T[i\dd i+\ell)\in \strset\}$ can be constructed in $\Oh(n/\log_\sigma n)$ time.
\end{lemma}
\begin{proof}
    First, we construct an array $B[0\dd \floor{\sqrt{n}})$ so that, for each string $U$ of length ${\ell \le |U| \le 2\ell \leq (\log_\sigma n)/4}$, the entry $B[\int(U)]$ is a bitmask of length $|U|-\ell+1$ representing $\{i\in [0\dd |U|-\ell] : U[i\dd i+\ell)\in \strset\}$.
    Due to constant-time oracle access to $\strset$, the construction of $B$ takes $\Oh(\sqrt{n}\cdot \ell)=\tOh(\sqrt{n})$ time.

    Next, for each $j\in [0\dd \floor{n/\ell})$, we extract a string $T_j = T[j\ell \dd \min(j\ell+2\ell-1,n))$. 
    In other words, this is a decomposition of $T$ into substrings of length $2\ell-1$ (with the last one of length between $\ell$ and $2\ell-1$) overlapping by exactly $\ell-1$ positions.
    Finally, we construct the resulting bitmask by concatenating the bitmasks $B[\int(T_j)]$ for subsequent indices $[0\dd \floor{n/\ell})$. 
    This is valid because, for each $i\in [0\dd n-\ell]$, we have $T[i\dd i+\ell)=T_{\floor{i/\ell}}[i\bmod \ell \dd i\bmod \ell+\ell)$ and the concatenated bitmasks are of length exactly $\ell$ (except for the last one, whose length is between $1$ and $\ell$).
\end{proof}
\endgroup

The well-known folklore \lcnamecref{lem:bitmasktoexplicit} below efficiently transforms a bitmask into an explicit set; a proof is provided for completeness.

\begin{lemma}\label{lem:bitmasktoexplicit}
    Given a bitmask $M$ of length $n$, the set $\{ i \in [0\dd n) \mid M[i] = 1 \}$ can be output in increasing order in $\Oh(n / \lg n + \sum_{i = 0}^{n - 1} M[i])$ time.
\end{lemma}

\begin{proof}
    We use a lookup table $L[0\dd2^{\floor{\lg n /2}})$, where entry $L[x]$ contains a list of the one-bits in the binary representation of $x$ in left to right order. 
    Clearly, the table can be computed in $\tilde\Oh(\sqrt{n})$ time. Then, we can process the mask $M$ in chunks of size $\floor{\lg n / 2}$ (padding the final chunk with zeros). For each chunk, we look up the list of one-bits and report them after applying the appropriate offset. This takes $\Oh(n / \lg n + \sum_{i = 0}^{n - 1} M[i])$ time.
\end{proof}

\begin{lemma}\label{lem:bk-bitmask}
    Given $T\in [0\dd \sigma)^n$ and a parameter $k\in \Zz$, a set $\B_k$ (represented as a bitmask) satisfying the conditions of \cref{prop:bk_properties} can be constructed in $\Oh(n/\log_\sigma n)$ time.
\end{lemma}
\begin{proof}
    If $k \le K$, the set $\tB_k$ constructed in the proof of \cref{lem:ck} serves as an $\Oh(1)$-time membership oracle to the set of boundary contexts. We then convert $\tB_k$ to $\B_k$ using \cref{lem:bitmask} on the text $\texttt\textdollar^{\alpha_{k}} \cdot T \cdot \texttt\textdollar^{\alpha_{k}}$, discarding the initial and final value (corresponding to positions $0$ and $n$). (\cref{lem:bitmask} marks the leftmost position of each boundary context occurrence, and contexts are of length $2\alpha_k$; since position $i$ in $\texttt\textdollar^{\alpha_{k}} \cdot T \cdot \texttt\textdollar^{\alpha_{k}}$ corresponds to position $i - \alpha_k$ in $T$, the reported positions are the central positions of boundary contexts in $T$, as required.)
    If $k > K$, we first construct the bitmask for $\B_K$ and then convert it into an explicit list of positions with \cref{lem:bitmasktoexplicit}. 
    Then, we use \cref{cor:large_rounds:from_Bk} to obtain $\B_k$ in explicit representation in $\Oh({n / \log_\sigma n})$ time, which we convert back into a bitmask in $\Oh(n / \log n + \absolute{\B_k})$ time.
    Due to $\absolute{\B_k} \leq \absolute{\B_K} = \Oh(n / \log_\sigma n)$, the time is as claimed.
\end{proof}

\begin{lemma}\label{lem:bk-explicit}
A text $T\in [0\dd \sigma)^n$ can be preprocessed in $\Oh(n/\log_\sigma n)$ time so that, given $k\in \Zz$, a set $\B_k$ satisfying the conditions of \cref{prop:bk_properties} can be constructed in $\Oh(1+n/\lambda_k)$ time. The elements of $\B_k$ are reported from left to right.
\end{lemma}
\begin{proof}
    For $k \leq K$, we can simply use \cref{lem:bk-bitmask,lem:bitmasktoexplicit}, which takes $\Oh(n / \log_\sigma n + \absolute{\B_k}) = \Oh(n / \lambda_k)$ time.
    For all $k > K$, the sets can be computed during preprocessing as follows. We compute $\B_K$ in explicit representation using \cref{lem:bk-bitmask,lem:bitmasktoexplicit}. Then, we obtain all the remaining sets using \cref{cor:large_rounds:from_Bk}.
\end{proof}

\section{Computing \texorpdfstring{\boldmath$\tau$\unboldmath}{tau}-Runs in Sublinear Time}\label{sec:runs}

\newcommand{\per}{\mathsf{per}}
\newcommand{\run}{\mathsf{run}}
\newcommand{\RUNS}{\mathsf{RUNS}}
\newcommand{\R}{\mathsf{R}}

An integer $p\in [1\dd |S|]$ is a \emph{period} of a string $S$ if $S[i] = S[i+p]$ holds for each $i\in [0\dd |S|-p)$.
We denote the smallest period of a non-empty string $S$ by $\per(S)$, and we call $S$ \emph{periodic} whenever $\per(S)\le \frac12 |S|$.

A \emph{run} (a \emph{maximal repetition}) in a string $T$ is a periodic fragment $\gamma=T[i\dd j)$ of $T$ that can be extended neither to the left nor to the right without increasing the smallest period $p = \per(\gamma)$, i.e., $i=0$ or $T[i - 1] \ne T[i + p - 1]$, and $j=n$ or $T[j] \ne T[j - p]$. 
The set of all runs in $T$ is denoted by $\RUNS(T)$.
Crucially, the periodicity lemma~\cite{FineWilf1965} implies that distinct runs cannot overlap too much.

\begin{fact}[{
\cite[Lemma~1(ii)]{DBLP:conf/focs/KolpakovK99}}]\label{fct:overlap}
    Let $\gamma,\gamma'\in \RUNS(T)$ be distinct yet overlapping runs in a string $T$.
    If $p=\per(\gamma)$ and $p'=\per(\gamma')$, then $|\gamma \cap \gamma'| < p+p'-\gcd(p,p')$.
\end{fact}

Among many consequences of \cref{fct:overlap} is that every periodic fragment $T[i\dd j)$ can be uniquely extended to a run $T[i'\dd j')\in \RUNS(T)$ such that $i'\le i \le j \le j'$ and $\per(T[i'\dd j'))=\per(T[i\dd j))$. We denote this unique extension by $\run(T[i\dd j))$. If $T[i\dd j)$ is not periodic, we write $\run(T[i\dd j))=\bot$ to indicate that the run extension is undefined.

\begin{theorem}[{\cite[Theorem 1.5]{DBLP:journals/siamcomp/KociumakaRRW24}}]\label{thm:run_ext}
    A text $T\in [0\dd \sigma)^{n}$ can be preprocessed in $\Oh(n/\log_\sigma n)$ time so that, given a fragment $x$ of $T$, the run extension $\gamma = \run(x)$ can be computed in $\Oh(1)$ time. 
    If $\gamma\ne \bot$, then the query algorithm also returns the shortest period $\per(\gamma)=\per(x)$.
\end{theorem}

For $\ell,p\in [0\dd n]$, we write
$%
\RUNS_{\ell,p}(T) = \{\gamma \in \RUNS(T) : |\gamma|\ge \ell \text{ and }\per(\gamma)\le p\}.
$
The construction of string synchronizing sets relies on \emph{$\tau$-runs} defined for $\tau\in [1\dd n]$ as \[\RUNS_\tau(T)= \RUNS_{\tau,\floor{\tau/3}}(T).\]

\begin{restatable}{proposition}{restaterunsexplicit}\label{prp:runs-explicit}
    One can preprocess a text $T\in [0\dd \sigma)^n$ in $\Oh(n/ \log_\sigma n)$ time so that, given integers $\ell,p\in [0\dd n]$ with $\ell \ge 2p$, one can output the set $\RUNS_{\ell,p}(T)$ in $\Oh(n/(\ell+1-2p))$ time.
    The runs are reported together with their periods and simultaneously ordered by their start and end positions.
\end{restatable}
\begin{proof}
    In the preprocessing phase, we construct the data structure of \cref{thm:run_ext}.
    At query time, if $p = 0$, then we trivially return the empty set. Otherwise, we pick $\Delta = \ell+1-2p$ and, for  $i \in [0\dd \floor{(n-2p)/\Delta}]$, consider fragments $T_i = T[i\Delta \dd i\Delta+2p)$.
    For each $T_i$, we compute $\gamma_i = \run(T_i)$ and output $\gamma_i$ as long as $\gamma_i\ne \bot$, $|\gamma_i|\ge \ell$, and $\gamma_{i}\ne \gamma_{i-1}$ (or $i=0$).

    It is easy to see that the query time is $\Oh(n/\Delta)=\Oh(n/(\ell+1-2p))$ as claimed, so it remains to show that the query algorithm is correct. We start with a simple claim:
    \begin{claim}\label{clm:runs-explicit}
    If a fragment $T_i$ is contained in $\gamma\in \RUNS_{\ell,p}(T)$, then $\gamma_i=\gamma$.
    \end{claim}
    \begin{proof}
        Observe that $\per(T_i)\le \per(\gamma)\le p = \frac12|T_i|$, so $T_i$ is periodic and $\gamma_i \ne \bot$.
        The periods of $\gamma$ and $\gamma_i$ are both at most $p$ and $|\gamma\cap \gamma_i| \ge 2p$, so \cref{fct:overlap} implies $\gamma_i=\gamma$.
    \end{proof}
    
    The correctness of our query algorithm consists of a few simple statements.

    \begin{description}
        \item[Every reported fragment belongs to $\RUNS_{\ell,p}(T)$.] We only report $\gamma_i=\run(T_i)$ if $\gamma_i\ne \bot$. In particular, $T_i$ is periodic, and thus $\per(\gamma_i)=\per(T_i)\le \frac12|T_i|=p$.
        We additionally check $|\gamma_i|\ge \ell$; these two conditions together guarantee that $\gamma_i \in \RUNS_{\ell,p}(T)$.
        \item[Every run $\gamma\in \RUNS_{\ell,p}(T)$ is reported.]
        Consider a run $\gamma = T[b\dd e)\in \RUNS_{\ell,p}(T)$ and define $i = \ceil{b/\Delta}$ so that $b\in ((i-1)\Delta \dd i\Delta]$.
        Observe that $e = b+|\gamma|\ge b+\ell \ge 1 + (i-1)\Delta + \ell = i\Delta +2p$.
        In particular, $i\Delta+2p \le e \le n$ implies $i \in [0\dd \floor{(n-2p)/\Delta}]$.
        Moreover, $b \le i\Delta < i\Delta+2p \le e$ means that $T_i$ is contained in $\gamma$,
        and thus $\gamma=\gamma_i$ follows by \cref{clm:runs-explicit}.
        Furthermore, $|\gamma|\ge \ell$ since  $\gamma\in \RUNS_{\ell,p}(T)$ and $\gamma\ne \gamma_{i-1}$ since $\gamma$ starts at a position $b > (i-1)\Delta$.
        Thus, we conclude that the algorithm reports $\gamma_i=\gamma$. 
        \item[No run $\gamma\in \RUNS_{\ell,p}(T)$ is reported more than once.] For a proof by contradiction, suppose that $\gamma$ is reported both as $\gamma_{j}$ and $\gamma_{i}$ for some $i<j$. 
        This means that $T_{j}$ and $T_i$ are both contained in $\gamma$ and, in particular, $T_{i-1}$ is also contained in $\gamma$. 
        By \cref{clm:runs-explicit}, we conclude that $\gamma_{i-1}=\gamma = \gamma_{i}$, which is a contradiction since we only report $\gamma_i$ if $\gamma_{i}\ne \gamma_{i-1}$.
        \item[The runs are simultaneously ordered by their start and end positions.] 
        Suppose that two distinct reported runs $\gamma_j=T[b_j\dd e_j)$ and $\gamma_i=T[b_i \dd e_i)$ for $j<i$ violate the condition, i.e., $b_i \le b_j$ or $e_i \le e_j$. 
        In the former case, $T_j$ is contained in $\gamma_i$, whereas in the latter case, $T_i$ is contained in $\gamma_j$.
        In both cases, \cref{clm:runs-explicit} implies $\gamma_i=\gamma_j$, which is a contradiction.\qedhere
    \end{description}
\end{proof}

\begin{lemma}\label{lem:runs-vs-r}
    For a text $T$ of length $n$ and integers $\ell,p\in [1\dd n]$, define 
    \[\R_{\ell,p}(T) = \{i\in [0\dd n-\ell]: \per(T[i\dd i+\ell))\le p\}.\]
    If $\ell \ge 2p$, then
    \[\RUNS_{\ell,p}(T) =\{ T[b\dd e) : [b\dd e-\ell]\text{ is a maximal interval contained in }\R_{\ell,p}(T)\}.\]
\end{lemma}
\begin{proof}
    Consider $\gamma = T[b\dd e)\in \RUNS_{\ell,p}(T)$.
    For each $i\in [b\dd e-\ell]$, the fragment $T[i\dd i+\ell)$ is contained in $\gamma$, and thus $\per(T[i\dd i+\ell))\le \per(\gamma)\le p$.
    Consequently, $[b\dd e-\ell]\subseteq \R_{\ell,p}(T)$.

    Next, for a proof by contradiction, suppose that $b-1\in \R_{\ell,p}(T)$.
    This means that $\per(T[b-1\dd b+\ell-1))\le p$ and thus $\gamma'=\run(T[b-1\dd b+\ell-1))\in \RUNS_{\ell,p}(T)$. 
    The fragment $T[b\dd b+\ell-1)$ is contained in both $\gamma$ and $\gamma'$, so the intersection of these two runs consists of at least $\ell-1 \ge 2p-1 \ge \per(\gamma)+\per(\gamma')-\gcd(\per(\gamma),\per(\gamma'))$ positions. 
    This contradicts \cref{fct:overlap}, and thus $b-1\notin \R_{\ell,p}(T)$.
    A symmetric argument shows that $e-\ell+1\notin \R_{\ell,p}(T)$.

    It remains to prove that every maximal interval $[i\dd j]\subseteq \R_{\ell,p}(T)$ corresponds to a run $T[i\dd j+\ell)\in \RUNS_{\ell,p}(T)$.
    Since $i\in \R_{\ell,p}(T)$, we have $\per(T[i\dd i+\ell))\le p$, and thus $\gamma = \run(T[i\dd i+\ell))\in \RUNS_{\ell,p}(T)$. 
    Denote $\gamma = T[b\dd e)$; as previously shown, $[b\dd e-\ell]$ is a maximal interval contained in $\R_{\ell,p}(T)$.
    Since $i\in [b\dd e-\ell]$ holds due to $b\le i $ and $e \ge i+\ell$ and since $[i\dd j]$ is also a maximal interval contained in $\R_{\ell,p}(T)$, we conclude that $[i\dd j]=[b\dd e-\ell]$, i.e., $T[b\dd e)=T[i\dd j+\ell)$ holds as claimed.    
\end{proof}

\begin{proposition}\label{prp:runs-bitmask}
    Given a string $T\in [0\dd \sigma)^n$ and parameters $\ell,p \in [0\dd \floor{(\log_\sigma n)/8}]$,
    the bitmask $\R_{\ell,p}(T)$ can be constructed in $\Oh(n/\log_\sigma n)$ time. 
\end{proposition}
\begin{proof}
    Observe that, in $\tOh(\sqrt{n})$ time, we can construct a bitmask of length $\floor{\sqrt{n}}$ whose set bits are $\{\int(R) : R\in [0\dd \sigma)^\ell \text{ and } \per(R) \le p\}$.
    This gives us constant-time oracle access to the set $\{R\in [0\dd \sigma)^\ell : \per(R) \le p\}$.
    Now, \cref{lem:bitmask} lets us construct in  $\Oh(n/\log_\sigma n)$ time a bitmask representing $\R_{\ell,p}(T)=\{i \in [0\dd n-\ell] : \per(T[i\dd i+\ell))\le p\}$.
    The overall running time is $\tOh(\sqrt{n})+\Oh(n/\log_\sigma n)=\Oh(n/\log_\sigma n)$.
\end{proof}

\section{String Synchronizing Sets in Sublinear Time}\label{sec:sync}

\begin{definition}[Synchronizing set~\cite{KK19}]\label{def:sync}
    For a string $T[0\dd n)$ and parameter $\tau\in [1\dd \floor{n/2}]$, a set $\Sync \subseteq [0\dd n-2\tau]$ is a \emph{$\tau$-synchronizing set} of $T$ if it satisfies the following two conditions:
    \begin{description}
        \item[Consistency:] $\ $
        
        For $i,j\in [0\dd n-2\tau]$, if $i\in \Sync$ and $T[i\dd i+2\tau)\eqmatch T[j\dd j+2\tau)$, then $j\in \Sync$.
        \item[Density:] $\ $
        
        For $i\in [0\dd n-3\tau+1]$, we have $[i\dd i+\tau)\cap \Sync = \emptyset$ if and only if $\per(T[i\dd i+3\tau-1)) \le \frac13\tau$. 
    \end{description}
\end{definition}

\begin{proposition}[{\cite[Construction 3.5 and Lemma 5.2]{DBLP:journals/siamcomp/KociumakaRRW24}}]\label{prp:sync}
For a string $T[0\dd n)$, let $(\B_k)_{k\in \Zz}$ be a descending chain satisfying \cref{prop:bk_properties}.
Consider a parameter $\tau\in [1\dd \floor{n/2}]$.

Define a set $\Sync\subseteq [0\dd n-2\tau]$ so that a position $i\in [0\dd n-2\tau]$ belongs to $\Sync$ if and only if $\per(T[i\dd i+2\tau))>\tfrac13\tau$ and at least one of the following conditions holds:
\smallskip
\begin{itemize}
    \item $i+\tau \in \B_{k(\tau)}$, where $k(\tau) = \max\{ j \in \mathbb Z_{\geq 0} \mid j = 0 \textnormal{ or } 16\lambda_{j - 1} \leq \tau \}$;
    \item there exists a $\tau$-run $T[b\dd e)\in \RUNS_\tau(T)$ such that $b = i+1$; or
    \item there exists a $\tau$-run $T[b\dd e)\in \RUNS_\tau(T)$ such that $e = i+2\tau-1$.
\end{itemize}
\smallskip
Then, $\Sync$ is a $\tau$-synchronizing set of size $|\Sync| < \frac{70n}{\tau}$.
\end{proposition}

\begin{theorem}[{Compare \cite[Theorem 1.13]{DBLP:journals/siamcomp/KociumakaRRW24}}]\label{thm:ss-explicit}
    A string $T\in [0\dd \sigma)^n$ can be preprocessed in $\Oh(n/\log_\sigma n)$ time so that, given $\tau \in [1\dd \floor{n/2}]$, a $\tau$-synchronizing set $\Sync$ of $T$ of size $|\Sync| < \frac{70n}{\tau}$ can be constructed in $\Oh(\frac{n}{\tau})$ time.
    Moreover, $\per(T[i\dd i+2\tau))>\frac13\tau$ holds for every $i\in \Sync$.
\end{theorem}
\begin{proof}
We build $\Sync$ based on the construction specified in \cref{prp:sync}.
At preprocessing time, we build an array of intervals $I[0\dd k(n)]$ with $I[k] = \{\tau : k(\tau)=k\}$.
Note that $I[0]=[0\dd \ceil{16\lambda_0})$ and $I[k]=[\ceil{16\lambda_{k-1}}\dd \ceil{16\lambda_k})$ for $k\in [1\dd k(n)]$.
Since $k(n)=\Oh(\lg n)$ and the values $\lambda_j$ can be computed in $\Oh(1)$ time each (using constant-time arithmetic operations on $\Oh(\lg n)$-bit integers), this preprocessing takes $\Oh(\lg n)$ time.
We also apply the $\Oh(n/\log_\sigma n)$-time preprocessing of \cref{lem:bk-explicit,prp:runs-explicit}.

Given $\tau$, we first compute $k(\tau)$ by scanning $I[0\dd k(n)]$ from the top down to find the largest $k$ with $\tau \in I[k]$.
This takes $1+k(n)-k(\tau) = 1 + \Oh(\lg\frac{n}{\tau})= \Oh(\frac{n}{\tau})$ time.
Next, we apply \cref{lem:bk-explicit} to generate $\B_{k(\tau)}$ in $\Oh(1+n/\lambda_{k(\tau)})=\Oh(n/\tau)$ time, where the definition of $k(\tau)$ implies $16\lambda_{k(\tau)}>\tau$.
Then, we use \cref{prp:runs-explicit} to compute $\RUNS_{\tau}(T)$, with runs simultaneously ordered by their start and end positions; this also takes $\Oh(n/\tau)$ time.
Finally, we proceed exactly as in the proof of \cite[Theorem 1.13]{DBLP:journals/siamcomp/KociumakaRRW24}, where it is shown how to derive $\Sync$ from $\B_{k(\tau)}$ and $\RUNS_{\tau}(T)$ in~$\Oh(n/\tau)$~time.
\end{proof}

\begin{theorem}\label{thm:ss-bitmask}
    Given a string $T\in [0\dd \sigma)^n$ and a parameter $\tau \in [1\dd \floor{n/2}]$, a $\tau$-synchronizing set $\Sync$ of $T$ of size $|\Sync| < \frac{70n}{\tau}$, represented as a bitmask, can be constructed in $\Oh(n/\log_\sigma n)$ time.
    Moreover, $\per(T[i\dd i+2\tau))>\frac13\tau$ holds for every $i\in \Sync$.
\end{theorem}
\begin{proof}
If $\tau \ge (\log_\sigma n)/16$, we apply \cref{thm:ss-explicit} with preprocessing followed by a single query.
In $\Oh(n/\log_\sigma n + n/\tau)=\Oh(n/\log_\sigma n)$ time, this yields the explicit representation of the set $\Sync$ satisfying the desired conditions.
We convert this set into a bitmask by starting with an all-0 bitmask and setting the $i$-th bit for every $i\in \Sync$. 
This conversion also takes $\Oh(n/\log_\sigma n + n/\tau)=\Oh(n/\log_\sigma n)$ time.

In the complementary case of $\tau < (\log_\sigma n)/16$, we first compute $k(\tau)$ in $\Oh(\lg n)$ time by naively iterating over all possibilities.
Then, we apply \cref{lem:bk-bitmask} to compute a bitmask representing $\B_{k(\tau)}$
and \cref{prp:runs-bitmask} to compute bitmasks representing $\R_{2\tau,\floor{\tau/3}}(T)$ and $\R_{\tau,\floor{\tau/3}}(T)$; both subroutines take $\Oh(n/\log_\sigma n)$ time.
According to \cref{prp:sync}, a synchronizing set satisfying the desired conditions can be obtained by including a position $i$ if and only if $\per(T[i\dd i+2\tau))> \tau/3$ (which is equivalent to $i\notin \R_{2\tau,\floor{\tau/3}}(T)$)
and at least one of the following conditions hold:
\smallskip
\begin{itemize}
    \item $i+\tau \in \B_{k(\tau)}$,
    \item there exists a $\tau$-run $T[b\dd e)\in \RUNS_\tau(T)$ such that $b = i+1$ (which, by \cref{lem:runs-vs-r}, is equivalent to $i\notin \R_{\tau,\floor{\tau/3}}(T)$ and $i+1\in \R_{\tau,\floor{\tau/3}}(T)$),
    \item there exists a $\tau$-run $T[b\dd e)\in \RUNS_\tau(T)$ such that $e = i+2\tau-1$ (which, by \cref{lem:runs-vs-r}, is equivalent to $i+\tau\notin \R_{\tau,\floor{\tau/3}}(T)$ and $i+\tau-1\in \R_{\tau,\floor{\tau/3}}(T)$).
\end{itemize}
\smallskip
Consequently, the bitmask representing $\Sync$ can be computed using $\Oh(1)$ bitwise operations (AND, OR, NOT) applied to the following components: $\R_{2\tau,\lfloor \tau/3 \rfloor}$; the mask $\B_{k(\tau)}$ shifted by $\tau$ positions and padded with $0$-bits; and four copies of $\R_{\tau,\lfloor \tau/3 \rfloor}$, shifted and padded as follows—by $0$ positions with $1$-bits, by $1$ position with $0$-bits, by $\tau-1$ positions with $0$-bits, and by $\tau$ positions with $1$-bits.
All these bit-wise operations, including shifts and padding, can be implemented in $\Oh(n/\lg n)$ time, giving a total time of $\Oh(n/\log_\sigma n)$.
\end{proof}

\section{Improving the Query Time: Overview}\label{sec:faster}

In the remainder of the paper, we further improve the query time for synchronizing set construction from $\Oh(\frac n\tau)$ to $\Oh(\frac {n \lg \tau}{\tau \lg n})$, while retaining the $\Oh(n / \log_\sigma n)$ preprocessing time. 
We only outline the main algorithmic ideas here, with proofs of all claims in Appendices~\ref{appendix:parsing}--\ref{sec:appendixsynch}.

The synchronizing set will be output in a representation of size $\Oh(\frac{n \lg \tau}\tau)$ bits. 
We show that this is optimal in the lemma below, which
models the query algorithm as a function mapping an input string $T$ into an encoding $\mathsf{E}(T)$ of a $\tau$-synchronizing set $\Sync$ of $T$ such that $\Sync$ can be recovered from $\mathsf{E}(T)$ alone using an accompanying decoding function $\mathsf{D}$.

\begin{lemma}\label{lem:opt}
    Consider integers $n,\tau \in \Zp$ such that $n \ge 3\tau$, as well as a pair of functions $\mathsf{E}: \{0,1\}^n \to \{0,1\}^*$ and $\mathsf{D} : \{0,1\}^* \to 2^{[0\dd n)}$ jointly satisfying the following property:
    \begin{center}
        For every $T\in \{0,1\}^n$, the set $\mathsf{D}(\mathsf{E}(T))$ is a $\tau$-synchronizing set of $T$.
    \end{center}
    Then, there exists $T\in \{0,1\}^n$ such that $|\mathsf{E}(T)| =\Omega(\frac{n}\tau \lg \tau)$.
\end{lemma}

\begin{proof}
    Let $k = \floor{\frac{n}{3\tau}}$. For an arbitrary string $S \in [0\dd \tau)^k$, we construct $T \in \{0, 1\}^n$ defined by $\forall_{i \in [0\dd k)}\; T[3\tau i\dd 3\tau i + 3\tau) := \texttt{0}^{2\tau + S[i] - 1}\cdot\texttt1\cdot\texttt{0}^{\tau - S[i]}$, with an arbitrary suffix $T[3\tau k\dd n)$. For any $\tau$-synchronizing set, the first synchronizing position in $[3\tau i\dd 3\tau i + 3\tau)$ is $3\tau i + S[i]$. This holds due to \cref{prp:sync} and the fact that $T[3\tau i\dd 3\tau i + 2\tau + S[i] - 1)$ is a suffix of a $\tau$-run. Hence, given any $\tau$-synchronizing set of $T$, we can restore~$S$. Consequently, over all the $S \in [0\dd \tau)^k$, the average size of the encoded set is at least $k \lg \tau = \Omega(\frac n\tau \lg \tau)$ bits.
\end{proof}

\subparagraph{Sparse encodings.} Our representation of the constructed $\tau$-synchronizing set is its bitmask in a special \emph{sparse encoding}. During the construction, we also encode intermediate integer arrays, and hence we give a more general definition. The encoding is based on Elias-$\gamma$ codes:

\begin{restatable}[Elias-$\gamma$ code~\cite{Eli75}]{definition}{restatedefgamma}\label{def:elias}
The \emph{Elias-$\gamma$ code} of a positive integer $x\in \mathbb Z_+$ is a bitmask $\gamma(x):= \mathtt{0}^{\ell} \cdot X[0\dd \ell]$, where $\ell = \floor{\lg x}$ and $X[0\dd \ell]$ is the $(\ell + 1)$-bit binary representation of~$x$ (with most significant bit $X[0]=\mathtt{1}$).
\end{restatable}

\begin{restatable}[Sparse encoding]{definition}{restatedefsparse}\label{def:sparse}
    The \emph{sparse encoding} of a string $A \in \Zz^n$ is a bitmask $\senc{A}$ that encodes $A$ from left to right as follows. Each symbol $u \in \Zp$ is stored as a \emph{literal token} $\mathtt{1} \cdot \gamma(u)$.
    Each inclusion-wise maximal fragment of the form $0^x$ for $x\in \mathbb Z_+$ is stored as a \emph{zero-run token} $\mathtt{0} \cdot \gamma(x)$.
        \end{restatable}

\begin{restatable}{example}{restateexamplesparse}%
\edef\encexample{0/3,1/3,0/2,1/5,0/7,1/9,1/1,0/2}%
Consider string $A = \texttt{%
\foreach \indicatorbit/\x in \encexample {%
    \ifnum\indicatorbit>0 \x\else
    \foreach \ctr in {1,...,\x} {0}\fi
}%
}%
=%
\foreach[count=\i from 0] \indicatorbit/\x in \encexample {%
    \ifnum\i>0 \cdot\fi
    \ifnum\indicatorbit>0 \texttt{\x}\else
    \texttt0^{\x}\fi
}$ and its sparse encoding shown below. The first bit of each token indicates whether it is a literal~$x$ or a zero-run $\texttt0^x$.
The rest of the token consists of $\floor{\lg x}$ zeros (to the left of each dotted line) and the binary representation of $x$ (to the right of each dotted line).

\medskip\centering

\noindent\begin{tikzpicture}[x=1em, y=1em]
    \tikzset{bitnode/.style={inner sep=0pt, minimum width=.5em, minimum height=1.25em, right=0 of tmp, draw=white}}
    \tikzset{indicatorbitnode/.style={bitnode, minimum width=.8em, draw=gray}}
    \tikzset{markernode/.style={inner sep=0pt, minimum width=.5em, minimum height=.5em}}
    \tikzset{tokenspacer/.style={bitnode, minimum width=0pt, inner sep=0pt, right=.4em of tmp}}
    \tikzset{indicatorspacer/.style={tokenspacer, right=.1em of tmp}}
    \tikzset{splitspacer/.style={tokenspacer, right=.3em of tmp}}

    \node (tmp) {};
    \node[left=-1em of tmp] {$\senc{A}\ =\ $};
    
    \foreach[
        evaluate=\x as \floorlogx using {int(floor(ln(\x)/ln(2)))}, 
        evaluate=\floorlogx as \bitlen using {int(2*\floorlogx + 1)},
    ] \indicatorbit/\x in \encexample {
        \node[tokenspacer] (tmp) {};
        \node[indicatorbitnode] (tmp) {\clap{\texttt\indicatorbit}};
        \node[markernode] (l) at (tmp.south west) {};
        
        \ifnum\floorlogx>0 
        \node[indicatorspacer] (tmp) {};
        \foreach \ctr in {1,...,\floorlogx} {
            \node[bitnode] (tmp) {\clap{\texttt0}};
        }
        \fi

        \node (tmp2) at (tmp.east) {};
        \node[splitspacer] (tmp) {};
        \path (tmp2.center) to node[midway, minimum height=1.5em] (tmp2) {} (tmp.center);
        \draw[densely dotted, thick, gray] (tmp2.north) to (tmp2.south);  
        
        \node[bitnode] (tmp) {\clap{\texttt1}};

        \ifnum\floorlogx>0
        \foreach[evaluate=\ctr as \nextbit using {int(mod(floor(\x / pow(2, \ctr - 1)), 2))}] \ctr in {\floorlogx,...,1} {
            \node[bitnode] (tmp) {\clap{\texttt\nextbit}};
        }
        \fi
        
        \node[markernode] (r) at (tmp.south east) {};
        \draw[thick, decorate, decoration={brace, amplitude=.25em}] (r.south west) -- node[midway, below=.25em] {%
            \ifnum\indicatorbit>0 \texttt{\x}\else
            $\texttt0^{\x}$\fi$\strut$
        } (l.south east);
    }
\end{tikzpicture}
\vspace{-.2cm}
\end{restatable}

We note that a zero-run $\texttt0^x$ or a literal $x > 0$ contributes $2\floor{\lg x} + 2$ bits to the encoding. 
By applying Jensen's inequality to the concave log-function, we get the following bounds.

\begin{restatable}{observation}{restateobssparsesize}\label{obs:sparse_size}
    For $A\in\Zz^n$, let $i_1 < i_2 < \dots < i_a$ be the elements of $\{ i \in [0\dd n) \mid A[i] \neq 0 \}$, and let $i_0 = -1$ and $i_{a + 1} = n$.
    Regarding $\absolute{\senc{A}}$, we observe:
    
    \begin{itemize}
        \item The number of bits contributed by zero-run tokens is at most
        \[\textstyle 2a + 2 + \sum_{j = 0}^{a}2\floor{\lg(i_{j+1} - i_j)} \in \Oh((a + 1) \cdot \lg \frac{n + 1}{a + 1})
        \subseteq \Oh(n).\]
        \item The number of bits contributed by literal tokens is exactly
        \[\textstyle 2a + \sum_{j = 1}^{a}2\floor{\lg A[i_j]} \in \Oh(\sum_{j=0}^{n - 1} \lg(1 + A[j]))\cap \Oh(a \cdot \lg (\frac{1}{a}\sum_{j=1}^{a} A[i_j])).\]
        \item If $\sum_{i = 0}^{n - 1}A[i] \in \Oh(n)$, then $\absolute{\senc{A}} \in \Oh((a + 1) \cdot \lg \frac{n + 1}{a + 1})$.
    \end{itemize}
    
\end{restatable}

\smallskip

If we encode the bitmask of a synchronizing set from \cref{thm:ss-bitmask}, then the encoding consists of
$\Oh(\absolute{\Sync} \cdot \lg \frac{n}{\absolute{\Sync}}) = \Oh(\frac n\tau \lg \tau)$ bits. It is easy to see that the size of the encoding is minimized when the string $A$ is all-zero, in which case it consists of $2\floor{\lg n} + 2$ bits.

\begin{restatable}{observation}{restateobssparsesizelower}\label{obs:sparse_size_lower}
    For every $n \in \mathbb Z_{+}$ and $A\in \mathbb Z_{\geq_0}^n$, it holds $2\floor{\lg n} < \absoluteenc{A}$.
    \end{restatable}

If a string $A$ is a prefix of another string $A'$, then the sparse encoding of $A$ is also a prefix of the sparse encoding of $A'$, unless $A$ ends with a zero-run that can be extended further in $A'$. In the latter case, the zero-run is encoded differently in $A$ and $A'$. Since the set of Elias-$\gamma$ codes (and hence the set of possible tokens) is prefix-free, this implies that the encoding of $A$ cannot be a prefix of the encoding of $A'$. 

\begin{restatable}{observation}{restateEncPrefOfOther}\label{obs:encoding_prefix_of_other}
The following holds for every two distinct non-empty strings $A,A'\in \Zz^+$.
The encoding $\senc{A}$ is a prefix of the encoding $\senc{A'}$ if and only if $A$ is a prefix of $A'$ and $A'[\absolute{A}-1\dd \absolute{A}]\ne 00$.
\end{restatable}

\subparagraph{Parsing sparse encodings.}

\newcommand{\crefbullet}[1]{\cref{#1_host}(\ref{#1_bullet})}

We introduce basic tools for sparse encodings.
Computing Elias-$\gamma$ codes can be accelerated with precomputed lookup tables.
Then, we can encode and decode strings one token at a time, resulting in \cref{lem:encode_sparse_naive_host} below. 
For \crefbullet{lem:encode_ds}, we merely split the computation time of \crefbullet{lem:encode_sparse_naive} into preprocessing and query time.

\begin{restatable}{lemma}{restateeliasencodedecode}\label{lem:encode_elias}\label{lem:decode_elias}
    For every $N \in [2\dd 2^w]$, after an $\Oh(N)$-time preprocessing, the following holds
    for every $u\in {{\Oh(w)}}$ and $x\in [1\dd 2^u)$.
    Computing $B:=\gamma(x)$ from the $u$-bit representation of $x$ and computing the $u$-bit representation of $x$ from a bitmask with prefix $B$ can be done in 
    $\Oh(1 + \lg x / \lg N)$ time. The size $\absolute{B} = 2\floor{\lg x} + 1$ of the code is also reported.
\end{restatable}

\restatespacefix

\begin{restatable}{lemma}{restateEncodeDecodeDS}\label{lem:encode_sparse_naive_host}\label{lem:decode_sparse_naive_host}\label{lem:encode_ds_host}
    For every $N \in [2\dd 2^w]$, after an $\Oh(N)$-time preprocessing, the following holds for every $u \in \Oh(w)$.
    Let $A\in[0\dd 2^u)^n$ with $n \in 2^{\Oh(w)}$.
    \begin{enumerate}[(i)]
        \item Computing the $u$-bit representation of $A$ from $\senc{A}$ and vice versa can be done in $\Oh(n + \absolute{\senc{A}} / \lg N)$ time.\label{lem:encode_sparse_naive_bullet}\label{lem:decode_sparse_naive_bullet}
        \item If $A$ is given in $u$-bit representation, then in $\Oh(n)$ time and space, one can compute a data structure that returns $\senc{A}$ in $\Oh(1 + \absolute{\senc{A}} / \lg N)$ time.\label{lem:encode_ds_bullet}
    \end{enumerate}
    
\end{restatable}

We can encode and decode sparse arrays (in list representation) in constant time per non-zero entry, using \cref{lem:encode_elias} to process one token at a time.

\begin{corollary}\label{lem:encode_list}\label{lem:decode_list}
    After an $\Oh(\sqrt{n})$-time preprocessing, given any array $A \in [0\dd n^{\Oh(1)}]^n$ as a list of its $a \in [0\dd n]$ non-zero entries as position-value pairs in increasing order of position, we can output $\senc{A}$ and vice versa in $\Oh(1 + a)$ time.
\end{corollary}

\subsection{Processing Sparse Encodings With Transducers}

We will repeatedly process bitmasks, strings, and arrays in the sparse encoding.
As a general tool for this task, we propose a preprocessing scheme for \emph{deterministic finite-state transducers}, henceforth simply called transducers. 
\def\transducerdef{A transducer consists of a finite set $Q$ of states, an initial state $s_0 \in Q$, a finite alphabet $\Sigma$, and a transition function $\delta : Q \times \Sigma^{t} \rightarrow Q \times \Sigma$.
Given input strings $S_1, \dots, S_t$, the transducer produces output string $T$, where all strings are of common length $n$ and are over alphabet $\Sigma$. The computation is performed in a sequence of $n$ steps. Before the $i$-th step (with $i \in [0\dd n)$), the transducer is in state $s_i$, and it has already written $T[0\dd i)$. The $i$-th step is performed by evaluating $\delta(s_i, S_1[i], \dots, S_t[i]) = (s_{i+1},T[i])$, resulting in the new state and the next symbol of the output string.
Since the input strings are processed from left to right, we say that the transducer has $t$ \emph{input streams}, and we call it a \emph{single-stream transducer} if $t = 1$.}%
\transducerdef

In the \lcnamecref{lem:accelerate_single_stream} below, we show how to preprocess a single-stream transducer so that it can efficiently work directly on sparse encodings.
For every state $s$ and every string $S'$ for which $\absolute{\senc{S'}}$ is at most a small fraction of $\lg N$, we precompute the entire chain of transitions performed when reading $S'$ in state $s$. 
The result of this precomputation is the new state reached after reading $S'$, as well as the sparse encoding of the produced output. This way, for a longer input encoding, we can process up to $\Omega(\lg N)$ bits in constant time. 
\begin{restatable}{theorem}{restatesingletape}\label{lem:accelerate_single_stream}
    Consider a single-stream transducer over alphabet $[0\dd \sigma)$ with states $[0\dd q)$, where ${\sigma, q \in 2^{\Oh(w)}}$.
    For every $N \in [2 \dd 2^w]$, if evaluating the transition function with input symbol $x$ (at any state) takes $\Oh(1 + \lg(1 + x) / \lg N)$ time, then after an $\Oh(qN)$-time preprocessing, the following holds.
    
    If $S$ is an input string of length $n \in 2^{\Oh(w)}$ for which the transducer produces output~$T$, then $\senc{T}$ can be computed from $\senc{S}$ in ${\Oh(1 + {(\absoluteenc{S} + \absoluteenc{T}) / \lg N})}$~time.
\end{restatable}
\begin{proofsketch}
We follow the full proof in~\cref{sec:lem:accelerate_single_stream}. Fix a small constant $\eps>0$ and set $M=\Theta(N^\eps)$ to be a power of two (so $\lg M=\Theta(\lg N)$).
The proof has two parts: (i) tabulate how the transducer behaves on \emph{all} sparse encodings whose bit-length is $\le \lg M$, and (ii) process the remaining long tokens (especially long zero-runs) with a separate jumping mechanism.

\subparagraph*{Preprocessing.}
We precompute information that lets us process $\Omega(\lg M)$ bits of the input encoding in $\Oh(1)$ time.
\begin{enumerate}
    \item \emph{Tabulating short prefixes of sparse encodings.}
    For every state $s\in[0\dd q)$ we build a lookup table $L_s$ indexed by all bitmasks $B\in\{0,1\}^{\lg M}$.
    For each $B$, we find the \emph{largest} prefix length $b\le \lg M$ such that $B[0\dd b)$ is a sparse encoding.
        Let $A$ be the decoded length-$a$ string with $B[0\dd b)=\senc{A}$.
    We simulate the transducer for $a$ steps on input $A$, obtaining the new state $s'$ and the produced output string $A'$.
    We then store in $L_s[\int(B)]$ (besides $b,a,s'$) a representation of $\senc{A'}$ that separates \emph{leading} and \emph{trailing} zeros:
    we keep the counts $z_1,z_2$ of leading/trailing zeros of $A'$, and build the data structure $D$ from \crefbullet{lem:encode_ds} for the middle part $A'[z_1\dd a-z_2)$.
    This way, later we can append $\senc{A'}$ to the global output in time $\Oh(1+\absolute{\senc{A'}}/\lg M)$, without ever materializing it bit-by-bit.
    Computing a single entry takes $\Oh(1+a)$ time (decode $A$, simulate $a$ transitions, and build $D$). 
    Since $\absolute{\senc{A}}\le \lg M$ implies $a=\Oh(M)$ by \cref{obs:sparse_size_lower}, the total for all $qM$ entries is $\Oh(qM^2)$.
    \item \emph{Short runs of input zeros.}
    For each state $s$, we also build a table that stores, for $y\in [0\dd M^{1/4}]$, the value $L'_s[y]$ obtained by looking up $L_s$ on the $\lg M$-bit word that starts with the sparse encoding of the zero-run token $\senc{0^y}$.
    This allows us to advance through $y$ consecutive input zeros in one constant-time lookup.
    The key point is that, for $y\le M^{1/4}$, the token $\senc{0^y}$ has length $<\lg M$, so it fits into one lookup word together with padding.
    If the tables $L_s$ are given, then the tables $L'_s$ can be easily constructed in $\Oh(qM)$ time.
    \item \emph{Long stretches where input/output are both zero.}
    Finally, for handling \emph{very long} zero-run tokens, we build the directed graph on states in which there is an edge $s\to s'$ iff the transducer transitions from $s$ to $s'$ on input $0$ \emph{and} outputs $0$.
    This is a pseudoforest (outdegree $\le 1$), so we can preprocess it with a simple adaptation of level ancestor queries~\cite{DBLP:journals/tcs/BenderF04} such that we can “jump” in $\Oh(1)$ time through long chains of transitions that keep outputting zero.
    Computing the graph and auxiliary data structure takes $\Oh(q)$ time.
\end{enumerate}

\subparagraph*{Running the transducer on \boldmath$\senc{S}$.\unboldmath}
Let $X:=\senc{S}$.
As an invariant, we maintain: offsets $x \in [0\dd \absolute{X}]$ and $n \in [0\dd \absolute{S}]$ such that $X[0\dd x)=\senc{S[0\dd n)}$, the current state $s$ reached by the transducer after processing $S[0\dd n)$, and an output representation consisting of
(i) a counter $z$ that indicates the number of trailing zeros of $T[0\dd n)$ and
(ii) the sparse encoding $Y=\senc{T[0\dd n-z)}$. 
Here, $Y$ is always a prefix of $\senc{T}$ due to \cref{obs:encoding_prefix_of_other}.
We process the remaining suffix $X[x\dd]$ in \emph{macro-steps}:
\begin{itemize}
    \item If $L_s[X[x\dd x+\lg M)]$ returns $b>0$, then we have a whole sparse encoding $X[x\dd x+b)=\senc{A}$ fitting in $\lg M$ bits.
    We take the precomputed entry $(b,a,s',z_1,z_2,D)$ and update $(x,n,s)$ by $(x+b,n+a,s')$.
    If the corresponding output $A'$ is all-zero, we just increase $z$.
    Otherwise we first flush the pending zeros by appending $\senc{0^{z+z_1}}$ (if needed), then append the middle part via $D$, and finally set $z\gets z_2$.
    Near the end of $X$, we pad the $\lg M$-bit window with an incomplete token to prevent decoding beyond $\absolute{X}$.
    \item Otherwise $b=0$, i.e., the next token of $X[x\dd]$ is \emph{long} (more than $\lg M$ bits).
    We decode this token (and its bit length $b'$) using \cref{lem:decode_elias}.
    If it is a literal, we decode the symbol, evaluate one transducer transition (cost $\Oh(1+\lg(1+x)/\lg N)$ by assumption), and append the encoded output symbol using \cref{lem:encode_elias}.
    If, however, the decoded token is a zero-run token $\senc{0^y}$, then we handle it in a series of \emph{micro-steps}.
    We alternate between (a) \emph{fixed} micro-steps that advance by $M'=\min(y,M^{1/4})$ using $L'_s[M']$ (so we can also handle prefixes whose \emph{output} is not all-zero), and (b) \emph{flexible} micro-steps that use the precomputed graph and auxiliary data structure to skip a maximal prefix of transitions that read $0$ and output $0$ (advancing $n$ and increasing $z$).
    This continues until all $y$ input zeros are consumed.
    Finally, regardless of whether the decoded token is a literal or a zero-run token, we advance $x$ by $b'$ bits.
\end{itemize}
Once we have processed the entire $X$, we append $\senc{0^z}$ to $Y$, obtaining $\senc{T}$.

\subparagraph*{Correctness.}
By construction, each lookup-table entry faithfully simulates the transducer on the decoded prefix it represents, so using the tables updates the current state exactly as in the real run and produces exactly the corresponding output substring.
The maintained invariant ensures that after every step, the processed part of $X$ corresponds to the processed prefix of the plain input, and $Y$ represents the output written so far up to the buffered trailing zeros.
The only subtlety is concatenating outputs while preserving a valid sparse encoding; this is precisely why we keep trailing zeros in a counter $z$ and only flush them when the next chunk produces a non-zero output, using the prefix characterization in \cref{obs:encoding_prefix_of_other}.

\subparagraph*{Time bound.}
Each macro-step with $b>0$ advances the input encoding by a \emph{maximal} sparse-encoding prefix inside a $\lg M$-bit window; hence, any two consecutive macro-steps advance the input encoding by more than $\lg M$ bits (as otherwise, the steps could have been merged), and there are $\Oh(\absoluteenc{S}/\lg M)$ macro-steps.

Micro-steps only occur inside long zero-run tokens. Every flexible micro-step advances the output to the next non-zero symbol.
Therefore, for four consecutive micro-steps of the form (flexible, fixed, flexible, fixed), 
the first output symbol produced by each of the two fixed steps is a non-zero symbol.
Hence, the encoding of the output produced by the central two steps is a substring of the encoding appended to $Y$ during the four steps.
Also, by the definition of a fixed step, the (plain text) output corresponding to this substring is of length at least $M^{1/4}$.
By \cref{obs:sparse_size_lower}, its sparse encoding contains $\Omega(\lg M)$ bits, so (ignoring $\Oh(1)$ initial/final micro-steps per macro-step) we get $\Oh(\absoluteenc{T}/\lg M)$ micro-steps overall. 
The number of initial/final micro-steps is limited by the number $\Oh(\absoluteenc{S}/\lg M)$ of macro-steps.

Finally, whenever a step takes more than constant time, it is because, for some value $r$, we (i) use \crefbullet{lem:encode_ds} to obtain an encoding consisting of $r$ bits and append it to the output, (ii) use \cref{lem:encode_elias,lem:decode_elias} to decode (resp.\ encode) a token consisting of $r$ bits and advance the input (resp.\ output) by $r$ bits, or (iii) evaluate a transition on a large literal whose encoding consists of $r$ bits.
Each such cost is $\Oh(1+r/\lg M)$, and it can be charged to the $\Theta(r)$ bits advanced in $X$ or appended to $Y$, yielding total overhead $\Oh((\absoluteenc{S}+\absoluteenc{T})/\lg M)$.
This gives total time $\Oh(1+(\absoluteenc{S}+\absoluteenc{T})/\lg M)=\Oh(1+(\absoluteenc{S}+\absoluteenc{T})/\lg N)$.
\end{proofsketch}

To handle multi-stream transducers, we use a reduction of multi-stream transducers to single-stream ones. For this purpose, we define a zipped string.

\begin{restatable}{definition}{restatezippeddefinition}
Let $A_1, \dots, A_t$ with $t \geq 1$ be strings in $\Zz^n$. The string $\zip(A_1, \dots, A_t)$ of length $n$ has, for each $i \in [0\dd n)$, its $i$-th symbol defined by $\zip(A_1, \dots, A_t)[i] = 0$ if $\forall_{j \in [1 \dd t]}\; A_j[i] = 0$, and $\zip(A_1, \dots, A_t)[i] = \senc{A_1[i]A_2[i]\dots A_t[i]}$ otherwise.
\end{restatable}

We stress that this zipped string is not, in itself, a sparse encoding -- it is a string in which each non-zero symbol is the sparsely encoded concatenation of the corresponding symbols of the original strings. Crucially, when sparsely encoding the zipped sequence, there is only a constant factor overhead over the sparsely encoded original strings.

\begin{restatable}{lemma}{restatezippedsize}\label{lem:zipped_encoding_size}
Given a constant number of equal-length strings $A_1, \dots, A_t\in \Zz^n$, it holds $\textstyle \absoluteenc{\zip(A_1, \dots, A_t)} = \Oh(\sum_{j = 1}^t \absoluteenc{A_j})$.
\end{restatable}

Given the sparse encodings of the input strings, we first produce the sparse encoding of their zipped string. For the result below, we use similar techniques as in the proof of \cref{lem:accelerate_single_stream}, processing up to $\Omega(\lg N)$ bits of the input encodings in constant time by exploiting precomputed information. 
Unlike in the proof of \cref{lem:accelerate_single_stream}, a long synchronized run of zeros in the input strings \emph{always} leads to a run of zeros in the output, simplifying the algorithm to some extent. 
However, non-synchronized runs of zeros pose an additional challenge, as described in detail in \cref{sec:lem:fast_multi_zip}.

\begin{restatable}{theorem}{restatemultizip}\label{lem:fast_multi_zip}
For every $N \in [2 \dd 2^w]$, after an $\Oh(N)$-time preprocessing, the following holds.
If $A_1, \dots, A_t \in [0\dd \sigma)^n$ with $t = \Oh(1)$ and $n,\sigma \in {2^{\Oh(w)}}$,
then $\senc{\zip(A_1, \dots, A_t)}$ can be computed from $\senc{A_1}, \dots, \senc{A_t}$ in $\Oh(1 + \sum_{j = 1}^t \absoluteenc{A_j} / \lg N)$ time.
\end{restatable}

Finally, there is an obvious reduction from a transducer with $t$ input strings to a single-stream transducer that receives the zipped version of the $t$ input strings.
By applying \cref{lem:accelerate_single_stream} to this single-stream transducer, we obtain the following result.

\begin{restatable}{corollary}{restateacceleratemulti}\label{lem:accelerate_multi_stream}
	Consider a transducer over alphabet $[0\dd \sigma)$ with states $[0\dd q)$, where ${\sigma, q \in 2^{\Oh(w)}}$, and
	$t = \Oh(1)$ input streams.
    If the transition function can be evaluated in $\Oh(1)$ time, then, for any $N \in [2 \dd 2^w]$, after an $\Oh(qN)$-time preprocessing, the following holds.

    Let $S_1, \dots, S_{t}$ be input strings of common length $n \in 2^{\Oh(w)}$ for which the transducer produces output $T$. 
    Then $\senc{T}$ can be computed from $\senc{S_1}, \dots, \senc{S_{t}}$ in ${\Oh((\absoluteenc{T} + \sum_{i = 1}^t \absoluteenc{S_i}) / \lg N)}$ time.
\end{restatable}

\subsection{Faster Synchronizing Set Queries}
Using \cref{lem:accelerate_multi_stream}, we obtain sparsely encoded versions of the sets $\B_k$, of $\RUNS_{\tau, \floor{\tau/3}}(T)$ and $\RUNS_{2\tau, \floor{\tau/3}}(T)$, and finally of a $\tau$-synchronizing set. 
Ultimately, we derive the following.

\begin{restatable}{theorem}{restatefastersss}\label{thm:sss_sparse}
    A string $T\in [0\dd \sigma)^n$ can be preprocessed in $\Oh(n/\log_\sigma n)$ time so that, given $\tau \in [1\dd \floor{n/2}]$, a $\tau$-synchronizing set $\Sync$ of $T$ of size $|\Sync| < \frac{70n}{\tau}$ can be constructed in $\Oh(\frac{n \lg \tau}{\tau \lg n})$ time and $\Oh(\frac{n \lg \tau}{\tau})$ bits of space.
    The set is reported as $\senc{M}$ for $M\in \{0,1\}^n$ such that $M[i]=1 \Leftrightarrow i\in \Sync$.
    Moreover, $\per(T[i\dd i+2\tau))>\frac13\tau$ holds for every $i\in \Sync$.
\end{restatable}

\subsection{Adding Rank and Select Support}\label{sec:rankselect}

\cref{thm:sss_sparse} returns the synchronizing set in sparse encoding, which by itself does not allow fast random access. Hence, we develop support data structures that augment the encoding for fast rank and select queries.
For any set $S \subseteq \mathbb Z$ and $x \in \mathbb Z_{\geq 0}$, we define $\rank_{S}(x) = \absolute{\{ y \in S \mid y < x\}}$ and $\pred_{S}(x) = \max(\{ y \in S \mid y \leq x\} \cup \{ -\infty \})$.
For every $i \in [1\dd \absolute{S}]$, we define $\select_{S}(i) = \max\{ y \in S \mid \rank_{S}(y) < i\}$.
For $S \subseteq [0\dd U)$ and its characteristic bitmask $A[0\dd U)$, we may equivalently use subscript $A$ rather than $S$.

We can use a precomputed lookup table to greedily parse any encoding into pieces of size around $\lg N$ bits, and then use another table to answer rank and select queries with respect to any piece in constant time. This is formalized below.

\begin{restatable}{lemma}{restatedecomposeencoding}\label{lem:decompose_for_rank_select}
    For every $N \in [5\dd2^w]$ and every bitmask $A[0\dd n)$, there is a sequence $(i, p_i, e_i, r_i)_{i = 0}^{h}$ with $h = \Oh(\absolute{\senc{A}} / \lg N)$ satisfying the following properties. The entries are defined by $p_0 = e_0 = 0$, $p_h = n$, $e_h = \absolute{\senc{A}}$, and, for $i \in [0\dd h)$,
    \smallskip
    \begin{itemize}
        \item $\senc{A}[e_i\dd e_{i + 1}) = \senc{A[p_i \dd p_{i + 1})}$ and $r_i = \rank_A(p_i)$, and
                \item either $A[p_i \dd p_{i + 1})$ is all-zero, or $r_{i + 1} - r_{i} \leq e_{i + 1} - e_{i} \leq \lg N$.
    \end{itemize}
    \smallskip
     After a $\tilde\Oh(N)$ time preprocessing, the following holds.
     Given $\senc{A}$, the sequence can be computed in $\Oh(\absolute{\senc{A}} / \lg N)$ time and $\Oh(\absolute{\senc{A}} \cdot \lg n/ \lg N)$ bits of space. In the same time and space, we can compute a data structure that,
    \smallskip
    \begin{itemize}
        \item given $i \in [0\dd h]$ and $j \in [p_i\dd p_{i + 1})$, returns $\rank_A(j)$ in constant time, and
        \item given $i \in [0\dd h]$ and $j \in (r_i\dd r_{i + 1}]$, returns $\select_A(j)$ in constant time.
    \end{itemize}
    
\end{restatable}

Now, to answer a query $\select_A(j)$, we only have to find the unique tuple of the decomposition with $r_i < j \leq r_{i + 1}$. This can be done via two auxiliary bitmasks of length $\absolute{\senc{A}}$ that respectively mark the starting positions of literal tokens in $\senc{A}$ and the positions $e_i$ of all the tuples. Then, we can find the correct tuples via rank and select queries on the auxiliary bitmasks. Crucially,
we can use existing data structures for constant time rank and select support on these bitmasks \cite{DBLP:conf/soda/BabenkoGKS15}.

\begin{restatable}{lemma}{restatelemmasparseselect}\label{lem:sparse_select}
    For every $N \in [5\dd2^w]$, after a $\tilde\Oh(N)$ time preprocessing, the following holds.
    Given a sparse encoding $\senc{A}$ of a bitmask $A[0\dd n)$ with $n \in 2^{\Oh(w)}$ and $\absolute{\senc{A}} = \Oh(\poly(N))$, we can compute a data structure for $\Oh(1)$ time select queries in $\Oh(\absolute{\senc{A}} / \lg N)$ time and $\Oh(\absolute{\senc{A}}\cdot  ( 1 +\lg n / \lg N))$ bits of space.
\end{restatable}

For implementing rank support, we rely on an improved version of van Emde Boas trees~\cite{DBLP:conf/focs/Boas75} given in the \lcnamecref{lem:vEB_final_space_and_time} below. It is well known that the claimed complexities can be achieved with \emph{expected} construction time and space \cite{DBLP:conf/stoc/PatrascuT06,DBLP:journals/corr/abs-cs-0603043}. Our deterministic solution can be obtained by essentially replacing the hash tables of van Emde Boas trees with deterministic dictionaries \cite{DBLP:conf/icalp/Ruzic08}. For completeness, we give a detailed description in \cref{sec:vEB}.

\restatevEB*

We construct this data structure for the sequence $p_0, \dots, p_h$ from \cref{lem:decompose_for_rank_select}, simulating words of width $w = \Theta(\lg n)$. This way, for a query $\rank_A(j)$, we can find the unique $i$ such that $j \in [p_{i} \dd p_{i + 1})$, allowing us to answer the query with \cref{lem:decompose_for_rank_select}.

\begin{restatable}{lemma}{restatesparserank}\label{lem:sparse_rank}
    For every $N \in [5\dd2^w]$, after a $\tilde\Oh(N)$ time preprocessing, the following holds.
    Given a sparse encoding $\senc{A}$ of a bitmask $A[0\dd n)$ with $n \in 2^{\Oh(w)}$ and a parameter $m \geq \absolute{\senc{A}} / \lg {N}$, we can compute a data structure for $\Oh(\lg \frac{\lg n - \lg m}{\lg \lg n})$ time rank and predecessor queries in $\Oh(m)$ time and $\Oh(m  \lg n)$ bits of space.
\end{restatable}

Finally, by combining \cref{thm:sss_sparse} with \cref{lem:sparse_select} and \cref{lem:sparse_rank} with parameter $m = \Theta(\frac{n \lg \tau}{\tau \lg n})$, we obtain the following main result.

\begin{restatable}{corollary}{restatesparsewithsupport}\label{thm:sss_sparse_with_support}
    A string $T\in [0\dd \sigma)^n$ can be preprocessed in $\Oh(n/\log_\sigma n)$ time so that, given $\tau \in [1\dd \floor{n/2}]$, a $\tau$-synchronizing set $\Sync$ of $T$ of size $|\Sync| < \frac{70n}{\tau}$ can be constructed in $\Oh(\frac{n \lg \tau}{\tau \lg n})$ time.
    The set is reported in a representation of size $\Oh(\frac{n \lg \tau}{\tau})$ bits that supports select queries in constant time, and rank queries in $\Oh(\lg\frac{\lg \tau}{\lg \lg n})$ time.
\end{restatable}

\appendix
\crefalias{subsection}{appendix}

\section{Parsing Sparse Encodings}\label{appendix:parsing}

Recall that the representation of the constructed $\tau$-synchronizing set is its bitmask in sparse encoding. In this section, we provide algorithms for efficiently constructing, decoding, and processing sparse encodings. For convenience, we repeat the definition and some key properties of sparse encodings below.

\restatedefgamma*

\restatedefsparse*

\restateexamplesparse*

\restateobssparsesize*

\restateobssparsesizelower*

\bigskip 
If a string $A$ is a prefix of another string $A'$, then the sparse encoding of $A$ is also a prefix of the sparse encoding of $A'$, unless $A$ ends with a zero-run that can be extended further in $A'$. In the latter case, the zero-run is encoded differently in $A$ and $A'$. Since the set of Elias-$\gamma$ codes (and hence the set of possible tokens) is prefix-free, this implies that the encoding of $A$ cannot be a prefix of the encoding of $A'$. 

\restateEncPrefOfOther*

\subsection{Basic Tools}

We start with auxiliary results for encoding and decoding the compressed representation.
Naive algorithms for this task, broadly speaking, require time linear in the number of encoded or decoded bits. 
Our more advanced algorithms are a factor $\lg N$ faster than the naive ones, at the cost of requiring an $\Oh(N)$ time preprocessing, for a parameter $N$ of our choice. 

\label{sec:sec_that_contains_elias}
\restateeliasencodedecode*

\begin{proof}
	Let $X[0\dd u)$ be the $u$-bit binary representation of $x$ (where $X[0]$ is the most significant bit). 
	Let $z_B$ and $z_X$ be the respective number of leading zeros in $B$ and $X$. 
		If $z_X$ (resp.\ $z_B$) is known, then $B = \mathtt{0}^{u - z_X - 1} \cdot X[z_X\dd u)$ (resp.\ $X = \mathtt{0}^{u - z_B - 1} \cdot B[z_B\dd 2z_B]$) can be computed in $\Oh(1)$ time. It remains to compute $z_B$ and $z_X$.

    We precompute $h = \ceil{(\lg N) / 2 }$.
    For computing $z_B$, we start with $\ell = 0$. As long as $B[0\dd \ell h + h)$ is all-zero, we increment $\ell$. For computing $z_X$, we instead start with $\ell = \floor{u /h}$ and, as long as $X[0\dd \ell h)$ is not all-zero, decrement $\ell$.
    In either case, this takes $\Oh(1 + \lg x / \lg N)$ time and reveals the $h$-bit block of $B$ or $X$ that contains the leftmost one-bit. The exact position of the bit can be isolated using a standard lookup table for bitmasks of length $h$, which can be computed in $\Oh(2^h) \subset \Oh(N)$ time.
\end{proof}

\restateEncodeDecodeDS*

\begin{proof}
    For (\ref{lem:encode_sparse_naive_bullet}), we either scan $A$ from left to right and encode each token separately, or we scan $\senc{A}$ from left to right and decode each token separately. 
	Computing the $u$-bit representation of $x > 0$ from $\gamma(x)$ or vice versa takes $\Oh(1 + \lg x / \lg N)$ time with \cref{lem:encode_elias,lem:decode_elias}. 	
	The term $\Oh(\lg x / \lg N)$ amortizes to $\Oh(1 / \lg N)$ per encoded or decoded bit of $\senc{A}$. 
	Hence, it sums to $\Oh(\absolute{\senc{A}} / \lg N)$.
	Whenever we encode or decode a run $0^x$, we need additional $\Oh(x)$ time to read or write $0^x$, which sums to $\Oh(n)$.

    For (\ref{lem:encode_ds_bullet}), we copy $A$ and then start running the algorithm from (\ref{lem:encode_sparse_naive_bullet}). We split the computation into $\Oh(n)$ time when $A$ is given, and (if necessary) $\Oh(1 + \absolute{\senc{A}} / \lg N)$ query time. Hence, the data structure consists of the copy of $A$ and the state of the algorithm after the initial $\Oh(n)$ computation time. At query time, we merely finish running the algorithm.
\end{proof}

\begin{lemma}\label{lem:encode_sparse_naive}\label{lem:decode_sparse_naive}
	For every $N \in [2\dd 2^w]$, after an $\Oh(N)$-time preprocessing, the following holds for every $u \in \Oh(w)$.
    Let $A\in[0\dd 2^u)^n$ with $n \in 2^{\Oh(w)}$. Computing the $u$-bit representation of $A$ from $\senc{A}$ and vice versa can be done in $\Oh(n + \absolute{\senc{A}} / \lg N)$ time.
\end{lemma}

\begin{proof}
    We either scan $A$ from left to right and encode each token separately, or we scan $\senc{A}$ from left to right and decode each token separately. 
	Computing the $u$-bit representation of $x > 0$ from $\gamma(x)$ or vice versa takes $\Oh(1 + \lg x / \lg N)$ time with \cref{lem:encode_elias,lem:decode_elias}. 	
	The term $\Oh(\lg x / \lg N)$ amortizes to $\Oh(1 / \lg N)$ per encoded or decoded bit of $\senc{A}$. 
	Hence, it sums to $\Oh(\absolute{\senc{A}} / \lg N)$.
	Whenever we encode or decode a run $0^x$, we need additional $\Oh(x)$ time to read or write $0^x$, which sums to $\Oh(n)$.
\end{proof}

\begin{lemma} \label{lem:encode_ds}
For every $N \in [2\dd 2^w]$, after an $\Oh(N)$-time preprocessing, the following holds for every $u \in \Oh(w)$.
Let $A \in [0\dd 2^u)^n$ with $n \in 2^{\Oh(w)}$ be given in $u$-bit representation. In $\Oh(n)$ time and space, one can compute a data structure that returns $\senc{A}$ in $\Oh(1 + \absolute{\senc{A}} / \lg N)$ time.
\end{lemma}

\begin{proof}
We copy $A$ and then use the algorithm from \cref{lem:encode_sparse_naive}. We split the computation into $\Oh(n)$ time when $A$ is given, and (if necessary) $\Oh(1 + \absolute{\senc{A}} / \lg N)$ query time. Hence, the data structure consists of the copy of $A$ and the state of the algorithm after the initial $\Oh(n)$ computation time.
\end{proof}

\begin{restatable}{lemma}{restatesparselongestprefix}\label{lem:sparse_longest_prefix_in_word}
    For every $N \in [2\dd 2^w]$, after an $\tilde\Oh(N)$-time preprocessing, the following holds.
    Given a bitmask~$B$ and an integer $\ell \in [0\dd \ceil{\lg N}]$, one can find in $\Oh(1)$ time the largest ${b \in [0\dd \min(\ell, \absolute{B})]}$ such that $B[0\dd b)=\senc{A}$ for some $A\in \Zz^*$.
    If $b > 0$, then each of the following values can be output in $\Oh(1)$ time:

    \begin{itemize}
        \item $a:=\absolute{A}$, $a_+ := \absolute{\{i \in [0\dd a) \mid A[i] > 0 \}}$, and ${x := \max\{ A[i] \mid i \in [0\dd a) \}}$,
        \item the bitmask $A'[0\dd a)$ marking the non-zero symbols in $A$,
        \item the bitmask $B'[0\dd b)$ marking the starting positions of literal tokens in $B[0\dd b)$,
        \item for any $j \in [0\dd a)$, the value $\rank_{A'}(j)$, 
        \item for any $j \in [1\dd a_+]$, the value $\select_{A'}(j)$.
    \end{itemize}
    
    If an optional parameter $u \in {\Oh(w)}$ with $u > \lg x$ is provided, then the $u$-bit representation of $A[0\dd a)$ can be returned in $\Oh(1 + au / \lg N)$ time.
\end{restatable}

\newcommand{\zero}{\mathtt{0}}
\newcommand{\one}{\mathtt{1}}

\begin{proof}
    Let $k = \ceil{\lg N}$.
    We store a subset of query answers in a lookup table $L$.
    For every $X \in \{\zero,\one\}^k$ and $\ell, u \in [1\dd k]$, we store 
    $L[X, \ell, u] = (b, a, a_+, x, A', \mathcal R, \mathcal S, A_u)$ 
    such that $X[0\dd b)$ with $b \in [0\dd \ell]$ is the encoding of a length-$a$ string $A$, 
    and there is no $b' \in (b\dd \ell]$ such that $X[0\dd b')$ is a sparse encoding. 
    Furthermore, if $b > 0$, then the values $a_+$, $x$, $A'$, and $B'$ 
    are defined like in the statement of the \lcnamecref{lem:sparse_longest_prefix_in_word}.
    Array $\mathcal R[0\dd a)$ contains the values $\rank_{A'}(j)$ in increasing order, while $\mathcal S[0\dd a_+)$ contains the values $\select_{A'}(j)$ in increasing order.
    If $x < 2^u$, then $A_u$ is the $u$-bit representation of $A$.

    Each of $b, a, a_+, x, A', B'$ is stored in a single word of memory. Each of $\mathcal R, \mathcal S$ is stored in (a prefix of) $k$ words of memory, using one word per array entry (which is possible due to $a < k$).
    Finally, $A_u$ is also stored in (a prefix of) $k$ words of memory, but in $u$-bit representation.
    Given a query $(B, \ell, u)$ with $\absolute{B} \geq k$, we can report $b, a, a_+, x, A'$ and any array entry of $\mathcal R, \mathcal S$ in constant time by merely looking up $L[B[0\dd k), \ell, u]$.
    We can report $A_u$ in a word-wise manner in $\Oh(1 + au/w)$ time.
    
    The Elias-$\gamma$ code $\gamma(x)$ is a substring of $X[0\dd b)$, and thus it consists of $2\floor{\lg x} + 1 \leq k$ bits. Hence, $\lg x < \floor{\lg x} + 1 \leq (k + 1)/2 \leq \lg N$, where the last inequality holds for any integer $N \geq 2$. We have shown $x < N \leq 2^k$. 
        For queries with $u > k$, we obtain $A_k$ in $\Oh(1 + ak/w)$ time by looking up $L[B[0\dd k), \ell, k]$. We convert $A_k$ into $u$-bit representation one symbol at a time, which takes $\Oh(a) \subseteq \Oh(au / \lg N)$ time (where $u / \lg N = \Omega(1)$ due to $u > k \geq \lg N$).
    For queries with $\absolute{B} < k$, we access the table with $B \cdot \one \cdot \zero^{k - \absolute{B} - 1}$ instead. Padding with an incomplete literal token does not affect the result.

	The table has $2^k\cdot k^2$ entries. Each entry consists of $\Oh(k)$ words, and computing an entry is straightforward in $\Oh(\poly(k))$ time. Hence, the preprocessing time is $\Oh(2^k \cdot \poly(k)) = \tilde\Oh(N)$, as required.
\end{proof}

\section{Adding Rank and Select Support to Sparse Encodings}

We will return a synchronizing set as the sparse encoding of its bitmask, which by itself does not allow fast random access. Hence, we develop support data structures that augment the encoding for fast rank and select queries.
For any set $S \subseteq \mathbb Z$ and $x \in \mathbb Z_{\geq 0}$, we define $\rank_{S}(x) = \absolute{\{ y \in S \mid y < x\}}$ and $\pred_{S}(x) = \max(\{ y \in S \mid y \leq x\} \cup \{ -\infty \})$.
For every $i \in [1\dd \absolute{S}]$, we define $\select_{S}(i) = \max\{ y \in S \mid \rank_{S}(y) < i\}$.
For $S \subseteq [0\dd U)$ and its characteristic bitmask $A[0\dd U)$, we may equivalently use subscript $A$ rather than $S$.

In \cref{sec:vEB}, we provide a deterministic implementation of van Emde Boas Trees, which we use in \cref{sec:sparserankselect} to obtain the support data structures for sparsely encoded bitmasks. 
\subsection{Deterministic van Emde Boas Trees}
\label{sec:vEB}

For a set of $n$ integers from range $[0\dd U)$, there is a predecessor data structure of size $\Oh(n \lg \lg U)$ bits with query time $\Oh(\lg(\lg \frac{{U}}{n}) / \lg \lg {U})$ and expected construction time $\Oh(n)$ \cite{DBLP:conf/stoc/PatrascuT06,DBLP:journals/corr/abs-cs-0603043}. The construction is based on van Emde Boas trees \cite{DBLP:conf/focs/Boas75}. 
By replacing the hash tables of the van Emde Boas tree with deterministic dictionaries \cite{DBLP:conf/icalp/Ruzic08}, the data structure can be modified to achieve $\Oh(n)$ worst-case construction time without affecting the query time or space. 
For completeness, we describe this modification in detail below.
The core of the data structure is the following van Emde Boas reduction of the universe.

\begin{lemma}[{see, e.g., \cite[Section 5.4.3]{DBLP:journals/corr/abs-cs-0603043}}]\label{lem:vEB_reduce}
    Let $S \subseteq [0\dd 2^\ell)$ of size $\absolute{S} = n$ with $\ell \geq 2$ and $n, 2^\ell \in 2^{\Oh(w)}$ be given as an array of $\ell$-bit integers. There exist $k \leq n + 1$ sets $S_1, \dots, S_k \subseteq [0\dd 2^{\ceil{\ell / 2}})$ with $n = \sum_{j = 1}^k \absolute{S_j}$ such that a predecessor or rank query in $S$ can be reduced to answering at most one predecessor or rank query in one of the sets.
    In $\Oh(n \lg^2 \lg n)$ time and words of space, we can construct the sets and an auxiliary data structure that performs the query reduction in constant time.
\end{lemma}

\begin{proof}
    We interpret each element as a length-$\ell$ bitmask.
    We first sort~$S$ in increasing order. 
    The initial set is $S_1 = \{ \floor{s / 2^{\floor{\ell / 2}}} \mid s \in S\}$, containing all the upper halves of elements in~$S$.
    Let $s_2 < \dots < s_{k}$ be the elements of $S_1$ in increasing order.
    For $i \in [2\dd k]$, we define $S'_{i} = \{ s - s_i \cdot 2^{\floor{\ell / 2}} \mid s \in S \textnormal{\ and\ } \floor{s / 2^{\floor{\ell / 2}}} = s_i \}$, i.e., each set $S'_{i}$ contains the lower $\ceil{\ell / 2}$ bits of all the elements with bit-prefix $s_i$.
    Finally, we define $S_{i} = S'_{i} \setminus \{\max(S_i')\}$.
    In a dictionary data structure, we associate each $s_i$ with the tuple $(i, m_i, M_i, \rank_{S}(m_i), \rank_{S}(M_i))$, where $m_i = \min(S_i') + s_i \cdot 2^{\floor{\ell / 2}}$ and $M_i = \max(S_i') + s_i \cdot 2^{\floor{\ell / 2}}$, i.e., $m_i$ and $M_i$ are the respectively minimal and maximal elements of $S$ that have bit-prefix $s_i$.
    Each element of $S$ with bit-prefix~$s_i$ contributes exactly one element to $S_i$, except for one element (because we excluded $\max(S_i')$), to which we charge the element $s_i$ of $S_1$. Hence, the total size of all sets is $n$.

    A predecessor or rank query $j \in [0\dd 2^{\ell})$ can be answered as follows. 
    We explicitly store $\min(S)$, and first test if $j < \min(S)$. Hence, we can easily detect $\pred_S(j) = -\infty$. From now on, assume $\pred_S(j) \geq 0$.    
    We extract the bit-prefix $j' = \floor{j / 2^{\floor{\ell / 2}}}$. If $j' \notin S_1$ (tested with the dictionary), then no element in $S$ has bit-prefix $j'$, and $\pred_S(j)$ is the largest element with a bit-prefix smaller than $j'$. We compute $s_i = \pred_{S_1}(j')$. By accessing the dictionary with $s_i$, we retrieve the query answers $\pred_S(j) = M_i$ and $\rank_S(j) = \rank_S(M_i) + 1$.
    
    It remains the case where $j' \in S_1$, which implies $j' = s_i$ for some $i \in [2\dd k]$. We obtain $(i, m_i, M_i, \rank_{S}(m_i), \rank_{S}(M_i))$ with the dictionary. If $M_i \leq j$, then we return $\pred_S(j) = M_i$ and $\rank_S(j) = \rank_S(M_i)$ if $M_i = j$ and otherwise $\rank_S(j) = \rank_S(M_i) + 1$. If $m_i > j$, then we proceed similarly to the case $j' \notin S_1$, but looking up $\pred_{S_1}(j' - 1)$ instead of $\pred_{S_1}(j')$. Finally, we are in the case $m_i \leq j < M_i$, i.e., we know that the predecessor is among the elements with bit-prefix $j' = s_i$.
    Hence, focus on the lower bits of $j$ given by $j \bmod 2^{\floor{\ell / 2}}$ and answer the query relative to $S_i$.
    For a predecessor query, we compute $p = \pred_{S_i}(j \bmod 2^{\floor{\ell / 2}})$ and return $\pred_S(j) = p + j' \cdot 2^{\floor{\ell / 2}}$. Similarly, for a rank query, we compute $r = \rank_{S_i}(j \bmod 2^{\floor{\ell / 2}})$ and return $\rank_S(j) = r + \rank_S(m_i)$.

    It is clear that the sets can be computed in $\Oh(n \lg \lg n)$ time by sorting and scanning~$S$~\cite{DBLP:journals/jal/Han04}. A dictionary data structure with constant query time can then be built in $\Oh(n \lg^2 \lg n)$ time and space \cite{DBLP:conf/icalp/Ruzic08}.
\end{proof}

\begin{corollary}\label{lem:vEB_logl}
    Let $S \subseteq [0\dd 2^\ell)$ of size $\absolute{S} = n$ with $\ell \geq 2$ and $n, 2^\ell \in 2^{\Oh(w)}$ be given as an array of $\ell$-bit integers.
    A deterministic data structure that answers rank and predecessor queries in $\Oh(\lg \ell)$ time can be built in $\Oh(n \lg^2 \lg n\cdot \lg \ell)$ time and words of space.
\end{corollary}

\begin{proof}
    It suffices to apply \cref{lem:vEB_reduce} recursively. After $\Oh(\lg \ell)$ levels of recursion, the universe and thus also the size of each set is constant. On each level, the total number of elements in all sets is exactly $n$, and thus the construction time and space are as claimed.
\end{proof}

As a first improvement, we observe that the recursion can be aborted early. As soon as the universe of each subinstance is small enough, we can either use complete tabulation or a dynamic fusion node to directly solve the subinstance.

\begin{corollary}[{see \cite[Section 5.4.1]{DBLP:journals/corr/abs-cs-0603043}}]\label{lem:vEB_logl_over_a}
    Let $S \subseteq [0\dd 2^\ell)$ of size $\absolute{S} = n$ with $\ell \geq 2$ and $n, 2^\ell \in 2^{\Oh(w)}$ be given as an array of $\ell$-bit integers. For $m \geq n$, let $a = \lg(m / n) + \lg w$.
    A deterministic data structure that answers rank and predecessor queries in $\Oh(\lg \frac{\ell}{a})$ time can be built in $\Oh(m + n \lg^2 \lg n\cdot \lg \ell)$ time and words of space.
\end{corollary}

\begin{proof}
    We proceed like in the proof of \cref{lem:vEB_logl}, but stop the recursion as soon as the remaining keys consist of $\leq a/2$ bits, which happens after $\Oh(\lg \frac\ell a)$ levels.    
    If $a \leq 2\lg w$, then there are at most $2^a = \Oh(\poly(w))$ keys in each set of the final level. In this case, we use a predecessor data structure with constant query time, and construction time and space linear in the number of elements (e.g., a dynamic fusion node \cite{DBLP:conf/focs/PatrascuT14}).
    Otherwise, it holds $a \leq 2\lg(m / n)$. In this case, for each set, we precompute all possible queries. There are at most $2^{a/2} \leq m/n$ possible queries, and hence their answers can be computed and stored in $\Oh(m / n)$ time and space.
    We observe that the total number of non-empty sets in each level of the recursion of \cref{lem:vEB_logl} is at most $n$.
    Since we always stop the recursion on the same level, it is clear that the additional time and space sum to $\Oh(m)$.
\end{proof}

To reach the final query time, we also use tabulation for the initial $\lg m$ bits of each key.

\begin{corollary}[{see \cite[Section 5.4.2]{DBLP:journals/corr/abs-cs-0603043}}]\label{lem:vEB_final_query_time}
    Let $S \subseteq [0\dd 2^\ell)$ of size $\absolute{S} = n$ with $\ell \geq 2$ and $n, 2^\ell \in 2^{\Oh(w)}$ be given as an array of $\ell$-bit integers. For $m \geq n$, let $a = \lg(m / n) + \lg w$.
    A deterministic data structure that answers rank and predecessor queries in $\Oh(\lg \frac{\ell - \lg m}{a})$ time can be built in $\Oh(m + n \lg^2 \lg n\cdot \lg \ell)$ time and words of space.
\end{corollary}

\begin{proof}
    We use \cref{lem:vEB_logl_over_a}.
    However, during the initial application of \cref{lem:vEB_reduce} on the very first level of recursion, we define the set $S_1$ such that it considers the highest $\floor{\lg m}$ (rather than $\floor{\ell/2}$) bits of each key. The sets $S_2, \dots, S_k$ then consider the lowest $\ell - \floor{\lg m}$ bits.
    Since the elements of $S_1$ are integers consisting of $\floor{\lg m}$ bits, the universe for predecessor queries on $S_1$ is $[0 \dd m)$. Hence, instead of recursing on $S_1$, we can afford to precompute the answers to all queries in $\Oh(m)$ time and words of space (by scanning~$S_1$).
    The sets $S_2, \dots, S_k$ are solved recursively like before.
    The depth of the recursion becomes $\Oh(\lg (\ell - \lg m))$ without the improvement of \cref{lem:vEB_reduce}, and $\Oh(\lg \frac{\ell - \lg m}a)$ with the improvement, leading to the claimed query time.
\end{proof}

Finally, by only computing the data structure for every $w^2$-th element of $S$ and solving each range of $w^2$ consecutive keys directly with a dynamic fusion node, we can reduce the space and construction time without increasing the query time (see \cite[Section 5.3]{DBLP:journals/corr/abs-cs-0603043}).

\restatevEB*

\begin{proof}
    If $n = \Oh(\poly(w))$, then a dynamic fusion node achieves the claimed complexities~\cite{DBLP:conf/focs/PatrascuT14}. Hence, assume $n \geq w^2$.
    The set is given as an increasing array $S[0\dd n)$.
    Let $S'[0 \dd n')$ with $n' := \ceil{n / w^2} = \Oh(n / w^2)$ be defined by $S'[i] := S[i \cdot w^2]$, i.e., we sample every $w^2$-th element of $S$.
    We build the data structure from \cref{lem:vEB_final_query_time} for the set $S'$, which takes $\Oh(m + n \lg^2 \lg n\cdot \lg \ell / w^2) = \Oh(m)$ time and space.
    For each $i \in [0\dd n')$, we build a dynamic fusion node over the elements $S_i :=S[i \cdot w^2\dd \min(n, (i + 1) \cdot w^2))$, which takes $\Oh(n)$ time and space overall, and allows us to answer queries with respect to $S_i$ in constant time~\cite{DBLP:conf/focs/PatrascuT14}.
    Given a query $j$, we immediately return $\pred_S(j) = -\infty$ and $\rank_S(j) = 0$ if $j < S[0]$. Otherwise, we obtain $r = \rank_{S'}(j)$.
    Clearly, $\pred_S(j)$ is one of the elements in $S_r$, and we obtain $\pred_S(j) = \pred_{S_r}(j)$ and $\rank_S(j) = r \cdot w^2 + \rank_{S_r}(j)$.
\end{proof}

\subsection{Rank and Select Support for Sparse Encodings}
\label{sec:sparserankselect}

Now we are ready to present the support data structures for sparsely encoded bitmasks.
The auxiliary \lcnamecref{lem:decompose_for_rank_select} below decomposes the encoding into small pieces.
Rank and select queries relative to these pieces can already be solved in constant time via precomputed lookup tables.

\restatedecomposeencoding*

\begin{proof}
    We perform the $\tilde\Oh(N)$ time preprocessing from \cref{lem:sparse_longest_prefix_in_word}.
    Now we compute the sequence of tuples, initialized with $(0, p_0, e_0, r_0) = (0,0,0,0)$.
    Assume that we have already created tuple $(i, p_i, e_i, r_i)$. If $p_i < n$, then our goal is to create $(i + 1, p_{i + 1}, e_{i + 1}, r_{i + 1})$.
    To this end, we access the data structure from \cref{lem:sparse_longest_prefix_in_word} with bitmask $\senc{A}[e_i\dd \absolute{\senc{A}})$ and $\ell = \ceil{\lg N}$.
    This results in the maximal $b \in [0\dd \min(\absolute{\senc{A}} - e_i, \ceil{\lg N}))$ such that $\senc{A}[e_i\dd e_i + b)$ is a sparse encoding. If $b > 0$, then we also obtain the length $a$ of the encoded string and the number $a_+$ of its non-zero symbols. Hence, we can assign $p_{i + 1} := p_i + a$, $e_{i + 1} = e_i + b$, and $r_{i + 1} = r_i + a_+$. Note that $a_+ < b \leq \lg N$, i.e., $r_{i + 1} - r_i \leq e_{i + 1} - e_{i} \leq \lg N$.
    
    If, however, $b = 0$, then it is clear that the next token of the encoding consists of $\geq \lg N$ bits (as otherwise we would have had $b > 0$), which also implies that it is a zero-run token encoding some run $\texttt0^x$ (as every $1$-bit is encoded as a literal token consisting of $2 < \lg N$ bits).
    We obtain $x$ and the length $2\cdot \floor{\lg x} + 2$ of the encoded token (in bits) in $\Oh(\lg x / \lg N)$ time using \cref{lem:decode_elias},
    and assign $p_{i + 1} := p_i + x$, $e_{i + 1} = e_i + 2\cdot \floor{\lg x} + 2$, and $r_{i + 1} = r_i$.
    The term $\Oh(\lg x / \lg N)$ sums to $\Oh(\absolute{\senc{A}} / \lg N)$ over all tokens that consist of $\geq \lg N$ bits.

    It is clear that this procedure decomposes $\senc{A}$ and $A$ correctly due to \cref{obs:encoding_prefix_of_other}. Also, apart from the large zero-runs, each element of the sequence is created in constant time.
    Hence, it remains to be shown that only $\Oh(1 + \absolute{\senc{A}} / \lg N)$ tuples are created.
    Consider any $i \in [1\dd h)$.
    We claim that $e_{i + 1} - e_{i - 1} > \lg N$.
    Indeed, if $e_{i + 1} - e_{i - 1} < \lg N$, then, while creating tuple $(i, p_i, e_i, r_i)$, we would have accessed the data structure from \cref{lem:sparse_longest_prefix_in_word} with $\senc{A}[e_{i - 1} \dd \absolute{\senc{A}})$, and it would have returned $b \geq e_{i + 1} - e_{{i - 1}}$.
    Therefore, every two consecutive tuples advance the encoding by at least $\lg N$ bits, which implies that there are $\Oh(1 + \absolute{\senc{A}} / \lg N)$ tuples.

    Finally, whenever $e_{i + 1} - e_{i} > \lg N$, we know that $A[p_i\dd p_{i + 1})$ is all-zero. The answer for all rank queries in this range is $r_i$, and there is no select query with answer in this range. If $e_{i + 1} - e_{i} \leq \lg N$, then we can already answer rank and select queries relative to $A[p_i\dd p_{i + 1})$ in constant time by accessing the data structure from \cref{lem:sparse_longest_prefix_in_word} with $\senc{A}[e_i \dd e_{i  + 1})$. Hence, we can answer global queries by offsetting the relative rank answers by $r_i$, and the relative select answers by $p_i$.
\end{proof}

\restatelemmasparseselect*

\begin{proof}
    We use the tuple sequence $(i, p_i, e_i, r_i)_{i = 0}^{h}$ and data structures from \cref{lem:decompose_for_rank_select,lem:sparse_longest_prefix_in_word}.     
    For a global select query $j$, our intermediate goal is to quickly identify the tuple $(i, p_i, e_i, r_i)$ with $r_i < j \leq r_{i + 1}$.
    To this end, we compute two auxiliary bitmasks $B[0\dd \absolute{\senc{A}})$ and $C[0\dd \absolute{\senc{A}})$.
    The former marks all the positions $e_i$ with $i \in [0\dd h]$, and can trivially be constructed in $\Oh(\absolute{\senc{A}} / w + h + 1)$ time.
    The latter marks the starting positions of all the literal tokens in $\senc{A}$.
    For any tuple $(i, p_i, e_i, r_i)$ with $e_{i + 1} - e_i \leq \lg N$, we can obtain $C[e_i\dd e_{i + 1})$ in constant time by accessing the data structure from \cref{lem:sparse_longest_prefix_in_word} with $\senc{A}[e_i\dd e_{i + 1})$ (while all other regions of $C$ are all-zero).
    Hence, also $C$ can be obtained in $\Oh(\absolute{\senc{A}} / w + h + 1)$ time.
    After an $\Oh(N)$ time preprocessing, it takes $\Oh(1 + \absolute{\senc{A}} / \lg N)$ time and $\Oh(\absolute{\senc{A}})$ bits of space to prepare $B$ and $C$ for constant time rank and select queries \cite[Lemma 2.1]{DBLP:conf/soda/BabenkoGKS15}\footnote{We invoke \cite[Lemma 2.1]{DBLP:conf/soda/BabenkoGKS15} with word width $\Theta(\lg N)$ and preprocessing time $\Oh(\sqrt{N})$. However, we can then only construct the data structure if $\absolute{\senc{A}} \in 2^{\Oh(\lg N)}$. If $\absolute{\senc{A}} > N$, then we perform the preprocessing of \cite[Lemma 2.1]{DBLP:conf/soda/BabenkoGKS15} at query time, using words of width $\Theta(\lg\; \absolute{\senc{A}})$, spending $\Oh(\sqrt{\absolute{\senc{A}}}) \subset \Oh(\absolute{\senc{A}} / \lg N)$ preprocessing time, and $\Oh(\absolute{\senc{A}} / \lg\; \absolute{\senc{A}}) \subseteq \Oh(\absolute{\senc{A}} / \lg N)$ construction time.}.
    
    The select query $j$ is answered by first retrieving the position $\select_C(j)$ of the $j$-th literal token in $\senc{A}$. Then, $r_i < j \leq r_{i + 1}$ holds for tuple $i = \rank_B(\select_C(j))$. Finally, we can obtain the answer in constant time by querying the data structure from \cref{lem:decompose_for_rank_select} with $i$ and $j$.
    Observing that $h = \Oh(\absolute{\senc{A}} / \lg N)$, and that words of width $\Oh(\lg n)$ bits are sufficient for storing the tuples, we obtain the claimed time and space complexity.
\end{proof}

For implementing rank support, we construct the data structure from \cref{lem:vEB_final_space_and_time} for the sequence $p_0, \dots, p_h$ from \cref{lem:decompose_for_rank_select}. This way, for a query $\rank_A(j)$, we can find the unique $i$ such that $j \in [p_{i} \dd p_{i + 1})$, allowing us to answer the query with \cref{lem:decompose_for_rank_select}.

\restatesparserank*

\begin{proof}
    We use the tuple sequence $(i, p_i, e_i, r_i)_{i = 0}^{h}$ and data structure from \cref{lem:decompose_for_rank_select}.
    We build the data structure from \cref{lem:vEB_final_space_and_time} for the sequence $p_0, \dots, p_h$, simulating words of width $w = \Theta(\lg n)$ and using space parameter $m' = m + h + 1$. Then, given any $j \in [0\dd n)$, we can use a rank query to find the unique $i$ such that $j \in [p_{i} \dd p_{i + 1})$ in $\Oh(\lg \frac{\lg n - \lg m}{\lg \lg n})$ time.
    Then, the query can be answered in constant time using the data structure from \cref{lem:decompose_for_rank_select}.
    The overall space is $\Oh((m + h) \lg n) = \Oh(m \lg n)$ bits, as required.
\end{proof}

\section{Accelerating Transducers}
\label{sec:transducer_main_sec}

In this section, we show how to efficiently process sparse encodings using transducers. For convenience, we repeat their definition.
\transducerdef

\subsection{Implementing a Single-Stream Transducer}\label{sec:lem:accelerate_single_stream}

We first show how to implement a single-stream transducer for sparse encodings. 
Later, we will show how to reduce a multi-stream transducer to a single-stream transducer, even when the multi-stream transducer operates on sparse encodings.
We therefore state the result for single-stream transducers in a more general form.
We will use precomputed information to fast-forward through multiple transitions.

The data structure below will be constructed for the graph induced by all transitions that have both input and output symbol zero. This will allow us to quickly skip through long chains of such transitions.

\begin{lemma}\label{lem:pseudoforest_levelancestor}
Consider a directed graph of $n$ nodes in which each node has at most one outgoing edge. The graph can be preprocessed in $\Oh(n)$ time and space such that the following queries take $\Oh(1)$ time. 
\smallskip
\begin{itemize}
    \item \emph{Jump query:} Given a node $v$ and a distance $d \in [0\dd 2^w)$, return the unique node $v'$ such that there is a directed path of length exactly $d$ edges from $v$ to $v'$, or return $\bot$ if such $v'$ does not exist.
    \item \emph{Furthest jump query:} Given a node $v$, return the maximal $d \in \mathbb{Z}_+\cup \{\infty\}$ such that a jump query $(v, d)$ returns a node (rather than $\bot$; we always have $d\in [0\dd n)\cup \{\infty\}$).
\end{itemize}
\end{lemma}

\begin{proof}
    Without loss of generality, assume that the graph is weakly connected (as otherwise we can treat each weakly connected component separately).
    If the graph contains no directed cycle, then the graph is a tree (with edges directed towards the root) and we can directly use a data structure for level ancestor queries~(e.g., \cite{DBLP:journals/tcs/BenderF04}). In this case, we answer a jump query with $\bot$ whenever $d$ exceeds the depth of the query node, and otherwise we return the tree ancestor of $v$ at distance $d$. We answer a furthest jump query by returning the depth of the query node in the tree.
    
    Otherwise, there is exactly one directed cycle (because every node has at most one outgoing edge, and thus a node cannot reach two different cycles).
    Let $S = \{s_0, \dots, s_{m - 1}\}$ be the nodes of this unique cycle, ordered such that there is an edge from $s_i$ to $s_{(i + 1) \bmod m}$ for every $i \in [0\dd m)$. 
    We consider a forest obtained by discarding the edges of the cycle.
    In this forest, edges are directed towards the root nodes, which are exactly the nodes from $S$.
    For every node $v$, we precompute the unique $\ell(v) \in [0\dd n)$ and $r(v) \in [0\dd m)$ such that $v$ is in the tree with root $s_{r(v)}$, and the path from $v$ to $s_{r(v)}$ is of length exactly $\ell(v)$ (measured in edges).
    We also construct a level ancestor data structure for each tree.
    
    Now we can answer a jump query $(v, d)$ as follows. 
    If $d \leq \ell(v)$, then we obtain the result using a level ancestor query in the tree with root $s_{r(v)}$.
    If $d > \ell(v)$, then we return $s_{(r(v) + d - \ell(v)) \bmod m}$ instead.
    For a furthest jump query, we simply return $\infty$, as every node can reach the cycle, and thus we can always reach a node at arbitrary distance.
    
    It takes $\Oh(n)$ time and space to construct a level ancestor data structure for each tree (e.g., \cite{DBLP:journals/tcs/BenderF04}). The remaining parts of the precomputation take $\Oh(n)$ time and space using standard techniques.
\end{proof}

\restatesingletape*

\begin{proof}
    For some small constant $\eps > 0$, let $M = \Theta(N^\eps)$ be a power of two. The preprocessing time will be $\Oh(q \cdot \poly(M)) \subseteq \Oh(qN)$ (where the exponent of $\poly(M)$ defines $\eps$), and the query time will be $\Oh(1+ (\absoluteenc{S} + \absoluteenc{T}) / \lg M) = \Oh(1+ (\absoluteenc{S} + \absoluteenc{T}) / \lg N)$.
    We precompute some information that will allow us to fast-forward through multiple transitions.

    \subparagraph{Lookup table for forwarding multiple tokens.} For every state $s \in [0\dd q)$,
    we precompute a lookup table $L_{s}$.
    For every bitmask $B \in \{0,1\}^{\lg M}$, the entry $L_{s}[B] = (b, a, s', z_1, z_2, D)$ indicates that $B[0\dd b)$ is a sparse encoding of a length-$a$ string, and $B[0\dd b')$ is not a sparse encoding for any $b' \in (b\dd \lg M]$.
	For explaining the remaining components, let $A[0\dd a)$ be the string encoded by $B[0\dd b)$.
	Then $s'$ is the state reached by the transducer when reading string $A$ in state $s$. Let $A'[0\dd a)$ denote the produced output string, then $z_1, z_2 \in [0\dd a]$ are the respectively maximal integers such that $A'[0\dd z_1) = 0^{z_1}$ and $A'[a - z_2\dd a) = 0^{z_2}$ (where $z_1 = z_2 = a$ if $A'[0\dd a)$ is all-zero).
	Finally, if $A'$ is not all-zero, then $D$ is the data structure from \crefbullet{lem:encode_ds} with preprocessing parameter $M$ computed for $A'[z_1\dd a - z_2)$.
	
	Given $B$, we can obtain $b$ and $a$ in constant time by using \cref{lem:sparse_longest_prefix_in_word} with preprocessing parameter~$M$.
        Since the lemma         also implies that $A$ is over alphabet $[0\dd M)$, we can obtain $A$ in $\lg M$-bit representation in $\Oh(a)$ time.
        We obtain $A'$ by performing $a$ transitions in overall $\Oh(a)$ time.
	Finally, we count the leading and trailing zeros of $A'$, and (if $A'$ is not all-zero) we compute $D$ in $\Oh(a)$ time.
	
	Computing an entry of $L_s$ takes $\Oh(1 + a)$ time. Since $A[0\dd a)$ is encoded by a prefix of~$B$, it holds $a = \Oh(M)$ by \cref{obs:sparse_size_lower}. There are $\Oh(qM)$ entries across all tables, and the preprocessing time is $\Oh(qM^2)$.
    
    \subparagraph{Data structure for forwarding long runs of zeros.}
    For every state $s$, we compute an additional table~$L'_s$ for fast-forwarding through zero-runs. For every $y \in [1\dd \floor{\sqrt[4]{M}}]$, we store $L'_s[y] = L_{s}[B]$, where $B=\senc{0^y} \cdot \one \cdot \zero^{\lg M - \absolute{\senc{0^y}} - 1}$. 
    Here, $B$ is always defined due to $\absolute{\senc{0^y}} = 2\floor{\lg y} + 2 < \lg M$.
    Since $B$ consists of $\senc{0^y}$ and an incomplete literal token, the entry $L'_s[y] = (b, a, s', z_1, z_2, D)$ satisfies $b = \absolute{\senc{0^y}}$ and $a = y$.
        Clearly, these tables can be computed in $\Oh(q \cdot \poly(M))$ time.

    We use a separate solution for forwarding through even longer runs of zeros.
    We construct a directed graph with nodes $[0\dd q)$. There is an edge from node $s$ to node $s'$ if and only if the transducer transitions from $s$ to $s'$ with input and output $0$.
    Hence, each node has at most one outgoing edge.
    The graph can be constructed in $\Oh(q)$ time by evaluating the transition with input symbol $0$ for every state.
    In additional $\Oh(q)$ time, we build the data structure $\mathcal D$ from \cref{lem:pseudoforest_levelancestor} for this graph.

    \subparagraph{Running the transducer.}   
    Our goal is to compute $\senc{T}$ from $X:=\senc{S}$.
    We perform a sequence of macro-steps. As an invariant, before each macro-step, we maintain the following information:
    \begin{itemize}
        \item integers $x \in [0\dd \absolute{X}]$ and $n \in [0\dd \absolute{S}]$ such that $X[0\dd x)=\senc{S[0\dd n)}$, and
        \item the state $s \in [0\dd q)$ reached by the transducer after reading $S[0\dd n)$, and
        \item the largest $z \in [0\dd n]$ such that $T[0\dd n) = T[0\dd n - z) \cdot 0^z$, and
        \item the sparse encoding $Y=\senc{T[0\dd n-z)}$. (Since $T[0\dd n-z)$ is either empty or ends with a non-zero symbol, $Y$ is a prefix of $\senc{T}$, recall \cref{obs:encoding_prefix_of_other}.)
    \end{itemize}
    Before the first macro-step, $s$ is the initial state of the transducer, and the other variables are initialized so that $x = n = z = 0$ and $Y = \emptystring$.
    Now, we show how to perform the next macro-step.
    We focus on describing the algorithm and analyze the complexity later.
    We access $L_{s}[X[x\dd x + \lg M)] = (b, a, s', z_1, z_2, D)$. If $x + \lg M > \absolute{X}$, then we access the table with $X[x\dd \absolute{X}) \cdot \one \cdot \zero^{\lg M + x - \absolute{X} -1}$ instead, padding with an incomplete token and thus not affecting the result. (Here, we exploit that tokens, just like Elias-$\gamma$ codes, are prefix-free, and thus an incomplete token prevents decoding beyond the end of the sparse encoding.)
    
    \subparagraph{Case 1: Forward by at least one short token.} If $b > 0$, then we first check if $a = z_1$, i.e., if the corresponding output string is all-zero.
    If $a = z_1$, then we leave $Y$ unchanged and assign $z \gets z + a$.
    Otherwise, we append $\senc{0^{z + z_1}}$ (if $z + z_1 > 0$) and the sparse encoding obtained from data structure $D$ to $Y$. In this case, we also assign $z \gets z_2$.
    Either way, we update $s \gets s'$, $x \gets x + b$, and $n \gets n + a$.
    This maintains the invariant for the next macro-step.

    \subparagraph{Case 2: Forward by one large token.}
    If $b = 0$, then the first token of $X[x\dd ]$ consists of over $\lg M$ bits. We decode it using \cref{lem:decode_elias}. Let $b'$ be the size of the token in bits.
    
    \subparagraph{Case 2a: Forward by one literal token.} If $X[x\dd x+b')=\senc{S[n]}$ is a literal token, then we perform the next transition naively.
    If the transducer transitions from $s$ to $s'$ with input $S[n]$ and output $t > 0$, then we append the encoding of $0^{z}$ (if $z > 0$) followed by the encoding of $t$ to $Y$. We then assign $z \gets 0$.
    If the output is $t = 0$, then we leave $Y$ unchanged and assign $z \gets z + 1$.
    Either way, we assign $s \gets s'$, $x \gets x + b'$, and $n \gets n + 1$.

    \subparagraph{Case 2b: Forward by one zero-run token.} If $X[x\dd x+b')=\senc{0^y}$ is a zero-run token representing the prefix $S[n\dd n+y) = 0^y$ of the remaining input, then the procedure is more complicated.
    We perform the macro-step in a series of micro-steps. Each micro-step decreases $y$, and we terminate the macro-step as soon as $y$ becomes $0$.
    We alternate between \emph{fixed} and \emph{flexible} micro-steps.
    
    For a fixed micro-step, our goal is to advance the output by $M' = \min(y, \floor{\sqrt[4]{M}})$ symbols.
    Hence, we look up $L'_s[M'] = (b, a, s', z_1, z_2, D)$. Similarly to Case~1, if $a \neq z_1$, then we first append the sparse encoding of $0^{z + z_1}$ to $Y$ (if $z + z_1 > 0$), and then we also append the sparse encoding obtained from data structure $D$. We further assign $z \gets z_2$. If $a = z_1$, then we leave $Y$ unchanged and assign $z \gets z + M'$ (in which case $a = z_1 = z_2 = M'$).    
    In any case, we update $s \gets s'$, $n \gets n + a$, and $y \gets y - M'$. (We do not update $x$ after a micro-step.)

    For a flexible micro-step, our goal is to advance the output to the next non-zero symbol. 
    Using a furthest jump query to $\mathcal D$, we obtain the maximal length $d$ such that reading $0^d$ in state $s$ leads to output $0^d$; this value satisfies $d\in [0\dd q)\cup \{\infty\}$.
    Let $\ell = \min(y, d)$, and let $s'$ be the result of a jump query $(s, \ell)$ to $\mathcal D$.
    Then, $s'$ is the state reached when reading input $0^\ell$ in state $s$.
    We assign $s \gets s'$, $z \gets z + \ell$, $n \gets n + \ell$, and $y \gets y - \ell$. (Again, we leave $x$ unchanged.)

    As soon as some micro-step terminates with $y = 0$, we conclude the macro-step by assigning $x \gets x + b'$.

    \subparagraph{Finalizing the computation.}
    Once the final macro-step terminates, if $z > 0$, we append $\senc{0^z}$ to $Y$ so that $Y=\senc{T}$. The correctness follows from maintaining the invariant for the variables, and from the description of the cases. It remains to analyze the time complexity.

    \subparagraph{Counting the steps.}
    We start by showing that the number of macro-steps is bounded by $\Oh(\ceil{\absoluteenc{S} / \lg M})$. 
    Consider two consecutive macro-steps that advance the input encoding by respectively $b_1$ and $b_2$ bits, i.e., they process $X_1 = X[x\dd x+b_1)$ and $X_2 = X[x+ b_1\dd x+b_1+b_2)$. Since both $X_1$ and $X_2$ are sparse encodings, also $X_1 \cdot X_2 = X[x\dd x+b_1 + b_2)$ is a sparse encoding. This implies $b_1 + b_2 > \lg M$ because $b_1$ is at least the maximal value from $[0\dd \lg M]$ such that $X[x\dd x+b_1)$ is a sparse encoding.
    Thus, any two consecutive macro-steps advance the input encoding by more than $\lg M$ bits, which implies that there are $\Oh(\absolute{X} / \lg M) = \Oh(\absoluteenc{S} / \lg M)$ macro-steps overall.

    Now we count the total number of micro-steps. Consider four consecutive micro-steps that are part of the same macro-step, where the first and third micro-steps are flexible, and the second and fourth are fixed.
    We define the steps via the amount of (plain) input symbols they covered. Let $m_1, \dots, m_4$ be such that, for $i \in [1\dd 4]$, the $i$-th step processed $S[n+\sum_{j=1}^{i - 1} m_j\dd n+m_i + \sum_{j=1}^{i - 1} m_j)$.
    Since each of the fixed steps is preceded by a flexible one, and each flexible step advances the output to the next non-zero symbol, it is clear that the first output symbol written in each of the fixed steps is a non-zero symbol.
    Hence, $T' = T[n+m_1\dd n+m_1+m_2+m_3]$ both starts and ends with a non-zero symbol. 
    The sparse encoding of $T'$ is therefore a substring of the sparse encodings that were appended to $Y$ during the micro-steps.
    Note that $m_2 = \floor{\sqrt[4]{M}}$ bounds the length of $T'$ from below, which also implies that the sparse encoding of $T'$ consists of $\Omega(\lg M)$ bits by \cref{obs:sparse_size_lower}.
    If we ignore the initial and final four micro-steps performed by each macro-step, then every micro-step advances the output by $\Omega(\lg M)$ bits on average. 
    Hence, there can only be $\Oh(\lceil\absoluteenc{T} / \lg M\rceil)$ micro-steps. 
    The initial and final four micro-steps of each macro-step sum to another $\Oh(\lceil\absoluteenc{S} / \lg M\rceil)$ micro-steps.

    \subparagraph{Analyzing the time complexity.}
    We have already shown that the number of micro- and macro-steps is $\Oh(\lceil(\absoluteenc{S} + \absoluteenc{T}) / \lg M \rceil)$. Thus, we can afford constant time for each of these steps. Whenever we spend more than constant time during one of the steps, the computation falls into one of the following categories:

    \begin{itemize}
        \item We look up $(b, a, s', z_1, z_2, D)$ in either of the lookup tables, and append the encoding $C$ obtained from $D$ to $Y$. This occurs in Case 1, as well as in fixed micro-steps in Case 2b. 
        By \crefbullet{lem:encode_ds}, the encoding can be obtained and appended in $\Oh(1 + \absolute{C} / \lg M)$ time.
                \item We encode a single symbol $x$ (in Case 2a) or a zero-run $0^x$ (in Cases 1 and 2a, and in a fixed micro-step in Case 2b) and append the encoding to $Y$. Either way, the token consists of $\Theta(\lg (1 + x))$ bits, and we take $\Oh(1 + \lg x / \lg M)$ time using \cref{lem:encode_elias}.         \item We decode the initial token of $X[x\dd ]$ in Case 2. The token consists of $b'$ bits, and we decode it in $\Oh(1 + b' / \lg M)$ time using \cref{lem:decode_elias}.         \item We perform a transition in Case 2a. The input symbol is either $0$, in which case the transition takes constant time, or it is some literal $x > 0$ with encoding of size $\absolute{\senc{x}} = \Theta(\lg (x + 1))$ bits. In the latter case, the transition takes $\Oh(1 + \lg(1 + x) / \lg M)$ time.     \end{itemize}
    As seen above, whenever we exceed constant time, we do so by some term $\Oh(r / \lg M)$. We then also advance either the input encoding or the output encoding by $\Theta(r)$ bits.
    Hence, the additional time sums to $\Oh((\absoluteenc{S} + \absoluteenc{T}) / \lg M)$, as required.
\end{proof}

\subsection{Reducing Multi-Stream Transducers to Single-Stream Transducers}

In this section, we show how to reduce a multi-stream transducer to single-stream transducers.
It may seem tempting to use the following straightforward reduction: rather than using $t$ input streams, say, each over alphabet $[0\dd 2^u)$, we may use a single input stream over alphabet $[0\dd 2^{u \cdot t})$. Instead of multiple length-$n$ input strings $S_1,\dots, S_t$, we would then process a single length-$n$ string $S$ defined by $\forall_{i \in [0\dd n)}\; S[i] = \sum_{j = 1}^t 2^{u \cdot (j - 1)} \cdot S_j[i]$.
However, with this technique, $\senc{S}$ can be significantly larger than the combined size of $\senc{S_1},\dots, \senc{S_t}$.
For example, if $\forall_{ j \in [1\dd t)}\; S_j = 0^n$ and $S_t = 1^n$, then $\sum_{j = 1}^t \aabsolute{S_j} = \Theta(n + (t - 1) \cdot \lg n)=\Theta(n)$, but $\aabsolute{S} = \Theta(n \cdot u \cdot (t - 1))=\Theta(n\cdot u)$.
This motivates the following definition of a zipped string that more accurately preserves the size of the sparse encoding.

\restatezippeddefinition*

We may view $x = \zip(A_1, \dots, A_t)[i]$ as a bitmask of length $\absolute{\senc{A_1[i]A_2[i]\dots A_t[i]}}$ or as an integer with binary representation $\senc{A_1[i]A_2[i]\dots A_t[i]}$. From now on, in a slight abuse of notation, we switch between these interpretations depending on which one is more convenient.
For example, we may write either $\floor{\lg x} + 1$ or $\absolute{x}$ to refer to the length of bitmask $\senc{A_1[i]A_2[i]\dots A_t[i]}$.
In the encoding $\senc{\zip(A_1, \dots, A_t)}$ of the zipped string, we will encode $x$ as a literal token $\senc{x} = \senc{\senc{A_1[i]A_2[i]\dots A_t[i]}}$. For this nested notation $\senc{\senc{\dots}}$, the inner value $\senc{\dots}$ is interpreted as an integer.

Now we analyze $\absolute{x}$ for the case where all strings $A_j$ are over $[0\dd 2^u)$ with $u \geq 1$.
Every $j\in[1\dd t]$ with $A_j[i] = 0$, even if it is not adjacent to other zeros in $A_1[i]A_2[i]\dots A_t[i]$, contributes at most two bits to~$x$.
Every $j\in[1\dd t]$ with $A_j[i] > 0$ contributes its Elias-$\gamma$ code, which consists of $2\floor{\lg A_j[i]} + 1 < 2u$ bits. 
Hence, $x$ consists of less than $2ut$ bits, which implies:

\begin{observation}\label{obs:zipped_alphabet}
Let $A_1, \dots, A_t$ be strings of equal length, each over alphabet $[0\dd 2^u)$ for a positive integer $u$. 
Then, $\zip(A_1, \dots, A_t)$ is over alphabet $[0\dd 2^{2ut})$.
\end{observation}

Zipping asymptotically preserves the size of sparse encodings for a constant number of strings.

\restatezippedsize*

\begin{proof}
Let $Z = \zip(A_1, \dots, A_t)$. 
If $Z[i\dd j)$ is a lengthwise maximal run of zeros in $Z$, then $A_h[i \dd j)$ is all-zero for every $h \in [1\dd t]$, and there must be at least one $h$ such that $A_h[j] \neq 0$ (or $j = n$). Hence, for each zero-run in $Z$, there is at least one zero-run in some $A_h$ that has the same ending position and is at least as long as the zero-run in $Z$. It follows that encoding the zero-runs of $Z$ takes at most as many bits as encoding the zero-runs of $A_1, \dots, A_t$.

Now consider a non-zero symbol of $Z$, i.e., a literal $Z[i] = \senc{A_1[i]\dots A_t[i]}$.
In $\senc{Z}$, the contribution of $Z[i]$ is $2 \cdot \absolute{\senc{A_1[i]\dots A_t[i]}} + 2$ bits. We can ignore constant factors, and hence it suffices to analyze the bit complexity of $\senc{A_1[i]\dots A_t[i]}$.
We compare $\senc{A_1[i]\dots A_t[i]}$ against the sum of all the $\absolute{\senc{A_j[i]}}$ for which it holds $A_j[i] \neq 0$. It is easy to see that a non-zero symbol $A_j[i] = x$ contributes exactly $2\floor{\lg x} + 2$ bits to both $\senc{A_1[i]\dots A_t[i]}$ and $\senc{A_j[i]}$.
The zero-runs within $\senc{A_1[i]\dots A_t[i]}$ are encoded in a constant number of bits (due to the constant number of strings $t$), and hence we can ignore them. 
We have shown that $\absolute{\senc{A_1[i]\dots A_t[i]}}$ equals $\sum_{j \in [1\dd t], A_j[i] \neq 0} \absolute{\senc{A_j[i]}}$ up to constant factors, from which the claim follows.
\end{proof}

\begin{lemma}\label{lem:reduce_multi_to_single_nonsparse}
Consider a transducer over alphabet $[0\dd \sigma)$ with states $[0\dd q)$, where $\sigma, q \in 2^{\Oh(w)}$, and a constant number $t$ of input streams.
Assume that the transition function can be evaluated in constant time.
Then, there is a single-stream transducer with states $[0\dd q)$ that, for input $\zip(S_1, \dots, S_t)$, produces the same output as the multi-stream transducer for input $S_1, \dots, S_t$.
For every ${N \in [2\dd 2^w]}$, after an $\Oh(N)$-time preprocessing, a transition from any state of the single-stream transducer with input symbol $x$ takes $\Oh(1 + \lg (1 + x) / \lg N)$ time.
\end{lemma}

\begin{proof}
By \cref{obs:zipped_alphabet}, the single-tape transducer has input alphabet of size $2^{\Oh(w)}$.
It uses the same states as the multi-stream transducer and merely simulates its transition function.
If the input symbol of the single-stream transducer is $x \neq 0$, then it is some sparse encoding ${x = \senc{x_1\cdots x_t}}$. We obtain $x_1\cdots x_t$ using \crefbullet{lem:decode_sparse_naive} in $\Oh(1 + \lg x / \lg N)$ time (due to $t = \Oh(1)$).
Otherwise, we use $x_1 = x_2 = \dots = x_t = 0$. 
Now, we evaluate the transition function of the multi-stream transducer for the current state with input symbols $x_1, \dots, x_t$, which takes constant time and results in the new state and the output symbol. To allow for $x = 0$, we write the complexity as $\Oh(1 + \lg(x + 1) / \lg N)$.
\end{proof}

\begin{corollary}\label{lem:accelerate_multi_stream_on_zipped}
	Consider a transducer over alphabet $[0\dd \sigma)$ with states $[0\dd q)$, where $\sigma, q \in 2^{\Oh(w)}$, and a constant number $t$ of input streams.	
	Assume that the transition function can be evaluated in constant time.
	For every $N \in [2 \dd 2^w)$, after an $\Oh(qN)$ time preprocessing, the following holds.

    Let $S_1, \dots, S_t$ be input strings of length $n \in 2^{\Oh(w)}$ for which the transducer produces output string $T$. Then, in ${\Oh(1 + {(\sum_{j = 1}^t \absoluteenc{S_j} + \absoluteenc{T}) / \lg N})}$ time, one can compute $\senc{T}$ from $\senc{\zip(S_1, \dots, S_t)}$.
\end{corollary}

\begin{proof}
Due to $\sigma \in 2^{\Oh(w)}$, and due to the constant number of streams, \cref{obs:zipped_alphabet} implies that $\zip(S_1, \dots, S_t)$ is over integer alphabet of size $2^{\Oh(w)}$.
Hence, we can use the transducer from \cref{lem:reduce_multi_to_single_nonsparse} and preprocess it with \cref{lem:accelerate_single_stream}. The query time bound then follows from \cref{lem:zipped_encoding_size}.
\end{proof}

\subsection{Zipping Sparse Encodings}

In this section, we show how to efficiently zip strings in sparse encoding.
Then, the main result stated in \cref{lem:accelerate_multi_stream} is a direct corollary of \cref{lem:fast_multi_zip} (below) and \cref{lem:accelerate_multi_stream_on_zipped}.

\begin{theorem}\label{lem:fast_pair_zip}
For every $N \in [2 \dd 2^w]$, after an $\Oh(N)$-time preprocessing, the following holds.
Let $A_1, A_2 \in [0\dd \sigma)^n$ be strings with $n,\sigma \in {2^{\Oh(w)}}$.
Given $\senc{A_1}$ and $\senc{A_2}$, one can compute $\senc{\zip(A_1, A_2)}$ in $\Oh(1 + (\absoluteenc{A_1} + \absoluteenc{A_2}) / \lg N)$ time.
\end{theorem}

\begin{proof}
For some small constant $\eps > 0$, we use $m, M \in \mathbb Z_+$ with $2^m = M = \Theta(N^\eps)$.
We can afford preprocessing time $\Oh(\poly(M))$ and we aim to gain speedup $m = \Theta(\lg N)$.
We define $A_3 = \zip(A_1, A_2)$ and sparse encodings $B_j = \senc{A_j}$ for $j \in [1\dd 3]$.
We have to show how to compute $B_3$ from $B_1$ and $B_2$. 

We use a precomputed lookup table $L$ to move through the merging process in steps of encoded blocks of length approximately $m$.
Since runs of zeros are not necessarily aligned in $A_1$ and $A_2$, we index the lookup table not only with the upcoming bits of $B_1$ and $B_2$ but also with some number of already decoded but not yet processed zeros.
The table stores, for every $X,Y \in \{0,1\}^{m}$ and $(z_1, z_2) \in \{ (0, z), (z, 0) \mid z \in \Zz \}$, an entry $L[X, Y, z_1, z_2] = (x, y, z_1', z_2', \ell, r, c, C)$.
(We will show later that it suffices to consider $\Oh(\poly(M))$ values of $(z_1,z_2)$.)
Here, $x,y \in [0\dd m]$ must satisfy the following constraints:

\begin{enumerate}[(1)]
\item\label{prop:defineX} The integer $x$ indicates that $X[0 \dd x) = \senc{\mathcal X'}$ for some string $\mathcal X'$. Let $\mathcal X = 0^{z_1} \cdot \mathcal X'$.
\item\label{prop:defineY} The integer $y$ indicates that $Y[0 \dd y) = \senc{\mathcal Y'}$ for some string $\mathcal Y'$. Let $\mathcal Y = 0^{z_2} \cdot \mathcal Y'$.
\item\label{prop:danglingzeros} If $\absolute{\mathcal X} < \absolute{\mathcal Y}$, then $\mathcal Y$ has suffix $0^{\absolute{\mathcal Y} - \absolute{\mathcal X}}$. If $\absolute{\mathcal X} > \absolute{\mathcal Y}$, then $\mathcal X$ has suffix $0^{\absolute{\mathcal X} - \absolute{\mathcal Y}}$.
\end{enumerate}

For each entry, among all the assignments of $x$ and $y$ that satisfy the constraints, we choose one that maximizes $x + y$. (Trivially, there is at least one suitable assignment, i.e., $x = y = 0$).
Once we have fixed $x$ and $y$, we can define the rest of the entry as follows:

\begin{itemize}
\item The integer ${\ell \in [0\dd \min(\absolute{\mathcal X}, \absolute{\mathcal Y})]}$ is the largest value such that $\mathcal X[0\dd \ell) = \mathcal Y[0\dd \ell) = 0^\ell$, and ${r \in [0\dd \min(\absolute{\mathcal X}, \absolute{\mathcal Y})]}$ is the smallest value such that both $\mathcal X[r \dd \absolute{\mathcal X})$ and $\mathcal Y[r \dd \absolute{\mathcal Y})$ are all-zero. Note that $\ell \geq r$ if and only if both $\mathcal X$ and $\mathcal Y$ are all-zero, which is the case if and only if $r = 0$.
\item If $r > 0$, then $C = \senc{\zip(\mathcal X[\ell\dd r), \mathcal Y[\ell\dd r))}$ and $c = \absolute{C}$. If $r = 0$, then $C$ is empty and $c = 0$.
\item Finally, $z_1' = \absolute{\mathcal X} - r$ and $z_2' = \absolute{\mathcal Y} - r$ indicate the number of trailing zeros of $\mathcal X$ and $\mathcal Y$ that have not been encoded in $C$.
\end{itemize}

We cannot afford to compute the table for arbitrary $z_1,z_2$. For efficiently implementing the case when $z_1 \geq M$, we analyze the strings $\mathcal X, \mathcal Y$ in the definition of entry 
\[L[X, Y, z_1, 0] = (x, y, z_1', z_2', \ell, r, c, C).\]
Since $Y$ is of length at most $m$, the string $\mathcal Y' = \mathcal Y$ must be of length less than $M$ (see \cref{obs:sparse_size_lower}).
The string $\mathcal X$ has prefix $0^{z_1}$ and is thus of length at least $M$.
Due to $\absolute{\mathcal Y} < M \leq \absolute{\mathcal X}$ and $\mathcal X[0\dd M) = 0^M$, constraint~\ref{prop:danglingzeros} implies that $\mathcal X$ is all-zero, which implies that $x$ is the (unique and) maximal integer such that $X[0\dd x)$ encodes a zero-run token $0^{x'}$ (possibly $x = x' = 0$). For this value $x$, and for any $y$ such that $Y[0\dd y)$ is a sparse encoding, constraint~\ref{prop:danglingzeros} is trivially satisfied.
Consequently, it must be that $y$ is the maximal value such that $Y[0\dd y)$ is a sparse encoding. Note that $x$ depends only on $X$, and $y$ depends only on~$Y$, but neither of the two depends on $z_1$.
Since $\min(\absolute{\mathcal X},\absolute{\mathcal Y}) = \absolute{\mathcal Y}$ is independent of $z_1$, also the values $\ell$, $r$, and $z_2'$ are independent of $z_1$. Finally, we have $z_1' = \absolute{\mathcal X} - r = z_1 + x' - r$, where $x'$ depends on $X$, but not on~$z_1$.
We have shown that the entire entry, apart from $z_1'$, is independent of $z_1$.

Instead of looking up $L[X, Y, z_1, 0] = (x, y, z_1 + x' - r, z_2', \ell, r, c, C)$, we can look up the almost identical entry $L[X, Y, M, 0] = (x, y, M + x' - r, z_2', \ell, r, c, C)$. Then, we can compute $z_1' = (M + x' - r ) - M + z_1$ in constant time. Hence, we can simulate access to entries with $z_1 \geq M$ without explicitly precomputing them. A symmetric approach can be used for $z_2 \geq M$. Therefore, it suffices to compute the table for $(z_1, z_2) \in \{ (0, z), (z, 0) \mid z \in [0\dd M] \}$. 
The table has $\Oh(\poly(M))$ entries and each entry can be computed naively by brute-force in $\Oh(\poly(M))$ time.

\subparagraph{Running the merging procedure.} Now we are ready to describe the merging procedure.
We proceed in a sequence of steps. Before each step, we maintain the following information.

 	\begin{itemize}
        \item For non-negative integers $b_1, b_2$ (initialized with $0$), we have already processed $B_1[0\dd b_1)$ and $B_2[0\dd b_2)$, and we have written bitmask $B_3$; each of the three strings is a sparse encoding (i.e., we did not split a token), and $B_3$ does not end with a zero-run token.
        \item We know non-negative integers $a, z_1, z_2, z_3$ (initialized with $0$) such that
        \begin{itemize}
            \item $B_1[0\dd b_1)$ encodes $A_1[0\dd a + z_1) = A_1[0\dd a) \cdot 0^{z_1}$, and
            \item $B_2[0\dd b_2)$ encodes $A_2[0\dd a + z_2) = A_2[0\dd a) \cdot 0^{z_2}$, and
            \item $B_3$ encodes $A_3[0\dd a - z_3)$, and $A_3[0\dd a) = A_3[0\dd a - z_3) \cdot 0^{z_3}$.
        \end{itemize}
    \end{itemize}
    
At the beginning of a step, let $z = \min(z_1, z_2)$. If $z > 0$, then we decrease~$z_1$ and $z_2$ by $z$, and increase~$a$ and $z_3$ by~$z$, which does not affect the above invariants and ensures that at least one of $z_1$ and $z_2$ is zero.

\subparagraph{Using the lookup table.}
We retrieve \[L[B_1[b_1 \dd b_1 + m), B_2[b_2 \dd b_2 + m), z_1, z_2] = (b_1', b_2', z_1', z_2', \ell, r, c, C)\]
and assign $b_1 \gets b_1 + b_1'$, $b_2 \gets b_2 + b_2'$, $z_1 \gets z_1'$, and $z_2 \gets z_2'$. (We use the same padding as in the proof of \cref{lem:accelerate_single_stream} when we are close to the end of $B_1$ or $B_2$.)
We terminate the step if $c = 0$.
Otherwise, we first append $\senc{0^{z_3 + \ell}}$ (assuming $z_3 + \ell > 0$) and then $C$ to $B_3$, and conclude the step by assigning $a \gets a + r$ and $z_3 \gets 0$.
Due to the constraints satisfied by the lookup table, this correctly maintains the invariants of the algorithm.
Producing $\senc{0^{z_3 + \ell}}$ takes $\Oh(1 + \absolute{\senc{0^{z_3 + \ell}}} / m)$ time with \cref{lem:encode_elias}.
The term $\Oh(\absolute{\senc{0^{z_3 + \ell}}} / m)$ amortizes to $\Oh(1 / m)$ time per bit of $B_3$, summing to $\Oh(\absoluteenc{A_3} / m)$ time overall.
Apart from that, the time per step is constant.
Later, we show that the number of steps is $\Oh((\absoluteenc{A_1} + \absoluteenc{A_2}) / m)$. Hence, we can afford constant time per step.

\subparagraph{Processing large tokens.}
If the table returns $b_1' = b_2' = 0$, then we make no progress and need to pursue a different approach.
In this case, we first observe that neither $B_1[b_1\dd \absolute{B_1})$ nor $B_2[b_2\dd \absolute{B_2})$ starts with a zero-run token of $m$ bits or less.
Otherwise, the table would have been able to advance $B_1$ or $B_2$ by at least this token.
If $B_1[b_1\dd \absolute{B_1})$ starts with a zero-run token $\senc{0^x}$, i.e., if $B_1[b_1] = 0$, we obtain $x$ using \cref{lem:decode_elias}. 
Then, we increase $b_1$ by $\absolute{\senc{0^x}}$ and $z_1$ by $x$, which concludes the step.
This takes overall $\Oh(1 + \lg x / m)$ time, which amortizes to $\Oh(1 / m)$ time per decoded bit (due to $\lg x = \Omega(m)$).
If $B_2[b_2\dd \absolute{B_2})$ starts with a zero-run token, then we proceed symmetrically.

It remains to consider the case where both $B_1[b_1\dd \absolute{B_1})$ and $B_2[b_2\dd \absolute{B_2})$ start with a literal. If $z_3 > 0$, we append $\senc{0^{z_3}}$ to $B_3$ and assign $z_3 \gets 0$.
If $z_1 = 0$, then $B_1[b_1\dd \absolute{B_1})$ starts with $\senc{A_1[a]}$. We obtain $A_1[a]$ using \cref{lem:decode_elias} and increase $b_1$ by $\absolute{\senc{A_1[a]}}$.
If $z_2 = 0$, then $B_2[b_2\dd \absolute{B_2})$ starts with $\senc{A_2[a]}$. We obtain $A_2[a]$ using \cref{lem:decode_elias} and increase $b_2$ by $\absolute{\senc{A_2[a]}}$.
Now we proceed according to the following three cases. If $z_1 = z_2 = 0$, then we obtain $D = \senc{A_1[a]\cdot A_2[a]}$.
If $z_2 > 0$ (and hence $z_1 = 0$), then we instead produce $D = \senc{A_1[a] \cdot 0}$ and assign $z_2 \gets z_2 - 1$.
Finally, if $z_1 > 0$ (and hence $z_2 = 0$), then we produce $D = \senc{0 \cdot A_2[a]}$ and assign $z_1 \gets z_1 - 1$.
We conclude the step by computing $\senc{D}$, appending it to $B_3$, and assigning $a \gets a + 1$.

Obtaining $A_1[a]$ with \cref{lem:decode_elias} takes $\Oh(1 + \absoluteenc{A_1[a]} / m)$ time, where the term $\Oh(\absoluteenc{A_1[a]} / m)$ amortizes to $\Oh(1/m)$ per bit of $B_1$. The same argument holds for $A_2[a]$ with respect to the bits of $B_2$.
Computing $D$ and then $\senc{D}$ takes $\Oh(1 + \absolute{D} / m) = \Oh(1 + \absoluteenc{D} / m)$ time using \cref{lem:encode_elias}, where $\Oh(\absoluteenc{D} / m)$ amortizes to $\Oh(1/m)$ per bit of $B_3$.

\subparagraph{Counting the steps.}
Now we count the number of times we access the lookup table.
For any two consecutive lookups, if at least one of them returns $b_1' + b_2' = 0$, then we will process a large token consisting of over $m$ bits. 
Hence, we can run into this case at most $\Oh((\absoluteenc{A_1} + \absoluteenc{A_2}) / m)$ times.
Now consider two consecutive lookups such that neither of them leads to $b_1' + b_2' = 0$.
In the first one, we retrieve entry%
\[L[B_1[b_1 \dd b_1 + m), B_2[b_2 \dd b_2 + m), z_1, z_2] = (b_1', b_2', z_1', z_2', \ell, r, c, C),\]
and in the second one we retrieve, for values $z_1'' = z_1' - \min(z_1', z_2')$ and $z_2'' = z_2' - \min(z_1', z_2')$, entry%
\[L[B_1[b_1 + b_1' \dd b_1 + b_1' + m), B_2[b_2 + b_2' \dd b_2 + b_2' + m), z_1'', z_2''] = (b_1'', b_2'', z_1''', z_2''', \ell', r', c', C').\]
We claim that either $b_1' + b_1'' \geq m$ or $b_2' + b_2'' \geq m$. Otherwise, there would have been values $\ell'', r'', c'', C''$ such that the first table entry could have been 
\[L[B_1[b_1 \dd b_1 + m), B_2[b_2 \dd b_2 + m), z_1, z_2] = (b_1' + b_1'', b_2' + b_2'', z_1''', z_2''', \ell'', r'', c'', C'').\]
Here, if $\ell \neq r$ then $\ell'' = \ell$, and otherwise $\ell'' = \ell'$. Similarly, if $r' \neq \ell'$ then $r'' = r'$, and otherwise $r'' = r$.
Finally, if $\ell'' \neq r''$, then $C'' = C \cdot \senc{0^x} \cdot C'$. 
It can be readily verified that the entry satisfies the conditions required by the lookup table. 
This contradicts the fact that we choose each entry such that it maximizes the progress.
Hence, in any two consecutive steps, we either encounter the case where we cannot use the table to make any progress, or we advance by at least $m$ bits in at least one of $B_1$ and $B_2$.
\end{proof}

\label{sec:lem:fast_multi_zip}
\restatemultizip*

\begin{proof}
The algorithm is recursive. For the base case $t = 1$, we describe a single-stream transducer that computes $\zip(A_1)$ from $A_1$. The transducer has a single state $s$, and the transition function is defined by $\delta(s, 0) = (s,0)$ and $\forall_{x \in [1\dd \sigma)}\; \delta(s, x) = (s,\senc{x})$, which can be evaluated in $\Oh(1 + \lg x / \lg N)$ time using \cref{lem:encode_elias}. 
Hence, we can apply \cref{lem:accelerate_single_stream} with preprocessing parameter $N$. Then, we can obtain $\senc{\zip(A_1)}$ from $\senc{A_1}$ in $\Oh(\absoluteenc{A_1} / \lg N) = \Oh(\absoluteenc{\zip(A_1)}/ \lg N)$ time.

For $t > 1$, we recursively compute $\senc{\zip(A_1, \dots, A_{t - 1})}$. By \cref{obs:zipped_alphabet}, the string $\zip(A_1, \dots, A_{t - 1})$ is over integer alphabet of size $2^{\Oh(w)}$. Thus, we can use \cref{lem:fast_pair_zip} to compute $\senc{\zip(\zip(A_1, \dots, A_{t-1}), A_t)}$ from $\senc{\zip(A_1, \dots, A_{t - 1})}$ and $\senc{A_t}$ in time $\Oh(1 + (\absoluteenc{\zip(A_1, \dots, A_{t-1})} + \absoluteenc{A_t}) / \lg N)$. By \cref{lem:zipped_encoding_size}, the time complexity can be written as $\Oh(1 + \sum_{j = 1}^t \absoluteenc{A_j} / \lg N)$. 
Finally, we describe a single-stream and single-state transducer that produces $\senc{\zip(A_1, \dots, A_t)}$ given $\senc{\zip(\zip(A_1, \dots, A_{t-1}), A_t)}$. On input zero, it produces output zero. In the most general case, a non-zero input symbol is of the form $\senc{\senc{A_1[i]\dots A_{t - 1}[i]} \cdot A_t[i]}$ for some $i \in [0\dd n)$, and the transition outputs $\senc{A_1[i]\dots A_t[i]}$.
Using \crefbullet{lem:encode_sparse_naive} (with constant $n$), we can decode $\senc{\senc{A_1[i]\dots A_{t - 1}[i]} \cdot A_t[i]}$ into plain symbols $A_1[i], \dots, A_t[i]$ and then obtain $\senc{A_1[i]\dots A_t[i]}$ in \[\Oh(\absoluteenc{A_1[i]\dots A_t[i]}) = \Oh(\absolute{\senc{\senc{A_1[i]\dots A_{t - 1}[i]} \cdot A_t[i]}})\] time.
Hence, we can apply \cref{lem:accelerate_single_stream} with preprocessing parameter $N$. Then, we can obtain $\senc{\zip(A_1\dots A_t)}$ from $\senc{\zip(\zip(A_1, \dots, A_{t-1}), A_t)}$ by running the transducer in $\Oh(1 + \sum_{j = 1}^t \absoluteenc{A_j} / \lg N)$ time (where the time again follows from \cref{lem:zipped_encoding_size}).
There are $t=\Oh(1)$ levels of recursion, and the overall time is as claimed.

\end{proof}

\section{Faster Recompression Queries}
\label{sec:appendixrecompression}

For a more efficient implementation of recompression, we report the sets in the following representation.
For positive integer $j$, we define the array $\mathcal B_j[0 \dd n)$ as $\mathcal B_j[0] := 0$ and \[\forall_{i \in [1 \dd n)}\; \mathcal B_j[i] := \max(\{ k - j + 1 \mid i \in \B_k\} \cup \{0\}).\]%
Since $\B_k$ is of size $\Oh(n / \lambda_k)$, it holds $\sum_{i \in [0\dd n), \mathcal B_j[i] = k} \lg(1 + \mathcal B_j[i]) = \Oh(n \lg (k - j + 2) / \lambda_k)$.
Therefore, we have $\sum_{i = 0}^{n - 1} \lg(1 + \mathcal B_j[i]) = \Oh(\sum_{k = j}^\infty n \lg (k - j + 2) / \lambda_k) = \Oh(n / \lambda_j)$.
Hence, by \cref{obs:sparse_size}, the non-zero entries contribute $\Oh(n / \lambda_j)$ bits to $\absolute{\senc{\mathcal B_j}}$, while the $\Oh(n / \lambda_j)$ runs of zeros contribute $\Oh(n \lg (1 + \lambda_j) / \lambda_j)$ bits.

\begin{observation}\label{obs:sizeofBj}
    It holds $\absolute{\senc{\mathcal B_j}} = \Oh(n \lg (1 + \lambda_j) / \lambda_j)$.
\end{observation}

\begin{restatable}{lemma}{restatecomputeBZero}\label{lem:faster_recomp_intermediate}
    The encoding $\senc{\mathcal B_0}$ can be computed in $\Oh(n / \log_\sigma n)$ time from $T$.
\end{restatable}

\begin{proof}
    We will show how to compute the encodings of $\mathcal B_K$ and $\mathcal B'[0 \dd n)$, where the latter is defined by $\forall_{i \in [0\dd n)}\;\mathcal B'[i] := \min(K, \mathcal B_0[i])$.
    Assume for now that $\senc{\mathcal B_K}$ and $\senc{\mathcal B'}$ are given.
    There is a single-state transducer that produces $\mathcal B_0$. It has two input streams, the first reading $\mathcal B'$, the second reading $\mathcal B_K$. On input $(x, y)$, it outputs $x$ if $y = 0$ and $y + K$ otherwise.
    Hence, after preprocessing the transducer in $\Oh(\sqrt{n})$ time with \cref{lem:accelerate_multi_stream}, we can compute $\senc{\mathcal B_0}$ in $\Oh((\absolute{\senc{\mathcal B'}} + \absolute{\senc{\mathcal B_K}} + \absolute{\senc{\mathcal B_0}})/\lg n)$ time. This is bounded by $\Oh(n / \lg n)$ due to $\absolute{\senc{\mathcal B'}} \leq \absolute{\senc{\mathcal B_0}}$ and \cref{obs:sizeofBj}.

    It remains to compute $\senc{\mathcal B'}$ and $\senc{\mathcal B_K}$. For $\senc{\mathcal B_K}$, we obtain each of $\B_K, \dots, \B_q$ as a sorted list using \cref{lem:bk-explicit}. The number of elements in all lists and the required time are $\Oh(\sum_{j = K}^q (1 + n/\lambda_j)) = \Oh(n / \lambda_K) = \Oh(n / \log_\sigma n)$.
    For each $i \in [0\dd n)$ and $j \in [K\dd q]$, if $i \in \B_j$, we produce a pair $(i, j - K + 1)$. Using radix sort, we can filter the pairs such that, for each $i$, we retain only the maximal $j - K + 1$.
    The result is a list of all non-zero entries of $\mathcal B_K$ sorted by position, from which we can obtain $\senc{\mathcal B_K}$ in $\Oh(1 + \absolute{\B_K}) = \Oh(1 + n / \lambda_K) = \Oh(n / \log_\sigma n)$ time with \cref{lem:encode_list}.

    For computing $\senc{\mathcal B'}$, we use a more elaborate version of the idea in \cref{lem:bitmask}.
    We use \cref{lem:ck} to obtain the boundary contexts $\tB_1, \dots,\tB_{K}$ as constant-time oracles in $\Oh(n / \log_\sigma n)$ time.
    We construct a lookup table $L[0\dd \floor{\sqrt{n}})$. For every $S \in [0\dd \sigma)^*$ of length $\absolute{S} \in (2\alpha_K\dd3\alpha_K]$, entry $L[\int(S)]$ stores the following information. Let $H \in [0\dd K]^{\absolute{S} - 2\alpha_K}$ be defined such that, for $i \in [0\dd \absolute{H})$, \[H[i] = \max(\{ j \in [1\dd K] \mid S[i + \alpha_K - \alpha_{j - 1}\dd i + \alpha_K + \alpha_{j - 1}) \in \tB_{j - 1} \} \cup \{0\}).\]
    Then $L[\int(S)] = (\absolute{\senc{H}}, \senc{H})$. 
    Intuitively, $H$ stores the $\mathcal B'$ information for the central fragment of $S$, ignoring the respectively first and last $\alpha_K$ positions of $S$.
    Due to $K, \alpha_K \ll \lg n$ and the constant-time oracles, the table can be computed in $\tilde\Oh(\sqrt{n})$ time.
    Finally, we concatenate the encodings obtained by accessing the table for all the text blocks $T[(h - 1) \cdot \alpha_K \dd \min(n + \alpha_K, (h + 2) \cdot \alpha_K))$ with $h \in [0\dd \ceil{n / \alpha_K})$. (Recall that $T$ is padded with $\texttt\textdollar^{\alpha_K}$ on either side.)
    Note that each fragment $T[h \cdot \alpha_K \dd \min(n, (h + 1) \cdot \alpha_K))$ is the central fragment (not ignored when computing $H$) of exactly one of the considered blocks, which implies that the concatenation of the encodings indeed results in $\senc{\mathcal B'}$.
         Particularly, since $\mathcal B'$ is entirely non-zero except for $\mathcal B'[0]$, we do not need special handling of zero-run tokens.  Whenever we have to append some $\senc{H}$ with $\absolute{\senc{H}} > \lg n$, we append the entry in a word-wise manner in $\Oh(1 + \absolute{\senc{H}} / \lg n)$ time. Hence, the total time is $\Oh(n / \alpha_K + \absolute{\senc{\mathcal B'}} / \lg n) = \Oh(n / \log_\sigma n)$.
\end{proof}

\begin{restatable}{lemma}{restatecomputeBAll}\label{lem:faster_recomp}
    Given $\senc{\mathcal B_0}$, we can compute the sequence $\senc{\mathcal B_1}, \dots, \senc{\mathcal B_q}$ in $\Oh(n / \lg n)$ time.
\end{restatable}

\begin{proof}
     There is a straightforward single-state and single-stream transducer that, for any~$j$, produces ${\mathcal B_{j + 1}}$ from ${\mathcal B_j}$. On input symbol $x$, the transducer outputs symbol ${\max(0, x - 1)}$. Using \cref{lem:accelerate_multi_stream} with preprocessing parameter $\sqrt{n}$, we prepare the transducer for efficiently handling sparse encodings.
     It then takes $\Oh((\absolute{\senc{\mathcal B_{j}}} + \absolute{\senc{\mathcal B_{j + 1}}}) / \lg n)$ time to produce $\senc{\mathcal B_{j + 1}}$ from $\senc{\mathcal B_{j}}$. By \cref{obs:sizeofBj}, this is bounded by $\Oh(n \lg (1+\lambda_j) / (\lambda_j \cdot \lg n))$. Summing over all $j$, the time is $\Oh(\sum_{j = 0}^\infty n \lg (1+\lambda_j) / (\lambda_j \cdot \lg n)) = \Oh(n / \lg n)$, and the preprocessing time is $\Oh(\sqrt{n}) \subset \Oh(n / \lg n)$. 
\end{proof}

\section{Faster \texorpdfstring{\boldmath$\tau$\unboldmath}{tau}-Run Queries}
\label{sec:appendixruns}

\begin{restatable}{theorem}{restatefasterruns}\label{lem:faster_runs}
    Let $c \in \Zp$ be constant. A text $T\in [0\dd \sigma)^{n}$ can be preprocessed in $\Oh(n/\log_\sigma n)$ time for the following type of query. Given $\tau \in [1\dd n]$ and $\ell \in [\tau, c\tau]$, return $\senc{S}, \senc{E}$, where $S, E \in \{0,1\}^n$ indicate the start and end positions of the runs in $\RUNS_{\ell, \floor{\tau/3}}(T)$.
    The output consists of $\Oh(\frac n\tau \lg \tau)$ bits, and a query takes $\Oh(\frac{n \lg \tau}{\tau \lg n})$ time.
\end{restatable}

\begin{proof}
    We split the construction depending on $\tau \leq P := \floor{\log_\sigma n / (16c + 4) } = \Theta(\log_\sigma n)$. We only describe the solution for $S$ because the solution for $E$ works analogously.
    \subparagraph{Preprocessing for large~\boldmath$\tau$.\unboldmath}
    If $\tau \geq P$, we build $\Oh(\lg n)$ data structures, one for each range $[\floor{1.1^j}\dd\floor{1.1^{j + 1}})$ with integer $j$ satisfying $P / 1.1 <\floor{1.1^j} \leq n$. We focus on a fixed $j$.
    We obtain $R_{j} := \RUNS_{\floor{1.1^j}, \floor{{0.4 \cdot 1.1^j}}}(T)$ as a sorted list using \cref{prp:runs-explicit} in $\Oh(\frac n{1.1^j})$ time, which sums to $\Oh(n / \log_\sigma n)$ over all the $j$.
    
    We define arrays $S_{j,\per}[0\dd n)$ and $S_{j,\len}[0\dd n)$, where $S_{j,\per}[i]$ and $S_{j,\len}[i]$ respectively store the period and length of the unique run in $R_{j}$ starting at position $i$, or $0$ if no such run exists. The uniqueness is guaranteed by \cref{fct:overlap}. Next, we analyze $\absolute{\senc{S_{j,\len}}}$.  As evident from \cref{prp:runs-explicit}, there are $\Oh(n/1.1^j)$ runs in $R_j$. Each run is of length at least $1.1^j$ and has period $\leq 0.4 \cdot 1.1^j$, and hence \cref{fct:overlap} implies that the lengths sum to $\Oh(n)$. By \cref{obs:sparse_size}, it holds $\absolute{\senc{S_{j, \len}}} = \Oh(\frac{n}{1.1^j} \cdot \lg \frac{n}{n / 1.1^j}) = \Oh(\frac n{1.1^j} \cdot j)$.
    The period of a run is less than its length, and thus the encoding of $S_{j,\per}$ clearly cannot be larger than the one of $S_{j,\len}$.

    For each $j$, we store $R_j$, $\senc{S_{j,\per}}$, and $\senc{S_{j,\len}}$. Each array can be obtained by scanning the sorted $R_j$ and applying \cref{lem:encode_list}. This takes $\Oh(\absolute{R_j}) = \Oh(\frac n{1.1^j})$ time, summing to $\Oh(n / \log_\sigma n)$ for all $j$. We also store the sequence of pairs $(\floor{1.1^0},0), (\floor{1.1^1}, 1), \dots, (\floor{1.1^{\ceil{\log_{1.1} n}}}, \ceil{\log_{1.1} n})$.

    \subparagraph{Answering queries for large~\boldmath$\tau$.\unboldmath}
    We obtain the unique $j$ such that $\tau \in [\floor{1.1^j}\dd \floor{1.1^{j + 1}})$ by scanning the sequence of pairs from right to left in $\ceil{\log_{1.1} n} - j + 1 = \Oh(\lg \frac n\tau) = \Oh(\frac{n \lg \tau}{\tau \lg n})$ time. Note that $\RUNS_{\ell, \floor{\tau/3}} (T) \subseteq \RUNS_{\floor{1.1^j}, \floor{{0.4 \cdot 1.1^j}}}(T) = R_j$ due to $\tau / 3 < 0.4 \cdot 1.1^j$.
    Hence, $S$ can be obtained from an appropriately filtered $R_j$. 

    If $\tau \leq \sqrt{n} / \lg n$, we can afford an extra $\Oh(\sqrt{n})\subseteq \Oh(\frac{n \lg \tau}{\tau \lg n})$ term in the query time.  Then, we obtain~$S$ with the following single-state transducer. It reads $S_{j, \per}$ and $S_{j,\len}$ on two input streams. Whenever $S_{j,\len}[i] \geq \ell$ and $S_{j,\per}[i] \leq \floor{\tau / 3}$, it outputs $S[i] = 1$, and otherwise $S[i] = 0$. We preprocess the transducer with \cref{lem:accelerate_multi_stream} and preprocessing parameter $N = \Theta(\sqrt{n})$. It then runs in $\Oh((\absolute{\senc{S_{j, \per}}} + \absolute{\senc{S_{j, \len}}} + \absolute{\senc{S}}) / \lg n)$ time, which is bounded by $\Oh(\frac n{1.1^j} \cdot j / \lg n) = \Oh(\frac{n \lg \tau}{\tau \lg n})$.
    (We have already shown that the input encodings are of size $\Oh(\frac n{1.1^j} \cdot j)$ bits, and this bound also applies to $\senc{S}$).

    If $\tau > \sqrt{n} / \lg n$, then $\frac{\lg n}{\lg \tau} = \Oh(1)$ and we can afford $\Oh(n / \tau)$ query time. Hence, we scan and filter the $\Oh(\frac n{1.1^j}) = \Oh(\frac n\tau)$ runs in $R_j$ and produce $\senc{S}$ using \cref{lem:encode_list}.

    \subparagraph{Preprocessing for small~\boldmath$\tau$.\unboldmath}
    Now we consider the runs of period less than $P$.
    We say that a run is \emph{relevant} if its period is less than $P$ and its length is at least three times its period.
    Note that a query with $\floor{\tau / 3} < P$ is only concerned with relevant runs.
    
    In a moment, we show how to precompute the following array $R[0\dd n)$ in sparse encoding. Each entry $R[i]$ contains $\senc{p_1\ell_1  p_2\ell_2 \dots p_h\ell_h}$, where $p_1 < p_2 < \dots < p_h$ are the periods of all relevant runs starting at position $i$, and $\ell_1, \dots, \ell_h$ are the corresponding lengths.
    For the lengths, however, we truncate their values to be at most $4cP$.
    For any query $\tau, \ell$ with $\floor{\tau / 3} < P$, it holds $\ell  \leq c\tau \leq 4cP$, and thus we gather enough information to determine whether the length of each run exceeds $\ell$. Whenever there is no run starting at some position~$i$, we store $R[i] = \senc{0}$ instead. Note that $R$ contains only non-zero entries.

    If we focus on the relevant runs with period in range $[2^j, 2^{j + 1})$ for some integer $j$, then \cref{fct:overlap} implies that their lengths and periods sum to $\Oh(n)$, and that there are only $\Oh(n / 2^j)$ such runs. The log-sum of lengths is $\Oh(n/2^j \cdot j)$ due to the concave log-function and Jensen's inequality.
    Summing over all the $j$, storing the literal tokens for each period and length takes $\Oh(\sum_{j = 0}^\infty n/2^j \cdot j) = \Oh(n)$ bits.
    Hence, $\Oh(n)$ is also an upper bound for $\absolute{\senc{R}}$.

    For computing $\senc{R}$, we proceed as follows.
    In a lookup table $L$, for every string $S \in [0\dd \sigma)^*$ with $\absolute{S} \in (4cP, 4cP + P]$, we store $L[\int(S)] = \senc{R'[0\dd \absolute{S} - 4cP)}$, where $R'[i]$ is defined exactly like $R[i]$ but for the string $S$ instead of $T$. Note that $4cP + P \leq (\log_\sigma n) / 4$, and hence $\int(S)$ is well-defined. The table can be computed with a naive algorithm in $\Oh(\sqrt{n} \cdot \poly(P)) \subset \Oh(n / \log n)$ time.
    Finally, we compute $\senc{R}$, starting with an empty sequence. For each $x \in [0\dd \ceil{n/P})$, we append $L[\int(T[xP\dd \min(n + 4cP, (x + 4c + 1) \cdot P)))]$, which is always defined due to the padding with $\texttt\textdollar$. Since $R$ contains no zero-run tokens, we do not have to separately handle zero-runs that span multiple length-$P$ chunks. 
    By appending in a word-wise manner, the time is $\Oh(n / P + \absolute{\senc{R}} / w) = \Oh(n / \log_\sigma n)$.
    It is also easy to see that the concatenation of the precomputed encodings is equivalent to $R$. Particularly, since in $R$ we truncate lengths of runs to $4cP$, it is clear that the substring used for accessing the table is sufficiently long.

    \textit{Defining filtered versions of $R$.}
    We created an encoding that consists of $n$ bits and, in principle, stores all the information needed to answer queries with $\floor{\tau / 3} < P$.
    However, we cannot always afford to spend $\Oh(n / \log n)$ query time to scan $\senc{R}$.
    Therefore, we create filtered sequences $R_j$ with $j \in [0\dd \ceil{\lg P})$, where $R_j$ will be responsible for answering queries with $\floor{\tau / 3} \in [2^j\dd 2^{j + 1})$. Note that for such queries, we have $\ell \geq \tau \geq 3\cdot2^j$.
    
    The initial $R_0$ is obtained by replacing all entries $R[i] = \senc{0}$ with $R_0[i] = 0$, and leaving all other entries unchanged.
    Then, each $R_j$ with $j > 0$ can be obtained from $R_{j - 1}$ by discarding all the runs of length less than $3 \cdot 2^j$.
    By our earlier observations, it is clear that there are only $\Oh(\sum_{j' = j}^\infty n/2^{j'})$ runs contributing to $R_j$, and encoding their lengths and periods takes $\absolute{\senc{R_j}} = \Oh(\sum_{j' = j}^\infty n/2^{j'} \cdot j') = \Oh(n / 2^j \cdot j)$ bits.

    \textit{Computing the filtered versions of $R$.}
    If we focus on a single entry of any of the arrays, by \cref{fct:overlap}, the number of relevant runs starting at any position is bounded by $\Oh(\lg P)$. Since the stored periods and lengths are bounded by $\Oh(P)$, it is clear that $\absolute{\senc{p_1\ell_1  p_2\ell_2 \dots p_h\ell_h}} = \Oh(\lg^2 P) = \Oh(\lg^2 \lg n)$, i.e., each $R_j$ is over alphabet $[0\dd 2^{\Oh(\lg^2 \lg n)}]$.
    Hence, we can use a precomputed lookup table $L_j$ that, given $R_{j - 1}[i]$, outputs $R_{j}[i]$ in constant time. 
    Clearly, such a table can be constructed in $\tilde\Oh(\sqrt{n})$ time for all the $j$.
    
    We use a single-state transducer $\mathcal T_j$ to compute $R_j$, where $\mathcal T_0$ merely replaces each symbol $\senc{0}$ with $0$.
    For $j > 0$, transducer $\mathcal T_j$ replaces each non-zero symbol by accessing the lookup table $L_j$.
    Either way, a transition takes constant time, and, by preprocessing all the transducers with \cref{lem:accelerate_single_stream} and $N = \Theta(\sqrt{n})$, we can compute $\senc{R_j}$ in $\Oh((\absolute{\senc{R_{j - 1}}} + \absolute{\senc{R_{j}}}) / \lg n) = \Oh(n / 2^j \cdot j / \lg n)$ time.
    Summing over all the $j$, this results in $\Oh(\sum_{j = 0}^{\ceil{\lg P} - 1}n / 2^j \cdot j / \lg n) = \Oh(n / \log n)$ time.

    Finally, we will also use transducers and a precomputed lookup table to answer queries. This table, when accessed with $\tau, \ell$ such that $\floor{\tau / 3} \in [0\dd P)$ and a sparse encoding $\senc{p_1\ell_1  p_2\ell_2 \dots p_h\ell_h}$ of length less than $(\lg n) / 2$ bits, returns true if and only if there is a pair $p_x\ell_x$ with $p_x \leq \floor{\tau / 3}$ and $\ell_x \geq \ell$. The table can be constructed in $\tilde\Oh(\sqrt{n})$ time.

    \subparagraph{Answering queries for small~\boldmath$\tau$.\unboldmath}

    Given a query $\tau, \ell$ with $\floor{\tau / 3} < P$, we can afford $\Oh(\frac{n \log \tau}{\tau\log n}) \supseteq \Oh(n / \log^2 n)$ query time.
    We first find $j$ such that $\floor{\tau/3} \in [2^j\dd 2^{j + 1})$ by trying all the values in $\Oh(\lg P) \subseteq \Oh(\lg n)$ time.    
    Then, we obtain $S$ by appropriately filtering $R_j$ with a single-state single-stream transducer. For each entry of $R_j$, the transducer answers in constant time whether a run of period $\leq \floor{\tau / 3}$ and length $\geq \ell$ is present using the precomputed lookup table (outputting either a one-bit or a zero-bit depending on the answer).
    We preprocess the transducer with \cref{lem:accelerate_single_stream} and parameter $N = \Theta(\sqrt{n})$.
    Hence, given the precomputed $\senc{R_j}$, we can obtain $\senc{S}$ in $\Oh(\absolute{\senc{R_j}} / \lg n) = \Oh(n / 2^j \cdot j / \log n) = \Oh(\frac{n \lg \tau}{\tau \lg n})$ time.
    \end{proof}

\section{Faster Synchronizing Set Queries}
\label{sec:appendixsynch}

\begin{lemma}\label{lem:mark_two_tau_runs}
    Consider $T \in [0\dd\sigma)^n$ and $\tau \in [1\dd n)$. Let $S$ and $E$ be length-$n$ bitmasks such that $S[a] = 1$ (resp. $E[b - 2\tau + 2] = 1$) if and only if there is a run $T[a\dd b]$ in $\RUNS_{2\tau, \floor{\tau/3}}(T)$. For $i \in [0\dd n)$, it holds $\sum_{j = 0}^i (S[j] - E[j]) \in \{ 0, 1\}$. Furthermore, if $i +2\tau \leq n$, then $\per(T[i\dd i + 2\tau)) \leq \floor{\tau / 3}$ if and only if $\sum_{j = 0}^i (S[j] - E[j]) = 1$.
    
\end{lemma}

\begin{proof}
    We say that a run $T[a\dd b]$ contributes to interval $[a\dd b - 2\tau + 2]$, meaning that it may cause one-bits in this interval.
    If there is another run $T[c\dd d]$ such that either $c \in [a\dd b - 2\tau + 2]$ or $d - 2\tau + 2 \in [a\dd b - 2\tau + 2]$, then the two runs overlap by at least $2\tau - 1$ symbols. However, since both runs have period at most $\floor{\tau / 3}$, this contradicts \cref{fct:overlap}.
    Therefore, any two distinct runs contribute to disjoint intervals.
    Hence, for any $i \in [0 \dd n)$, there is at most one run $T[a\dd b]$ in $\RUNS_{2\tau,\floor{\tau / 3}}(T)$ for which it holds $a \leq i$ and $b - 2\tau + 2 > i$, which implies $\sum_{j = 0}^i (S[j] - E[j]) \in \{ 0, 1\}$.
    If $\per(T[i\dd i + 2\tau)) \leq \floor{\tau / 3}$, then this fragment extends into a run $T[a\dd b]$ and it holds $i \in [a\dd b - 2\tau +2)$. By our earlier observation, it is clear that $\sum_{j = 0}^{a - 1} (S[j] - E[j]) = 0$, $S[a] = 1 \neq E[a]$, and $\forall_{j \in (a\dd b - 2\tau +2)}\;S[j] = E[j] = 0$.
    Conversely, if $\sum_{j = 0}^i (S[j] - E[j]) = 1$, then this can only be due to a run $T[a\dd b]$ with $a \leq i$ and $b - 2\tau + 2 > i$, which implies that $T[i\dd i + 2\tau)$ is a fragment of this run.
\end{proof}

\begin{lemma}\label{lem:truncateencoding}
    After an $\Oh(\sqrt{n})$-time preprocessing, given $\ell \in [1\dd n)$ and $\senc{V}$ with $V \in [0\dd n)^n$, we can compute $\senc{V[\ell\dd n) \cdot \texttt{\textup{0}}^{\ell}}$ in $\Oh(\ell + \absolute{\senc{V}} / \lg n)$ time.
\end{lemma}

\begin{proof}
    We use a single-state transducer with three input streams that read $V_1 := V \cdot \texttt1^{\ell}$, $V_2 := \texttt1^{\ell} \cdot \texttt0^n$, and $V_3 := \texttt0^n \cdot \texttt1^{\ell}$.
    It outputs $V'[0\dd n + \ell)$ with $V'[i] = 1$ if $V_2[i] = 1$, $V'[i] = 0$ if $V_3[i] = 1$, and otherwise $V'[i] = V_1[i]$.
    Hence $V' = \texttt 1^\ell \cdot V[\ell\dd n) \cdot \texttt 0^\ell$.
    We preprocess the transducer in $\Oh(\sqrt{n})$ time with \cref{lem:accelerate_multi_stream}.
    In the claimed query time, it is easy to produce $\senc{V_1}$, $\senc{V_2}$, and $\senc{V_3}$, and then run the transducer to obtain $\senc{V'}$.
    Finally, we discard the initial $2\ell$ bits of $\senc{V'}$, corresponding to $\senc{\texttt1^{\ell}}$.
\end{proof}

\restatefastersss*

\begin{proof}
    If $\tau > \sqrt{n} / \lg n$, then we can afford $\Oh(\frac{n \lg \tau}{\tau \lg n}) = \Oh(n / \tau)$ time to answer a query. In this case, we use \cref{thm:ss-explicit} to obtain the synchronizing set in explicit representation. Then, it is easy to encode its bitmask with \cref{lem:encode_list}, resulting in $\Oh(n / \tau)$ time. From now on, assume $\tau \leq \sqrt{n} / \lg n$, which implies $\tau \le \sqrt{n} = \Oh(\frac{n \lg \tau}{\tau \lg n})$.    
    During preprocessing, we construct the data structures from \cref{lem:faster_recomp_intermediate,lem:faster_recomp,lem:faster_runs}, where the data structure from \cref{lem:faster_runs} will be used with query parameters $\ell = \tau$ and $\ell = 2\tau$.

    Now we explain how to answer a query.
    We use the same technique as in \cref{thm:ss-explicit} to compute $k(\tau)$ in $1 + \Oh(\lg \frac{n}\tau) = \Oh(\frac{n \lg \tau}{\tau \lg n})$ time. 
    We have precomputed $\senc{\mathcal B_{k(\tau)}}$, i.e., the bitmask marking the boundary positions in $\B_{k(\tau)}$. 
    Apart from the synchronizing position contributed by $\tau$-runs, the synchronizing set consists only of (a subset of) positions~$i$ that satisfy $i + \tau \in \B_{k(\tau)}$, or equivalently $\mathcal B_{k(\tau)}[i + \tau] > 0$.
    We shift $\mathcal B_{k(\tau)}$ by removing its initial $\tau$ elements with \cref{lem:truncateencoding}. Afterwards, it holds $\mathcal B_{k(\tau)}[i] > 0$ if and only if $i + \tau$ is a boundary position in $\B_{k(\tau)}$.
    
    Next, we obtain the encodings of $S_\tau$ and $E_\tau$ (resp.\ $S_{2\tau}$ and $E_{2\tau}$) marking the starting and ending positions of all the $\tau$-runs (resp.\ runs in $\RUNS_{2\tau,\floor{\tau / 3}}(T)$), which takes $\Oh(\frac{n \lg \tau}{\tau \lg n})$ time using \cref{lem:faster_runs} with $\ell = \tau$ (resp.\ $\ell = 2\tau$).
    We remove the initial bit of $S_\tau$, and the initial $2\tau - 2$ bits of each of $E_\tau$ and $E_{2\tau}$ using \cref{lem:truncateencoding}. 

    Since each of $\mathcal B_{k(\tau)}$, $S_\tau$, $E_\tau$, and $E_{2\tau}$ is encoded in $\Oh(\frac{n}\tau \lg \tau)$ bits (see \cref{obs:sizeofBj,lem:faster_runs}), truncating them with \cref{lem:truncateencoding} takes $\Oh(\tau + \frac{n \lg \tau}{\tau \lg n}) = \Oh(\frac{n \lg \tau}{\tau \lg n})$ time. 
    From now on, $ \hat{\mathcal B}_{k(\tau)}$, $\hat{S}_\tau$, $\hat{E}_\tau$, and $\hat{E}_{2\tau}$ denote the truncated versions.    
    In particular, $\hat{S}_\tau[i] = 1$ if and only if a $\tau$-run starts at position $i + 1$, and $\hat{E}_\tau[i] = 1$ if and only if a $\tau$-run ends at position $i + 2\tau - 2$, i.e., $\hat{S}_\tau$ and $\hat{E}_\tau$ mark exactly the synchronizing positions contributed by $\tau$-runs.
    The bitmasks $S_{2\tau}$ and $\hat{E}_{2\tau}$ are exactly the ones used in \cref{lem:mark_two_tau_runs}.

    \newcommand{\tstate}[1]{{\setlength{\fboxsep}{2pt}\scriptstyle\textnormal{\fbox{\bfseries#1}}}}

    Finally, we describe a transducer that produces the bitmask of a $\tau$-synchronizing set. It reads $\hat{S}_\tau$, $\hat{E}_\tau$, $S_{2\tau}$, $\hat{E}_{2\tau}$, and $\hat{\mathcal B}_{k(\tau)}$ as input streams.
    The states are $\{\,\tstate0,\tstate1\,\}$ with initial state $0$, where the intended functionality is as follows.
    When reading the input at position~$i$, we enter or already are in state $\tstate 0$ if and only if $\per(T[i\dd i + 2\tau)) > \floor{\tau /3}$.   
    Conversely, we enter or already are in state $\tstate1$ if and only if $\per(T[i\dd i + 2\tau)) \leq \floor{\tau /3}$.
    Hence, regardless of the current state, whenever $S_{2\tau}[i] = 1$, we output $M[i] = 0$ and transition to state $1$.
    Whenever $\hat{E}_{2\tau}[i] = 1$ (this also implies $\hat{E}_\tau[i] = 1$, i.e., the current run has to contribute this position to the synchronizing set), we output $M[i] = 1$ and transition to state $\tstate0$.
    For all so far undefined transitions, we do not change state. The functionality of the states is as intended due to \cref{lem:mark_two_tau_runs}.

    Now we consider the remaining transitions, i.e., the ones with $S_{2\tau}[i] = \hat{E}_{2\tau}[i] = 0$.
    Whenever $\hat{S}_\tau[i] = 1$ or $\hat{E}_\tau[i] = 1$, we output $M[i] = 1$ and stay in the current state (adding the synchronizing positions contributed by $\tau$-runs).
    For all so far undefined transitions from state $\tstate1$, we output $M[i] = 0$ and stay in state $\tstate1$. For all so far undefined transitions from state $\tstate0$, we output $M[i] = \min(1, \hat{\mathcal B}_{k(\tau)}[i])$ and stay in state $0$.    

    The transducer merely emulates the bitmask operations performed in the proof of \cref{thm:ss-bitmask}, and hence it indeed computes the bitmask of a synchronizing set with the claimed properties. 
    We preprocess the transducer with \cref{lem:accelerate_multi_stream} and parameter $N = \Theta(\sqrt{n}) \subset \Oh(n / (\tau \lg n))$. Then, the time for running the transducer is     
    $\Oh((\absolute{\senc{M}} + \absolute{\senc{\hat S_\tau}} + \absolute{\senc{\hat E_\tau}} + \absolute{\senc{S_{2\tau}}} + \absolute{\senc{\hat E_{2\tau}}} + \absolute{\senc{\hat{\mathcal{B}}_{k(\tau)}}}) / \log n)$, which is bounded by $\Oh(\frac{n \lg \tau}{\tau \lg n})$.    
\end{proof}

\restatesparsewithsupport*

\begin{proof}
    We obtain the set in sparse encoding $\senc{M}$ using \cref{thm:sss_sparse}.
    Then, we use \cref{lem:sparse_select,lem:sparse_rank} with preprocessing parameter $N = \Theta(\sqrt{n})$ to obtain the support data structures.
    Recall that the sparse encoding is of size $\Oh(\frac{n \lg \tau}{\tau \lg n})$ bits.
    Hence, for select queries, \cref{lem:sparse_select} immediately results in the claimed complexities.
    For rank queries, we use space parameter $m = \absolute{\senc{M}} / \lg N + \frac{n \lg \tau}{\tau \lg n} = \Theta(\frac{n \lg \tau}{\tau \lg n})$.
    This results in the correct construction time and space. The query time is $\Oh(\lg\frac{\lg (n/m)}{\lg \lg n})= \Oh(\lg\frac{\lg (\tau \cdot \lg n / \lg \tau))}{\lg \lg n})$.
    If $\tau = \Oh(\polylog(n))$ then the time is constant, and otherwise it is bounded by $\Oh(\lg\frac{\lg \tau}{\lg \lg n})$.
\end{proof}

\bibliographystyle{plainurl}
\bibliography{main}

\end{document}